%% file: main.tex
\documentclass[11pt]{article}
\usepackage{amsmath,amsfonts,amssymb,amsthm,  mathrsfs,mathtools}
\usepackage[margin=1in]{geometry}
\usepackage{hyperref}
\hypersetup{
  colorlinks,
  citecolor=violet,
  linkcolor=blue,
  urlcolor=blue}
\usepackage{dsfont}
\usepackage{algorithm}
\usepackage{algorithmic}
\usepackage{subcaption}
\usepackage{xfrac}
\usepackage{enumitem}
\usepackage{graphicx}
\usepackage{xcolor}
\usepackage[square,numbers]{natbib}
\usepackage{algorithm}
\usepackage{thm-restate}

\usepackage{authblk}
\usepackage{booktabs}
\usepackage{multirow}

\input{macros}

\begin{document}

\title{How to Strategize Human Content Creation in the Era of GenAI?\thanks{This work is supported by the AI2050 program at Schmidt Sciences (Grant G-24-66104), Army Research
Office Award W911NF-23-1-0030, ONR Award N00014-23-1-2802 and NSF Award CCF-2303372.}}
\date{}  

\newcommand*{\affaddr}[1]{#1}
\newcommand*{\affmark}[1][*]{\textsuperscript{\normalfont#1}}
\newcommand*\samethanks[1][\value{footnote}]{\footnotemark[#1]}

\author{
  Seyed A. Esmaeili\affmark[1]\thanks{By contribution}~, Kevin Lim \affmark[1]\samethanks[2]~,Kshipra Bhawalkar\affmark[2]\thanks{Ordered Alphabetically}~, Zhe Feng\affmark[2]\samethanks, Di Wang\affmark[2]\samethanks,  Haifeng Xu\affmark[1,2]\samethanks[2]\\
  \affaddr{\affmark[1] University of Chicago}\\
  \affaddr{\affmark[2] Google Research}\\
  \texttt{\{esmaeili,kevinhlim,haifengxu\}@uchicago.edu}, \texttt{\{kshipra,zhef,wadi\}@google.com}\\
}
\date{}  
\maketitle

\begin{abstract}
Generative AI (GenAI) will have significant impact on content creation platforms. In this paper, we study the dynamic competition between a GenAI and a human contributor. Unlike the human, the GenAI's content only improves  when more contents are created by the human over time; however,   GenAI has the advantage of generating content at a lower cost. We study the algorithmic problem in this dynamic competition model about how the human contributor can maximize her utility when competing against the GenAI for content generation over a set of topics. In time-sensitive content domains (e.g., news or pop music creation) where contents' value diminishes over time, we show that there is no  polynomial time algorithm for finding the human's optimal (dynamic) strategy, unless the randomized exponential time hypothesis is false. Fortunately, we are able to design a polynomial time algorithm that naturally cycles between myopically optimizing over a short time window and pausing and provably guarantees an approximation ratio of $\frac{1}{2}$. We then turn to time-insensitive content domains where contents do not lose their value (e.g., contents on history facts). Interestingly,  we show that this setting permits a polynomial time algorithm that maximizes the human's utility in the long run. Finally, we conduct simulations that demonstrate the advantage of our algorithms in comparison to a collection of baselines.  
\end{abstract}

\input{introduction}
\input{relatedWork} 
\input{model}
\input{discounting_factor}
\input{long_horizon}
\input{experiments}
\input{conclusion}
\input{impact_statement}

\bibliography{refs}
\bibliographystyle{unsrtnat}

\input{appendix}

\end{document}

%% file: macros.tex
\DeclareMathOperator*{\argmax}{argmax}
\DeclareMathOperator*{\argmin}{argmin}

\theoremstyle{plain}
\newtheorem{theorem}{Theorem}[section]

\newtheorem{lemma}[theorem]{Lemma}
\newtheorem{fact}[theorem]{Fact}
\newtheorem{observation}[theorem]{Observation}
\newtheorem{claim}[theorem]{Claim}

\newtheorem*{theorem*}{Theorem}
\newtheorem*{lemma*}{Lemma}
\newtheorem{claim*}{Claim}

\theoremstyle{definition}
\newtheorem*{definition*}{Definition}
\newtheorem{definition}[theorem]{Definition}

\theoremstyle{remark}

\newtheorem*{remark*}{Remark}
\newtheorem*{property*}{Property}

\newtheorem*{problem*}{Problem formulation}

\DeclareMathOperator{\E}{\mathbb{E}}
\DeclareMathOperator{\Prob}{\mathbb{P}}

\newcommand{\floor}[1]{\left\lfloor #1 \right\rfloor}
\newcommand{\ceil}[1]{\left\lceil #1 \right\rceil}

\newcommand{\pullsh}[1]{\mathnormal{N_{#1}}}

\DeclareMathOperator{\uh}{\mathnormal{u}} 

\newcommand{\uairange}[2]{\mathnormal{u_{#1 : #2}^{\texttt{AI}}}}

\DeclareMathOperator{\meani}{\mathnormal{\mu_i}}

\DeclareMathOperator{\arms}{\mathnormal{[k]}}

\DeclareMathOperator{\armth}{\mathnormal{b_t}} 
 
\DeclareMathOperator{\armtai}{\mathnormal{a_t}}

\DeclareMathOperator{\probc}{\mathnormal{\sigma}}

\DeclareMathOperator{\mgm}{\mathnormal{a^*}} 
\DeclareMathOperator{\mgmt}{\mathnormal{a^*_t}}

\DeclareMathOperator{\bvec}{\mathnormal{p}} 
\DeclareMathOperator{\avec}{\mathnormal{q}}

\DeclareMathOperator{\muvec}{\mathnormal{\boldsymbol{\mu}}} 
\DeclareMathOperator{\gammavec}{\mathnormal{\boldsymbol{\gamma}}} 
\DeclareMathOperator{\cvec}{\mathnormal{\boldsymbol{c}}} 
\DeclareMathOperator{\qvec}{\mathnormal{\boldsymbol{q}}} 
\DeclareMathOperator{\hvec}{\mathnormal{\boldsymbol{h}}}

\newcommand{\genmu}[1]{\mathnormal{\Tilde{\mu}_{#1}}}
\newcommand{\muexit}[1]{\mathnormal{\mu_{#1,\text{exit}}}}

\newcommand{\nexit}[1]{\mathnormal{n_{#1,\text{exit}}}}

\newcommand{\agmgm}{\mathnormal{\mathcal{P}_{\mgm}}}
\newcommand{\hold}{\mathcal{H}_{\mgm}}

\newcommand{\maxgenmu}[1]{\genmu{a_{#1}^*}}

\DeclareMathOperator{\genmul}{\mathnormal{\Tilde{\mu}_{l}}}
\DeclareMathOperator{\genmuu}{\mathnormal{\Tilde{\mu}_{u}}}

\newcommand{\ngain}[2]{n^{(#1,#2)}}

\newcommand{\uo}{\mathnormal{u_0}}

\newcommand{\reth}{\texttt{rETH}}
\newcommand{\maxreward}{\texttt{MAXREWARD}}
\DeclareMathOperator{\bn}{\mathnormal{N}} 
\DeclareMathOperator{\tauwindow}{\mathnormal{\tau}}

\DeclareMathOperator{\astar}{\mathnormal{\mathbf{a}^*}}
\DeclareMathOperator{\bstar}{\mathnormal{\mathbf{b}^*}}
\DeclareMathOperator{\aprime}{\mathnormal{\mathbf{a}\prime}}
\DeclareMathOperator{\bprime}{\mathnormal{\mathbf{b}\prime}}
\DeclareMathOperator{\abold}{\mathnormal{\mathbf{a}}}
\DeclareMathOperator{\bbold}{\mathnormal{\mathbf{b}}}
\DeclareMathOperator{\bobv}{\mathnormal{\mathbf{b}^{\text{obv}}}}

\newcommand{\bmyopic}{\mathnormal{\mathbf{b}_{\myopic}}}

\DeclareMathOperator{\smgm}{\mathnormal{\mathbf{a}^*}}
\DeclareMathOperator{\sobv}{\mathnormal{\mathbf{b}^{\text{obv}}}}

\newcommand{\myopicopttau}{\textbf{Myopic-OPT($\tauwindow$)}}
\newcommand{\myopicoptone}{\textbf{Myopic-OPT($1$)}}

\newcommand{\myopic}{\textbf{Myopic}}

\newcommand{\poly}{\mathnormal{\texttt{poly}}}

\newcommand{\thetal}{\mathnormal{\theta_l}}
\newcommand{\thetau}{\mathnormal{\theta_u}}

\newcommand{\gmax}{\mathnormal{\gamma_{\max}}}
\newcommand{\dmax}{\mathnormal{D_{\max}}}
\newcommand{\ghat}{\mathnormal{\hat{\gamma}}}
 
\newcommand{\gstari}{\mathnormal{\gamma^*_i}} 

\newcommand{\neps}{\mathnormal{n_{\epsilon}}}

\newcommand{\tless}{\mathcal{T}_{\text{less}}}
\newcommand{\tmore}{\mathcal{T}_{\text{more}}}

\newcommand{\OPT}{\mathnormal{OPT}}

%% file: introduction.tex
\section{Introduction} 
The emergence of Generative Artificial Intelligence (GenAI) has ushered in a transformative era across numerous practical domains, including text~\cite{chatgpt}, image~\cite{imagen}, and video generation~\cite{ho2022imagen, bartal2024lumiere, sora}. As AI-generated content proliferates on online  platforms, ranging from media-sharing (e.g., Instagram \cite{meta} and TikTok \cite{tiktok}) to news outlets, the implications for human content creators are profound. 

The rise of GenAI presents both challenges and opportunities for human content creators. While GenAI poses a competitive threat due to its capacity for rapid and cost-effective content generation, it also relies on human-created data for training and refinement. This interdependence creates a strategic interplay where human creators may adapt their content strategies to differentiate themselves from AI-generated outputs --- a theme that this work aims to address.  

A key concern surrounding GenAI is its potential to displace human creators due to its lower production costs and its capability to generate useful content as well as human creators. However, the imperative for \emph{timely} and accurate content in online platforms necessitates a continuous influx of human-generated material to serve as training data for GenAI models, for example, media contents on hot topics and recent news comments. This reliance on human input suggests a cyclical relationship where the absence of human content could degrade the quality of AI-generated outputs, leading to user dissatisfaction and a renewed demand for human creativity. In this scenario, human creators may experience a resurgence in value as their content garners increased attention and engagement from users, which creates an intrinsic difficulty to analyze the \emph{dynamics} between human and GenAI competition. To capture this \emph{time-sensitive} nature of content generation, we discount the value of human contents when we model the capability of GenAI as a function of human created contents.  Naturally, the older the contents, the more its value is discounted. 

Notably, there are also situations where the value of contents is   \emph{time-insensitive}, e.g., historical facts, scientific knowledge   and professional skills training. For these domains, through our model we uncover the threat that GenAI will replace human content creation can be real. We model such time-insensitive scenarios by assuming that the quality of human-generated contents is not discounted over time. For each topic, there is a moment that human content creation will eventually stop. To optimize utility in this model, one important consideration of the human creator is to delay their time to be driven out by GenAI and meanwhile maximize their accumulated rewards before GenAI's content creation becomes as competitive.



\vspace{-0.1cm}
\subsection{Our Contributions}
Our main technical contributions are threefold: 
\textbf{(1)} In Section \ref{sec:model}, we introduce the first model to study dynamic competition between a human agent and a GenAI agent for content creation across different topics in an online platform. Our model introduces key elements that capture  advantages and disadvantages of both  human and GenAI agent. In addition, the model enables us to associate with each topic a measure of ``information value'' at current time by discounting the contents  according to how recently they were created. Different topics have different discounting factors: e.g., time-sensitive contents  are more discounted than insensitive ones. Our goal is to design algorithms/strategies that strategize the human agent's choice of content topics in order to maximize the human's utility.  

\textbf{(2)} In Section \ref{sec:discounted}, we study the setting where topics are time-sensitive and therefore the (human) contents used for GenAI training for all topics are discounted. We start by proving that the problem is computationally hard. Specifically, we show that finding an optimal strategy for the human cannot be done in polynomial time unless the randomized exponential time hypothesis is false. Perhaps surprisingly, this intractability remains true even in simple special settings, hence illustrating the fundamental barriers for the human to strategize in such highly dynamic environments. We then introduce a novel polynomial time algorithm that alternates between running a short-term optimal strategy   and pausing (not generating any content) and show that it gurantees  an approximation ratio close to $\frac{1}{2}$. We further prove that it is necessary for any algorithm in time-sensitive setting to pause for a portion of the time horizon as otherwise the algorithm will provably obtain a vacuous approximation of the optimal utility in some situations.

\textbf{(3)} In Section \ref{sec:long_horizon}, we consider the setting where topics are time-insensitive and therefore the contents are not discounted (i.e., their information value to GenAI will not reduce over time). We pay special interest to long-term interactions due to a desire of understanding the human's long term utility. For time-insensitive topics (e.g., facts or scientific knowledge), a corollary from our model is that the human player will eventually find it not profitable to create any new contents due to GenAI's increasing capability, hence the human will ``exit'' the platform at some round. However, how to maximize total utility for the human before exiting the platform is a highly non-trivial algorithmic question. We start with a simple baseline algorithm that runs in exponential time. Through careful analysis about the problem structure,  we are able to improve the algorithm, leading to a polynomial time algorithm that maximizes the human utility. This setting provides an interesting contrast to the discounted setting where optimizing utility is computationally hard.    

All of our algorithms from Sections \ref{sec:discounted} and \ref{sec:long_horizon} demonstrate their superiority over a collection of baselines as empirically tested in Section \ref{sec:experiments}.

A few remarks regarding this work's contribution are worth mentioning. As one of the first studies of such dynamic human-vs-GenAI content creation competition, our model is  simple and basic. It captures the essential aspects of the dynamic competition but admittedly does not capture its whole landscape. On one hand, as we show even this simple setup is already highly non-trivial for algorithm development. On the other hand, we do see various possibilities to extend our basic model formulation  which cannot be all covered in a single work (see  detailed discussions in Appendix \ref{sec:conclusion}). However, we believe our algorithmic results offer useful insights for future studies of richer and concrete problem settings. Finally, we present intuitions for all technical results in the main body but defer formal proofs to the appendix due to  space limits. 



%% file: relatedWork.tex
\subsection{Additional Related Work}
There has been a recent surge of interest in studying online content economy and competition among different content creators \cite{ben2017shapley,ben2018game,yao2023rethinking,yao2023bad,zhu2023online,hu2023incentivizing,jagadeesan2023supply,hron2022modeling,yao2024human}. More related to us are the works that aim to understand content creators' behaviors, for example, how creators will specialize at the equilibrium \cite{jagadeesan2023supply}, how creators' strategic behavior affects social welfare \cite{yao2023bad}, and how to design optimization method for long-term welfare considering content creators' strategic behaviors \cite{ben2017shapley,ben2018game,yao2023rethinking,zhu2023online,hu2023incentivizing,immorlica2024clickbait,mladenov2020optimizing}. Our work differ from these studies in a few key aspects. Most importantly, the previous works study the competition among only human creators whereas our model is set to study a more macro level question and features the competition between humanity collectively represented as a single Human player and technology represented as a single GenAI player.   


Recently, a few papers have studied the interaction between GenAI and human content creators. In particular, \cite{taitler2024braess} study the social impact from GenAI on content creators in question answering websites such as Stack Overflow. Their main focus is on designing algorithms that maximize revenue for the GenAI wherease we consider the optimization problem from the perspective of the human and accordingly our objectives and focus is very different. Moreover, \cite{raghavan2024competition} study content diversity/homogeneity when content creators use GenAI as a tool showing that GenAI usuage could lead to less diverse content. Additionally, \citet{yao2024human} study the competition between human and GenAI. However,  their model is at a more micro-level with many human creators who either compete with a single GenAI or use GenAI as an additional content creation tool. 



%% file: model.tex
\section{A Model of  Dynamic Human-vs-GenAI  Content  Competition}\label{sec:model}
\textbf{Basic Setup. }  Motivated by content creation competition to attract Internet user traffic on recommender systems (e.g., Instagram, Tiktok, Youtube, etc.), in this section we introduce a basic model to capture the dynamic competition between human content creators and GenAI content generation. 
There are $k$ different topics (e.g.,   sports, politics, or nature), denoted by $\arms = \{ 1, \cdots, k \}$, for which   humans and GenAI can create contents. For convenience, we shall refer to each topic as an \emph{arm}. Additionally, and naturally, there is always an ``opt out'' option for the human to exit the game and receive $0$ utility. In this paper we focus on the dynamic two-agent competition between a single human agent and a single GenAI agent. This model  can be interpreted as the competition between all humans collaboratively as one team and all GenAIs collectively as another team. The goal of our work   is to study the human team's optimal competition strategy against GenAIs. While a more complex model could consider multiple self-interested humans and GenAIs, we start our study from the two-agent setting because: (1) even such two-agent competition is already highly non-trivial hence worth a thorough investigation; (2) human creators'  distinct content creation behavior used to  be primarily caused by their different creation costs \cite{yao2023rethinking}, but this factor  become negligible when facing GenAIs with order-of-magnitude lower creation cost. Thus this assimilates human creators' behaviors due to facing the same strong GenAI competitor, effectively making humans as if in a team.
The human agent's content creation capability for each arm/topic $i \in \arms$ is described by two values $ c_i, \meani$  where $c_i$ is normalized to be within $[0, 1]$ is the human's cost of creating a topic-$i$ content and $\meani$ captures user satisfaction or \emph{reward} for a  human-created topic-$i$ content.\footnote{Internet users'  reward  for a  content  may have noise in general but this will be modeled later when describing the users' (random) choice between human-created and GenAI-created contents.  } 


 \vspace{2mm}
\textbf{Modeling the Capability of GenAIs. } We consider repeated content creation over rounds $1,2,\dots,T$. Let arm $\armth \in \arms$ denote the topic choice by the human agent at  round $t$, whereas the GenAI chooses an arm $\armtai \in \arms$. As mentioned above, the reward resulting from the human's pull of arm $\armth$ is $\mu_{\armth}$. 
However, the reward $\tilde{\mu}_{\armtai}$ resulting from the GenAI's 
pull requires more careful modeling and will naturally depend on the capability of the GenAI technology. Motivated by extensive recent studies of the GenAI's capability as a function of the number of training data (also widely known as  \emph{scaling laws} \cite{kaplan2020scaling,gao2023scaling}),
we assume $\tilde{\mu}_{\armtai}$ is a function of the total (possibly discounted) number of times that arm $\armtai$ has been pulled by the human. Notably, we assume GenAI's pulling of an arm does not improve the GenAI's capability on this arm; this is consistent with the state-of-the-art GenAI technologies which are all trained on high-quality human-created data but not synthetic  GenAI  contents. 
Formally, let 
\begin{equation}
     \pullsh{i}(t) =  \sum_{s = 1}^{t} (\gamma_i)^{t+1-s} \times \mathbb{I} (b_s = i) 
\end{equation}
denote  the total discounted number of times that the human pulled arm $i \in \arms$ at round $t$. Here the discounting factor $\gamma_i \in [0, 1]$ is a parameter that captures the time-sensitivity of a topic $i$ and can generally be different for different topics. For instance, news topics are time-sensitive and will have small $\gamma_i$ because earlier data are less useful for generating today's  news content\footnote{Currently,  one potential recipe for such time-sensitive topics with   low $\gamma_i$  could be using retrieval-augmented generation (RAG) \cite{lewis2020retrieval,siriwardhana2023improving} that feeds the news articles as prompts/contexts to generate. This improvement could still be captured by our model using an increased $\gamma_i$ value.} 
  whereas factual knowledge (e.g., history events, linear algebra, etc.) is less time-sensitive and will have large $\gamma_i$. 
We will pay special interest to the cases where all topics have a discounted number of pulls $\gamma_i<1, \ \forall i \in \arms$ in Section \ref{sec:discounted} and where the value of contents is not discounted $\gamma_i=1, \ \forall i \in \arms$ in Section \ref{sec:long_horizon}.  

Hence for any arm $i \in \arms$, we model the reward $\tilde{\mu}_{i}$ of the GenAI's pull of arm $i$ as $\tilde{\mu}_{i}(t):=\mu_i -g_{i}(\pullsh{i}(t-1))  $ where $g_{i}(.)$ is  the quality \emph{gap} between human-created contents  and GenAI-created ones for arm $i$ and is a function of the total number of human-created contents. We also refer to $g_{i}$ as the \emph{shrinkage} function. Besides naturally assuming $g_i$ is positive and decreasing (i.e., more human-created contents, less gap), we do not assume any concrete form for $g_i$. We allow each arm $i \in \arms$ to have its own shrinkage function $g_i(.)$, capturing the fact that the GenAI's content generation capabilities  may differ across topics. Note that $\Tilde{\mu}_i$ depends (only) on human's historical arm pulls. For convenience, we refer to $\Tilde{\mu}_i(t) = \mu_i -g_{i}(\pullsh{i}(t-1))  $ as the \emph{generative mean} of arm $i$ at round $t$. We assume the cost for GenAI to create a  topic-$i$ content is smaller than human's cost. Without loss of generality, we \emph{normalize} the GenAI's content creation cost to be $0$ (i.e., human's cost $c_i\in [0, 1]$ should be interpreted as the relative cost). Finally, we make a minor technical assumption that the rewards of all (human or GenAI) created contents are bounded in $[\mu_l, \mu_u]$. I.e., $\forall i \in [k]: \mu_u \geq   \mu_i \geq  \mu_i - g_i(T) \ge \mu_l \ge 0$. 

\begin{figure}
    \centering
    \includegraphics[width=.5\linewidth]{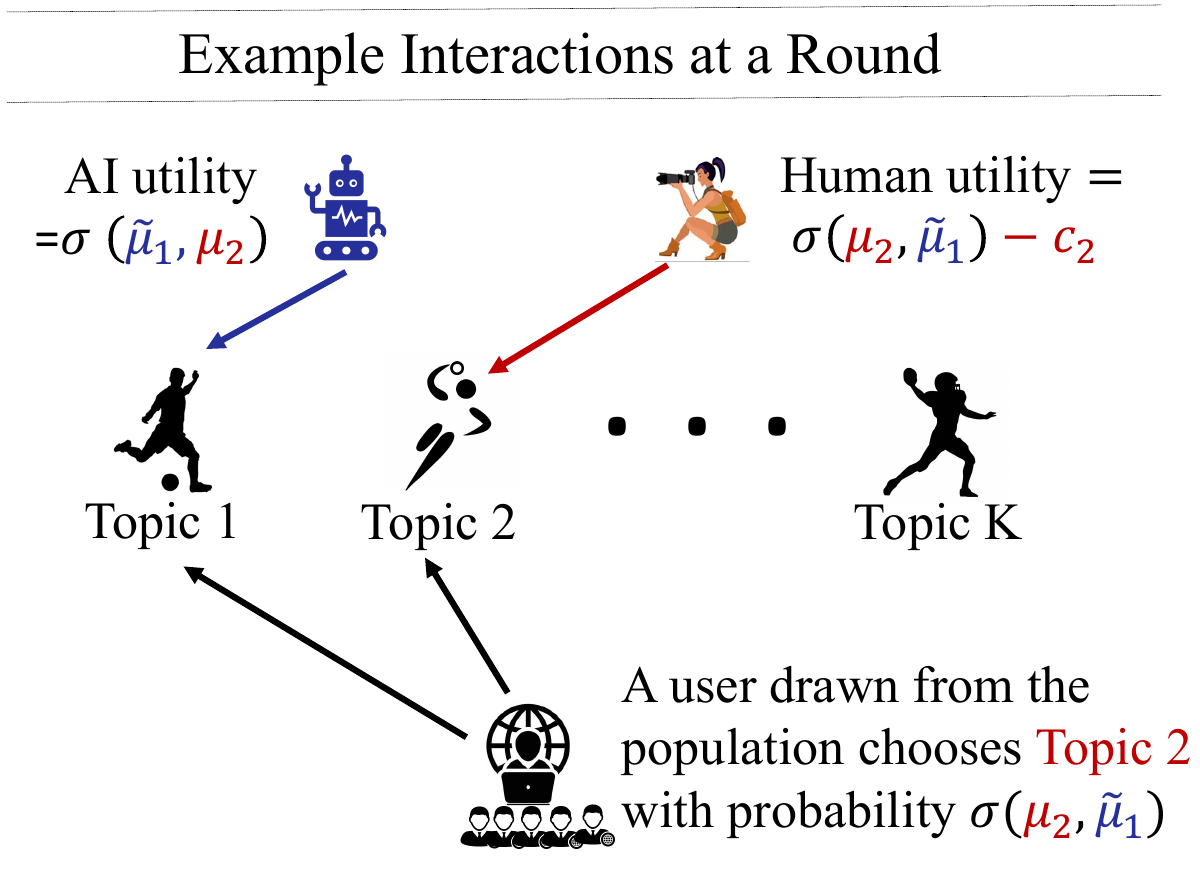}
    \vspace{-4.75mm}
    \caption{Example interaction between the human, GenAI, and Internet user at a given round. }
    \label{fig:interaction}
\end{figure}
\textbf{Competitions and Utilities. } We are now ready to model the dynamic Human-vs-GenAI competition. Figure \ref{fig:interaction} shows a visualization of the interaction. In each round $t$ an Internet user arrives and faces two contents, one created by human   with user  reward $r_{\armth}$ and another created by GenAI with reward $\Tilde{r}_{\armtai}$.  Following the standard choice model, we assume the user chooses one of the arms $b_t,a_t$ randomly with probabilities specified by a \emph{link function} $\probc$. Importantly, this means that the user is choosing between only the current round's contents. This is realistic since frequent recommender system users often only consume fresh contents as older contents have already been viewed (this is also why these systems often recommend fresh contents anyway). Let $ i \succ j $ denote the event that the user selects arm $i$ over $j$. The probability of this event is given by the link function $\probc(.,.)$ as follows
\begin{align}
\Prob(i \succ j)   = \probc(\mu_i,\mu_j) \in [0, 1]
\end{align}

Naturally, $\Prob(j \succ i)   = 1 -  \probc(\mu_i,\mu_j) = \probc(\mu_j,\mu_i)$. Besides the standard  (and natural) assumptions that $\probc(\cdot, \cdot)$ is monotonically increasing in the first variable and decreasing in the second, our results do not rely on any concrete form of the link function $\probc$. We also allow the human to possibly not select an arm in any round $t$, and denote this action by $\emptyset$. In that case, the human incurs no cost but   GenAI's content will be chosen  with probability $1$. That is, $\probc(\emptyset, \mu) = 0 = 1 - \probc(\mu, \emptyset)$ for any $\mu \in [\mu_l, \mu_u]$.  

The objective of both the human and GenAI is to maximize the user traffic of their contents, defined as the expected number of user pulls. Recall that the GenAI's generation costs have been normalized to $0$. Hence the competition leads to the following accumulated utilities for each agent.
\begin{align} 
     \text{{\bf Human Agent's Accumulated Utility}:} & \nonumber \\
     \uh_T(\mathbf{b};\mathbf{a} )   =   \sum_{t=1}^T \big[  \probc (\mu_{\armth}, \Tilde{\mu}_{\armtai}(t))    -    c_{\armth} \big]   \label{eq:human-U}   
     \\ 
     \nonumber
     \text{{\bf GenAI Agent's Accumulated Utility}:}    & \nonumber \\   u^{\texttt{AI}}_T(\mathbf{a};\mathbf{b}) = \sum_{t=1}^T  \probc (\Tilde{\mu}_{\armtai}(t), \mu_{\armth})      \label{eq:ai-U} 
\end{align}

where $\mathbf{b} = (b_1, \cdots, b_T)$ is the realized sequence of the human's actions,  similarly $\mathbf{a} = (a_1, \cdots, a_T)$ is the realized sequence of the GenAI's actions. At times, we may need to refer to the agents' accumulated utilities between rounds $t_1$ and  $t_2$ ($t_1 \leq t_2$, both inclusive). In this case, we use $\uh_{t_1:t_2}$ and $u^{\texttt{AI}}_{t_1:t_2} $ to denote these utilities which are defined similarly as in Equation \eqref{eq:human-U}  and \eqref{eq:ai-U} but with summation from $t_1$ to $t_2$. Note that if the human or GenAI randomize their actions then the utilities are the expectations of \eqref{eq:human-U} and \eqref{eq:ai-U} for the human and GenAI, respectively.  

The accumulated utility functions in Equation \eqref{eq:human-U} and \eqref{eq:ai-U} induce a \emph{dynamic} competition between the human and the GenAI. The two agents' interactions hinge on GenAI's  generative mean   $\Tilde{\mu}_{\armtai}(t) = \mu_{\armtai} -g_{\armtai}(\pullsh{\armtai}(t-1))  $ of   arm $\armtai$, which depends on  human's past actions. 


We make a few remarks about the link function $\probc$. User's probabilistic choice is motivated by reward noise, which is standard in choice models \cite{bradley1952rank,davidson1976bibliography}. That is, both $\mu_{\armth}, \Tilde{\mu}_{\armtai}(t)$ should be viewed as expected rewards of the contents whereas the realized true reward should be $\mu_{\armth} + \epsilon_1, \Tilde{\mu}_{\armtai}(t) +\epsilon_2$. The randomness in user choice comes from the random noise term $\epsilon_1, \epsilon_2$. For instance, if they are logit distributions, we arrive at the standard Bradley–Terry model with $\probc(\mu_{\armth}, \Tilde{\mu}_{\armtai}(t)) = \frac{e^{\mu_{\armth}}}{e^{\mu_{\armth}} + e^{\Tilde{\mu}_{\armtai}(t)}}$. If they are i.i.d Gaussian noise with zero mean and $\sigma^2$ variance, then  $\probc(\mu_{\armth}, \Tilde{\mu}_{\armtai}(t))$ equals the probability that Gaussian $\mathcal{N}(\mu_{\armth} -  \Tilde{\mu}_{\armtai}(t),2\sigma^2)$ is positive \cite{yue2012k}.

\vspace{2mm}
\textbf{Final Problem Formulation and Preliminary Model Analysis.} Standing in the human's shoes, we study the human's optimal dynamic competition with GenAI. One last key modeling piece is the content creation strategy of the GenAI. In the present work, we consider the competition with GenAI that has myopic content creation strategies as formalized in the following assumption. 
\begin{restatable}{assumption}{mgmmaximizesprob}\label{prop:mgm_maximizes_prob}
At a given round $t$ the GenAI selects the arm with the current highest generative mean $\mgmt = \max_{i \in [k]} \genmu{i}(t)$ to maximize  the probability of being chosen by the user for this round.\footnote{The fact that the arm with the highest  generative mean at round $t$ maximizes GenAI's probability of being chosen by the user at $t$ follows from the monotonicity   of link function $\probc$. See Appendix \ref{append:model-proofs} for details. }   
\end{restatable} 
 Let vector $\astar \in [k]^T$  denote the sequence of GenAI's arm choices. Notably, $\astar$ is not fixed and  is a function of the human pulls $\bbold$. To emphasize this dependence, we also write $\astar (\bbold)$. Hence, this gives rise to our final dynamic optimization problem for the human: 
 \begin{align} 
    & \textbf{Human's Dynamic Optimization Problem:} \nonumber \\
    & \quad \quad \quad  \max_{\bbold \in [k]^T} \quad \uh_T(\mathbf{b}; \astar (\bbold)) \label{fig:prob-formulation} \\
    & \quad \quad \quad \text{where $\mgmt= \max_{i \in [k]} \ \genmu{i}(t)$} \nonumber
\end{align}

In Appendix \ref{app:basic_stg}, we also present an equivalent formulation of the above dynamic optimization problem as searching for the longest path on certain ``state transition graph''. This formulation will be useful in constructing some technical arguments. We conclude this  section with some preliminary analysis about the above optimization problem. We start by illustrating the behavioral assumption of the GenAI, as detailed in Assumption \ref{prop:mgm_maximizes_prob}. Recall that our goal is to optimize human's utility in the competition with GenAI. The following result shows that optimizing against the GenAI  strategy  in Assumption \ref{prop:mgm_maximizes_prob} is equivalent to robust optimization of human's worst-case utility. 
\begin{restatable}{proposition}{leastutil}\label{theorem:utility_lower_bound}
For any $\bbold$, we have $\uh_T(\bbold ; \astar (\bbold)) \leq   \uh_T(\bbold ; \abold) $ for any $ \abold$.
\end{restatable}
Proposition \ref{theorem:utility_lower_bound} implies that by assuming GenAI strategy $\astar$, the human can establish a guarantee on her worst-case utility  even if GenAI does not follow the strategy at the end. This offers a justification of GenAI's strategy assumption through the lens of robust optimization.

Next we offer an alternative view of the GenAI's strategy through a game-theoretic lens. Notably, if the GenAI is forward-looking, its strategy generally could have been much more complex  since it may intentionally pull sub-optimal arms at times in order to encourage human content creation on certain arms so that it would benefit from this creation in the future. Such a strategic consideration will lead to a general dynamic game between the GenAI and the human, and finding an equilibrium in such dynamic games is notoriously hard \cite{koller1992complexity,letchford2010computing,hansen2007finding}. However, our following result shows that, if the human's strategy is independent of the GenAI's choices, the GenAI strategy in Assumption \ref{prop:mgm_maximizes_prob} is actually an optimal GenAI strategy over the entire horizon. Formally, we say $\bbold$ is \emph{oblivious}, if each $b_t$ is a function only of its own history $b_1, \cdots, b_{t-1}$.\footnote{This can be equivalently viewed as the human committing to a long-horizon strategy $\bbold$ at the beginning of the competition, after which the GenAI best responds.}  
 
\begin{restatable}{proposition}{myopicoptobv}\label{theorem:ai_best_resp_obv}
If the human uses an oblivious strategy $\bbold$ then the GenAI myopic strategy $\astar(\bbold)$ as desribed in Assumption \ref{prop:mgm_maximizes_prob} is also an optimal forward-looking strategy that maximizes its total utility over the entire time horizon. 
\end{restatable}

Finally, a minor concern one might have is whether a randomized human strategy $\bbold$ could be strictly beneficial for the human since randomization has been proven to be powerful in strategic plays. Our next lemma shows that the answer is No. 
\begin{restatable}{lemma}{pureoptexists}\label{lemma:pure_opt_exists}
Randomization is not more beneficial: there always exists an optimal  strategy for the human that is deterministic.  
\end{restatable} 

Lemma \ref{lemma:pure_opt_exists} also simplifies our formulation. For the rest of the paper we will focus on finding the optimal solution $\bstar$ to Problem \eqref{fig:prob-formulation}. Since $\astar(\bbold)$ is determined by $\bbold$, we will drop the argument $\astar$ for the GenAI strategy when it is clear in the context and simply write $\uh_T(\bbold)$ for $= \uh_T(\bbold;\astar)$.

%% file: discounting_factor.tex
\section{Time-Sensitive Domains: $\gamma < 1$}\label{sec:discounted}
We start our algorithmic study for  time-sensitive content domains, in which the relevance   or usefulness of the content diminishes   over time  (e.g., news or fashion contents). We capture this domain by considering discounting factors $\gamma_i$'s  that are all  bounded below $1$, i.e., $\forall i \in \arms, \gamma_i \leq \gmax <1$ for some $\gmax$. Based on Lemma \ref{lemma:pure_opt_exists} we restrict our analysis to deterministic  strategies. 
In Subsection \ref{subsec:discount_hardness} we prove the computational hardness of our problem. Specifically, unless the  well-believed Randomized Exponential Time Hypothesis (\reth{}) is false \cite{calabro2008complexity}\footnote{ \reth{} is a popular assumption in complexity theory and widely believed to be true \cite{roughgarden2020algorithms,carmosino2016nondeterministic,braverman2014approximating,dell2014exponential,williams2018some}.}, we cannot maximizie the human's utility in polynomial time. 
In fact, this intractability holds even in simpler special situations when all arms have the same cost and, additionally, have the same discounting factor (i.e., $\gamma_i=\gamma, \ \forall i \in \arms$) or the same shrinkage function (i.e., $g_i(.)=g(.), \ \forall i \in \arms$).  
Therefore, in Subsection \ref{subsec:approx_discounting} we turn to polynomial time approximation algorithms and design  a new $O(T + k^{O(\log \frac{1}{\epsilon})} )$ algorithm with approximation ratio   $\frac{1-\epsilon}{2}$ for any given constant $\epsilon \in (0,\frac{1}{2})$. Our algorithm is based on a tailored design of ``window switching'', it runs a myopic optimization algorithm for a consecutive time window (interval) of size $\tau$, then switches to ``pausing'' (i.e., not playing any arm) for an equally sized time window $\tau$. 


\input{discounting_factor_hardness}

\input{approximaton_discounting_factor}

%% file: discounting_factor_hardness.tex
\subsection{The Hardness of Human Utility Maximization}\label{subsec:discount_hardness}
The hardness proof for our problem is based on a connection to the blocking bandits problem of \citet{basu2019blocking}. Like in the standard stochastic bandit setting, in blocking bandits we have $k$   arms each having its own mean $\mu_i$. However, each arm also has a delay value $D_i$ which is a positive integer. If an arm $i$ is pulled at round $t$ then it cannot be pulled again for $D_i-1$ many consecutive rounds after $t$, making $t+D_i$ the first subsequent round where $i$ maybe pulled again. 
Interestingly, \citet{basu2019blocking} show that even when the delay values $D_i$ and means $\mu_i$ of each arm are known, the pure maximization problem\footnote{By pure here we mean that no learning is involved and optimization is done with all parameters known.} of selecting the arms to maximize the total accumulated rewards through the horizon subject to not violating the delay constraints of the arms is computationally hard. Specifically, they show that the problem does not admit a pseudo-polynomial time algorithm in the number of arms unless the randomized exponential time hypothesis is false. We prove the hardness of our problem by reducing a special instance of blocking bandits (shown to be computationally hard) to the problem of finding a deterministic optimal strategy in our discounted contents setting. Our main result   is as follows:
\begin{theorem}\label{th:gen_hardness}
Let $\gamma_i < 1, \ \forall i \in \arms$, unless the randomized exponential time hypothesis (\texttt{rETH}) is false, there is no polynomial time algorithm for finding a deterministic optimal strategy for the human even if all arms have the same cost $c$ and the same discounting factor $\gamma$ or if all arms have the same cost $c$ and the same shrinkage function $g(.)$.
\end{theorem}
\vspace{-0.25cm}
The key idea in the reduction is to make \emph{violating the delay $D_i$ for an arm $i$ costly} by making the resulting generative mean too high and therefore leading to small utility. Based on the delay value $D_i$ for an arm $i$ we essentially will either tune its shrinkage function $g_i(.)$ when all arms have the same discount factor $\gamma$ or tune its discount factor $\gamma_i$ when all arms have the same shrinkage function $g(.)$ so that the final generative mean is too high once the delay is violated but remains below a certain threshold if the delay is not violated. Therefore, the only way to maximize utility is to pull the arms but without violating their delay. The full details of the proof however require tuning a large number of parameters as well as careful bounding arguments leading to a tedious derivation. 




%% file: approximaton_discounting_factor.tex
\subsection{An (Almost) $1/2$-Approximation}\label{subsec:approx_discounting}
First, to establish our bounds, we make natural and standard assumptions about our setting. Specifically, we assume that all shrinkage functions have a Lipshitz constant upper bound $L_g$, i.e., $\forall i \in \arms: |g_i(n_1)-g_i(n_2)| \leq L_g |n_1-n_2|$. Similarly, we assume that the link function is also $L_{\probc}$ Lipshitz, meaning that $|\probc(\mu, \mu_1)-\probc(\mu, \mu_2)| \leq L_{\probc} |\mu_1-\mu_2|$. Further, we assume that the instances we consider have a lower bound on the maximum utility we can gain in the first round, i.e., $\max\limits_{i \in \arms} \big(\probc(\mu_i,\maxgenmu{1})- c_i \big) \ge \uo>0$.\footnote{Note this is a mild assumption as it can be established by assuming the instances have a lower bound on the maximum mean, the shrinkage function at value $0$, and an upper bound on the cost. In fact, the computationally hard instances in the proof of Theorem \ref{th:gen_hardness} fall under this category.}

Our approximation algorithm  follows a natural structure of cycling between a short window (interval) where it pulls arms based on a \emph{myopically} optimal strategy and an equally sized window where it pauses (does not pull arms). Before  introducing  our algorithm, we start with a useful observation. Denote by \myopicopttau{} the optimal value we obtain by running a strategy which optimizes utility from the first round $t=1$ to round $t=\tauwindow$. The following lemma states that the utility we obtain over this window is larger than the utility we obtain from a full horizon optimal strategy (optimizing until $t=T$) over any equally sized window. 
\begin{restatable}{lemma}{lemmamyopic}\label{lemma:myopic_opt_tauwindow}
Let $\bstar \in \arms^T$ be a deterministic optimal strategy for the human, then for any  $\tauwindow$ consecutive rounds  from any $t$ to $t+(\tauwindow-1)$ we have 
\begin{align}
    \myopicopttau \ge u_{t: t+ \tauwindow -1}(\bstar)
\end{align}
\end{restatable}


It is easy to see that an algorithm with running time exponential in $\tau$ (specifically, $O(k^\tau)$) can be easily designed (e.g., by simply running a single source longest path algorithm in a directed acyclic graph as detailed in Appendix \ref{app:basic_stg}) for finding the optimal human strategy $\bbold$ for \myopicopttau{}. We denote this strategy by $\bmyopic(\tauwindow)$. Although the run-time is exponential in $\tau$, when $\tau$ is a small constant (as will be in our design), this leads to a polynomial time algorithm for finding $\bmyopic(\tauwindow)$.


Using $\bmyopic(\tauwindow)$, we construct the strategy \textsc{Myopically-Optimize-then-Pause} as shown in Algorithm \ref{alg:myopic_then_pause_discount} which alternates between running the same actions from $\bmyopic(\tauwindow)$   and pausing for $\tauwindow$ rounds.  Theorem \ref{th:optimzethenpause} shows that we can   obtain a $\frac{1-\epsilon}{2}$ approximation if we set $\tauwindow =\ceil{\log_{(\frac{1}{\gmax})} \Big( \frac{\gmax}{(1-\gmax)^2} \cdot \frac{L_g L_{\probc}}{\uo} \cdot \frac{1}{\epsilon} \Big) }=O(\log \frac{1}{\epsilon})$. Note that this does not require $\tauwindow$ to have a dependence on the input parameters $k$ or $T$. The key intuition behind our algorithm is that pausing would enable us to essentially force the GenAI to ``forget'' the contents since their value will be discounted. Therefore, when we run $\bmyopic(\tauwindow)$ again we would obtain a large utility due to Lemma \ref{lemma:myopic_opt_tauwindow}. However, since we are pausing for only a finite duration, the number of contents does not truly go back to $0$. Therefore, the core technical challenges in the proof lie at analyzing the utility loss during the pause window and the utility degradation during the $\bmyopic(\tauwindow)$ window, and carefully balancing them. 



%
\begin{algorithm}[h!] 
    \caption{\textsc{Myopically-Optimize-then-Pause}}
    \label{alg:myopic_then_pause_discount}
    \begin{algorithmic}[1]
    
    \STATE Set $\tauwindow = \ceil{\log_{(\frac{1}{\gmax})} \Big( \frac{\gmax}{(1-\gmax)^2} \cdot \frac{L_g L_{\probc}}{\uo} \cdot \frac{1}{\epsilon} \Big) }$
    \STATE  Find the optimal human strategy for the first $\tauwindow$ rounds and denote it by $\bmyopic(\tauwindow)$.
    \WHILE{$t \leq T$} 
    \STATE Alternate between following strategy $\bmyopic(\tauwindow)$ and pausing for $\tauwindow$  many rounds. 
    \ENDWHILE
    \end{algorithmic}
\end{algorithm}

\begin{restatable}{theorem}{optimzethenpause}\label{th:optimzethenpause}
Suppose  $ \gmax, L_g, L_{\probc}, \uo$ are constants, then \textsc{Myopically-Optimize-then-Pause} (Algorithm \ref{alg:myopic_then_pause_discount}) runs in $O(T + k^{O(\log \frac{1}{\epsilon})} )$ and achieves an approximation ratio of $\frac{1-\epsilon}{2} \frac{\floor{\frac{T}{2 \cdot \tauwindow}}}{\floor{\frac{T}{2 \cdot \tauwindow}}+1}$. 
\end{restatable}

It follows that as $T \to \infty$ the approximation ratio is $\frac{1-\epsilon}{2}$. Moreover, it is easy to see that $\tauwindow$ in Algorithm \ref{alg:myopic_then_pause_discount} is increasing in $\gmax$. This implies that a smaller $\gmax$ leads to smaller $\tauwindow$ and a better approximation ratio for a fixed $T$.




A key feature of our algorithm is that it pauses for a large portion of the horizon. Although this seems counter-intuitive, we show that pausing is actually necessary for securing high human utilities. Specifically, by constructing a special example we obtain the below theorem which proves that if an algorithm does not pause and always pulls arms then it is possible for it to obtain an arbitrarily bad utility.

\begin{restatable}{theorem}{pauseth}\label{th:pauseth}
There   exists an instance such that human will receive negative utility if she pulls arms for more than $0.1T$ rounds (not necessarily consecutively).  
\end{restatable}

%% file: long_horizon.tex
\section{Time-Insensitive Domains: $\gamma = 1$}\label{sec:long_horizon}
We now turn to time-insensitive domains, where contents remain relevant, accurate, and useful regardless of when they are accessed or addressed. This captures content generation for history knowledge, timeless literature, science and math knowledge, etc. 
We capture this situation by assuming all arms have discounting factors equal to $1$: $\forall i \in \arms: \gamma_i=1$. Clearly, the more pulls an arm receives by the human the higher the value of its generative mean $\genmu{i}(t)= \mu_i - g_i(\pullsh{i}(t-1))$. Since there is no discounting in this setting, the generative means $\genmu{i}$'s will be non-decreasing and in the long run GenAI will be equally competitive as the human for every topic, i.e.,  $\genmu{i} \to \mu_i , \ \forall i \in \arms$.
For our algorithmic result,  the concrete format of the $g_i$ functions and the link function $\sigma$ do not matter. What will matter is the following  ``thresholding'' value for the generative mean and number of pulls, which are determined by $g_i$, $\sigma$, and $c_i$. 
\begin{definition}[\textbf{Exit Mean} and \textbf{Exit Pull}]
For an arm $i \in \arms$, we define $\muexit{i}$ as the smallest value such that $\probc(\mu_i,\muexit{i}) - c_i <0$. Further, the smallest number $n$ such that $\mu_i - g_i(n) \ge \muexit{i}$ is called the \emph{exit pull} and denoted by $\nexit{i}$. We assume that  for each $i$ we have $\nexit{i} < \infty$. 
\end{definition}
In other words, the $\nexit{i}$'th pull of arm $i$ is the first time $i$ has a generative mean larger than $\muexit{i}$ and hence the first time  a negative utility will be generated if both the human and GenAI pull $i$ since $\probc(\mu_i,\muexit{i}) - c_i <0$. The assumption $\nexit{i} < \infty$ implies any arm has such a moment. 

Throughout this section, we naturally assume that the human will never pull an arm if it has negative utility, since it is surely worse than the trivial $0$-utility opt-out action even when the long-term utility is considered. We first characterize the set of ``plausible'' arms that can possibly be pulled, as summarized in the following observation.
\begin{observation}
Let $ \agmgm(t)$ denote the set of arms at round $t$ that have non-negative human utilities. Then we have \begin{align}
    \agmgm(t) = \{i \in \arms | \muexit{i} > \Tilde{\mu}_{\mgmt}\}.
\end{align} Moreover, $\agmgm(t) \subseteq \agmgm(t-1)$ for every $t=2,3,\cdots, T$.
\end{observation}
To develop intuition about our problem, it is useful to  understand  why the observation holds. First, any $i \in \agmgm(t) $ must give non-negative utility. Indeed, when pulling $i$, the human's utility  $\probc(\mu_i,\Tilde{\mu}_{\mgmt}) - c_i  $ is at least $0$ because $ \Tilde{\mu}_{\mgmt} < \muexit{i}$ and $\mu = \muexit{i}$ is the smallest value to make $\probc(\mu_i, \mu) - c_i < 0$. Second, any $i \not \in \agmgm(t) $ must have negative human utility because such $i$ satisfies $ \Tilde{\mu}_{\mgmt} \geq   \muexit{i} $, hence pulling $i$ leads to human utility  $\probc(\mu_i,\Tilde{\mu}_{\mgmt}) - c_i \leq \probc(\mu_i,\muexit{i}) - c_i <0 $. Finally, $\agmgm(t) \subseteq \agmgm(t-1)$ simply because $\Tilde{\mu}_{\mgmt}$ is non-decreasing in $t$ whereas $\muexit{i} $ does not change for each $i$. As a consequence of the above observation, we know that $\agmgm(1)$ contains all the arms that can possibly have positive utilities for human. Now we define the long-horizon setting that we focus on in order to understand humanity's long-term welfare. 
\begin{definition}[\textbf{Long Horizon Setting}]\label{def:long horizon}
The horizon $T$ is long if $T \ge \sum_{i \in \agmgm(1)} \nexit{i}$. 
\end{definition} 
Notably, when the horizon is long, the human will eventually have negative utilities at every arm. Since each arm $i$ can be pulled at most $\nexit{i}$ times as the human never pulls an arm with negative utility due to the existence of $0$-utility opt-out choice, and only arms in $\agmgm(1)$ can possibly be pulled by the human.
Consequently, under the long horizon setting the human's optimization reduces to collecting as much utility as possible before being fully driven out from the competition.\footnote{Note that this is different from staying within the competition for as long as possible since longest stay may not necessarily lead to maximum accumulated utility.} Despite appearing frustrating at a first glance, this should be expected to hold in the long run for time-insensitive content domains such as history knowledge or basic math where creating good content is  costly for the human whereas on the other hand GenAI is increasingly better at answering these factual questions. However, maximizing the human's utility in long horizon settings is highly non-trivial. The theorem below shows our main result for this setting, which is an polynomial time algorithm, which sharply contrasts the computational hardness of Section \ref{sec:discounted}.  

   

\begin{theorem}\label{th:long_horizon_max_utility}
In the long horizon setting, there is an $O(T k^3)$ algorithm that maximizes the human's utility.   
\end{theorem}

We note that simple baseline algorithms would either run in exponential time or not lead to an optimal strategy, as detailed in Appendix \ref{app:long_horizon_greedy_subopt}.
Through a collection of lemmas we define a class of strategies which we call \emph{synchronizing} strategies. Specifically, at any round $t$ if there exists an arm whose pull would not increase the maximum generative mean value $\maxgenmu{t}$, then it would be pulled by a synchronizing strategy. We show that there always exists an optimal strategy that is synchronizing and that synchronizing strategies follow a much simpler structure. This enables us to obtain an optimal synchronizing strategy by solving the longest path problem in a \emph{reduced} DAG having at most $O(Tk)$ nodes and can be constructed in $O(T k^3)$ time.

%% file: experiments.tex
\section{Experiments}\label{sec:experiments}



In this section we carry out simulations in the time-sensitive domains (Subsection \ref{subsec:time_senstive}) and time-insensitive domains (Subsection \ref{subsec:time_insenstive}) comparing the performance of our algorithms to a collection of baselines in each setting. Throughout the experiments we set the link function $\sigma$ according to the standard Bradley–Terry model \cite{bradley1952rank}, i.e., $\probc(\mu_{i}, \mu_j ):= \frac{e^{\mu_i}}{e^{\mu_i} + e^{\mu_j}}$. Further, we set the shrinkage function for a given topic $i$ to $g_i(N):= \frac{q_i}{\sqrt{N + (\frac{q_i}{h_i})^2}}$ where $q_i$ and $h_i$ essentially control the vertical scaling and horizontal shift of the function with $g_i(0) = h_i$. Note that based on the above it follows that Lipshitz constants would be $L_{\probc} = \frac{1}{4}$ and $L_g = \max_{i \in [k]} \frac{h_i^3}{2q^2}$. In Appendix \ref{app:experiments} we conduct additional experiments. 


\subsection{Time-Sensitive Domains}\label{subsec:time_senstive}
Here we test the performance of our \textbf{Myopic-then-Pause} (Algorithm \ref{alg:myopic_then_pause_discount}) against a collection of baselines. We consider the \textbf{Greedy} baseline where at each round $t$ the human pulls the utility-maximizing arm, i.e., $b_t =\argmax_{i\in[k]} \big(\sigma(\mu_{i},\maxgenmu{t})-c_{i}\big)$. We also consider a ``pure'' myopic algorithm (\textbf{Pure-Myopic}) where the \textbf{Myopic-then-Pause} algorithm is run without an enforced pausing interval. Furthermore, we consider a fourth algorithm \textbf{BT-Pull} based on the Bradley-Terry choice model where the human randomly pulls the arms according to their means --- specifically, at round $t$ an arm $i$ is pulled with probability $p_i 
= \frac{j_t(i)\mu_i}{\sum_{i' \in [k]}j_t(i')\mu_i}$, where $j_t(i) = 1$ if pulling arm $i$ at round $t$ would result in nonnegative utility and $0$ otherwise. If all arms result in negative utility, the human instead pauses for one round. We consider a final baseline \textbf{Cycle} where the human cycles between arms, skipping arms that would result in negative utility, and pausing if all arms would have resulted in negative utility. 

For this experiment, we have 4 arms ($k=4$) and means $\muvec= [.73, .85, .9, .95]$,  discounting factors $\gammavec = [.5, .48, .47, .45]$, costs $\cvec = [.53, .56, .59, .58]$, shrinkage parameters $\qvec = [.1, .2, .5, .2]$, and $\hvec= [.3, .45, .6, .6]$. The horizon is set to $T=1,000$, and the approximation constant for \textbf{Myopic-then-Pause} is set to $\epsilon=0.1$.


Figure \ref{fig:experiments} (top) shows the performance of all of these algorithms. In particular, the figure shows the total accumulated utility at a given round $t$, i.e., $u_{1:t}$. Clearly, our algorithm outperforms the other baselines by a significant margin. 



\begin{figure}[h!]
    \centering
    \includegraphics[width=.5\linewidth]{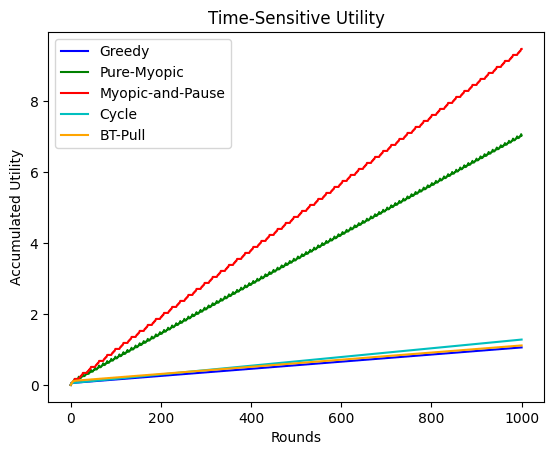}
    \includegraphics[width=.5\linewidth]{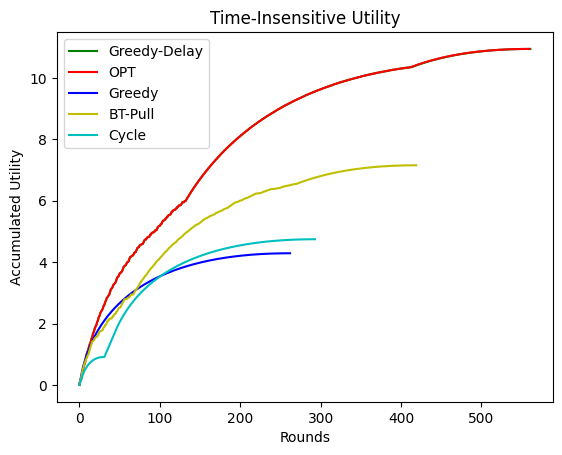}
    \caption{Accumulated utility vs round for each algorithm in: (top) time-sensitive domain. (bottom) time-insensitive domain.}  
    \label{fig:experiments}
\end{figure}




\subsection{Time-Insensitive Domains}\label{subsec:time_insenstive}
Our algorithm which we call \textbf{OPT} for this setting is the reduced DAG based algorithm discussed in Section \ref{sec:long_horizon} and Theorem \ref{th:long_horizon_max_utility} which is always guaranteed to achieve optimal performance. We include some of the baselines from the time-sensitive domain setting. In particular, we still compare against \textbf{Greedy}, \textbf{BT-Pull}, and \textbf{Cycle}. In addition, a reasonable baseline in this setting is \textbf{Greedy-Delay} where in each round $t$ the human pulls the arm that gives the smallest increase in the GenAI's maximum generative mean. At an intuitive level, this baseline would enable the human to stay longer in before exiting the platform where gaining non-negative utility is impossible.

We test an instance with $4$ arms where we set the means to $\muvec=[.86, .95, .87, .95]$, costs to $\cvec=[.56, .52, .55, .53]$, and have shrinkage parameters $\qvec = [1.4, 1.3, 1.5, 1.8]$ and $\hvec = [.75, .8, .6, .78]$. The horizon is set to $T=561$, which is the maximum number of rounds the human can pull arms for positive utility. 




Figure \ref{fig:experiments} (bottom) shows the performance of our algorithm along with the baselines. Clearly, our algorithm achieves the best performance significantly outperforming \textbf{Greedy}, \textbf{BT-Pull}, and \textbf{Cycle}. Interestingly, \textbf{Greedy-Delay} achieves an almost identical performance to our algorithm and has a nearly identical and overlapping trajectory. In fact, the difference in total utility is very small with \textbf{OPT} achieving a total of $10.9460$ while \textbf{Greedy-Delay} achieves $10.9459$. However, we show examples in Appendix \ref{app:experiments} where \textbf{OPT} outperforms \textbf{Greedy-Delay} by more significant margins. 

%% file: conclusion.tex
\section{Discussion, Conclusion, and Future Work}\label{sec:conclusion}

In this paper, we studied the dynamic interaction between a human contributor and a GenAI agent in an online content creation platform. We introduced a model that captures the low production cost of GenAI as well as its need for human content to improve its performance. Our focus was on designing algorithms for optimizing the utility of the human agent. Interestingly, we showed that time-sensitive domains are computationally difficult; not allowing polynomial time algorithms unless the randomized exponential time hypothesis is false. Therefore, we designed an approximation algorithm for this setting. Our algorithm is intuitive and simple, cycling between running a myopic optimization algorithm for a window and pausing for an equally sized window. Finally, we showed that time-insensitive domains with large horizon values permit optimal algorithms that follow a simple structure. Though we have a computational advantage in the time-insensitive domain over the time-sensitive domain, in the time-insensitive domain the human ends up necessarily exiting the platform whereas in the time-sensitive domain the human can essentially remain in the system for an arbitrarily long horizon.

Our work gives rise to many open questions for future investigation. The immediate follow-up question is to improve our algorithmic results, pinning down tight approximate ratios and hardness lower bound. 
Second, our current model has two players (human and GenAI) competing with each other. Future works can study competition among many players and alternative modes of competion in which humans use GenAI as an additional tool for content creation rather than another competitor. Third, we assumed that human and GenAI create contents at the same pace. It is interesting to investigate situations where GenAI content creation is much faster (e.g., can be instant) but is limited by the total number of creatable contents.     


Further, an important assumption implicit in this model is that the GenAI's capabilities at different topics are orthogonal to each other hence topic-$i$ contents do not affect AI's capability at another topic $j$. This assumption is also adopted in multiple recent works   \cite{yao2024human,immorlica2018incentivizing,dai2024can} for model simplicity, but it is an interesting future direction to study situations with cross-topic interactions (akin to the generalization from $i$-armed bandits \citep{auer2002finite} to linear stochastic bandits \citep{abbasi2011improved}).

Finally, we have not discussed the impacts on social welfare which would involve not only the human contributors but also the users. How does GenAI affect social welfare? Our time-insensitive (long horizon) setting shows that human contributors will eventually have to exit the platform, thereby degrading contributor welfare --- but could this degradation be balanced by users having an abundance of GenAI content? Are there interventions which the principal (platform operator) can employ to improve the overall welfare?  

%% file: impact_statement.tex
\section{Impact Statement}
In this paper we focused on analyzing the competition between a GenAI and a human agent. Although our work could shed light on potential consequences or behaviors that may arise in real life because of this competition, we do not see any potential negative societal consequences on our work. Furthermore, our work focused on mathematical analysis and simulations instead of designing software that can be deployed, hence we do not foresee any immediate effects on real life scenarios.




%% file: appendix.tex
\clearpage
\newpage
\onecolumn
\appendix
\input{useful_lemmas}
\input{omitted_proofs}

\section{Simple Baselines Are Computationally Intractable or Not Optimal for Time-Insensitive Domains (Undiscounted Long Horizon Setting)}\label{app:long_horizon_greedy_subopt}

\paragraph{Brute-Force Methods:}
It is not difficult to see that straightforward brute force algorithms would run in exponential time. A simple brute force algorithm would be to run dynamic programming with the states, transitions, and rewards being fully known. The issue is that the number of states would be $\Omega(T^k)$. Another brute-force method based on solving the longest path problem in the state transition directed acyclic graph would run in $\Omega(k^T)$. Therefore, both methods run in time exponential in the input parameters $T$ and $k$.

\paragraph{Greedy Approach:}
By a greedy algorithm we mean one which would always pick an arm that gives the smallest increase to the value of the maximum generative mean $\maxgenmu{t}$ in a given round $t$. Similar to Appendix \ref{app:long_horizon_proofs} we use the notation $\genmu{i}(p=r)=\genmu{i}(t)$ when $N_i(t-1)=r$, i.e., the value of the generative mean of arm $i$ if it has received $r$ many pulls. It is sufficient to consider the following simple example with two arms where $\maxgenmu{1}=\genmu{1}(p=0)=\genmu{2}(p=0)$ and $\genmu{1}(p=1) = \muexit{1}=\muexit{2}$, see Figure \ref{fig:greedy_does_not_work}. Clearly, any optimal strategy can pull arm $1$ once. We set $\mu_1 > \mu_2$. Since arm $2$ causes a smaller increase in the value of the generative mean as shown in the figure, for a greedy strategy the sequence of pulls would be $2,2,1$ leading to a utility of
\begin{align}
    \probc(\mu_2,\maxgenmu{1}) - c = \probc(\mu_2,\genmu{2}(p=1)) - c = \probc(\mu_1,\genmu{2}(p=2)) - c
\end{align}
Set $c_1=c_2=c=0.6$ and choose a function $\probc$ so that $\probc(\mu_2,\maxgenmu{1})=\probc(\mu_2,\genmu{2}(p=1))=\probc(\mu_1,\genmu{2}(p=2))=0.7-\frac{\epsilon}{3}$. The greedy strategy has a utility of $0.3-\epsilon$. On the other hand, let the $\probc(\mu_1,\maxgenmu{1})=0.9$ then pulling arm $1$ once would result in a utility of $0.3 > 0.3-\epsilon$. This shows that the greedy algorithm is not optimal. Note that we can devise a similar example that extends over a larger number of rounds/pulls.
\begin{figure}[h!]
    \centering
    \includegraphics[width=1.0\linewidth]{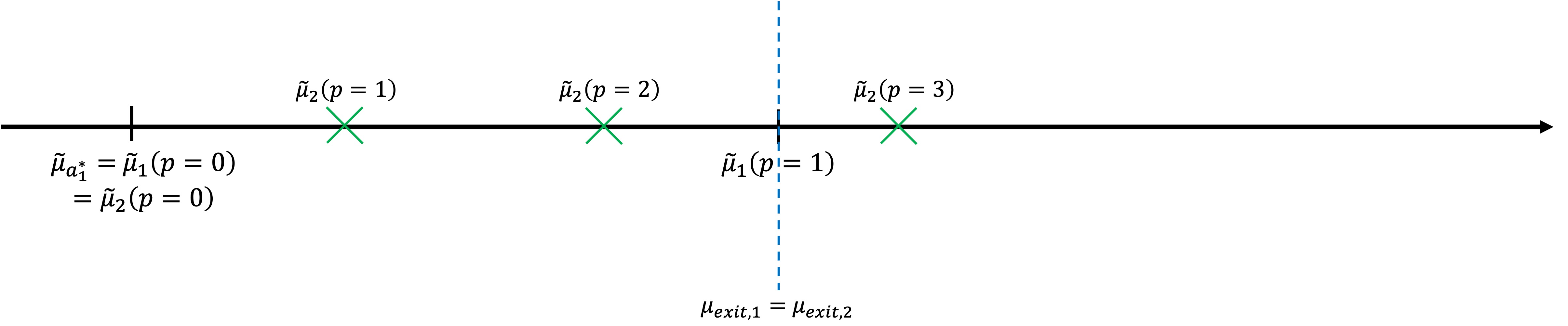}
    \caption{Figure shows how the generative means of arms $1$ and $2$ grow with more pulls.}
    \label{fig:greedy_does_not_work}
\end{figure}

\input{app_experiments}

%% file: useful_lemmas.tex
\section{Useful Mathematical Facts and Lemmas}
\begin{fact}\label{fact:binomial_upper_bound}
For any two number $a,b \ge 0$ and $n \in \mathbb{N}$ we have
\begin{align}
    (a+b)^n \leq a^n + n b (a+b)^{n-1}
\end{align}
\end{fact}
\begin{proof}
A straightforward proof can be given based on the binomial expansion. 
\end{proof}

\begin{lemma}\label{lemma:bound_f_gamma_D}
For $D \ge 2$ and $\gamma \in [0,\frac{1}{2}]$ we have 
\begin{align}
    \gamma \cdot \frac{1}{1-\gamma^{D}} < \gamma + \gamma^{D}   \label{eq:discount_ineq}
\end{align}
\end{lemma}
\begin{proof}
Using algebraic manipulations it is straightforward to show that the given inequality holds if and only if the following inequality holds 
\begin{align}
    \gamma^{D}  + \gamma -  1 < 0 \\
\end{align}
For all $\gamma \in [0,1]$ we have $\gamma^D + \gamma - 1 \leq \gamma^2 + \gamma - 1$ since $D\ge 2$. Therefore, it is sufficient to show that $\gamma^2 + \gamma - 1<0$ for $\gamma \in (0,\frac{1}{2}]$. 

First, we have $\gamma^2 + \gamma - 1\Big|_{\gamma=0} = -1$. Further, $\gamma^2 + \gamma - 1\Big|_{\gamma=\frac{1}{2}} = -\frac{1}{4}$. Therefore, it is sufficient to show that the derivative is positive (so that the function is always increasing) in the interval $(0,\frac{1}{2})$. The derivative is $\frac{d}{d \gamma} (\gamma^2 + \gamma - 1)= 2\gamma +1 > 0, \ \forall \gamma \ge -\frac{1}{2}$. Thus, the derivative is positive in $(0,\frac{1}{2})$. 
\end{proof}

\begin{lemma}\label{lemma:lipshitz_f}
For an integer $D \ge 2$ and $\gamma \in [0,\frac{1}{2}]$ and $\epsilon \in [0,\frac{1}{2}]$ we have for the function $f_D(x)=x^D+x$
\begin{align}
    f_D(\gamma+\epsilon)-f_D(\gamma) \leq (D+1) \epsilon
\end{align}
\end{lemma}
\begin{proof}
\begin{align*}
    f_D(\gamma+\epsilon)-f_D(\gamma) & = (\gamma+\epsilon)^D+(\gamma+\epsilon) - (\gamma^D+\gamma) \\ 
                                   & \leq \gamma^D + D \epsilon (\gamma+\epsilon)^{D-1} + (\gamma+\epsilon)- (\gamma^D+\gamma)  && \text{(By Fact \ref{fact:binomial_upper_bound})} \\
                                   & \leq \gamma^D + D \epsilon (1)^{D-1} + \gamma + \epsilon - \gamma^D - \gamma \\ 
                                   & \leq  D \epsilon + \epsilon  \\
                                   & = (D+1) \epsilon 
\end{align*}
\end{proof}

%% file: omitted_proofs.tex
\input{model_proofs}
\input{state_transition_graph}
\input{discounted_proofs}
\input{long_horizon_proofs}

%% file: model_proofs.tex
\section{Proofs for Section \ref{sec:model}}\label{append:model-proofs}
\textbf{Validty of Assumption \ref{prop:mgm_maximizes_prob}.}  For completeness, we start by justifying the statement in Assumption \ref{prop:mgm_maximizes_prob}; that is, at a given round $t$ the GenAI maximizes its probability of being chosen by the user in that round by selecting the arm with the current highest generative mean $\mgmt = \max_{i \in [k]} \genmu{i}(t)$.

We denote by $\bvec_t(i)$ and $\avec_t(j)$ the probability that the human and the GenAI will select arms $i$ and $j$ in round $t$, respectively, then we have 
\begin{align}
    \Prob(a_t > b_t) & = \sum_{i,j \in [k]} \bvec_t(j) \avec_t(i) \probc(\Tilde{\mu}_i,\mu_j) \\ 
                     & = \sum_{i \in [k]} \avec_t(i) \Big( \sum_{j \in [k]} \bvec_t(j) \probc(\Tilde{\mu}_i,\mu_j) \Big) \\
                     & \leq \sum_{j \in [k]} \bvec_t(j) \probc(\Tilde{\mu}_{\mgmt},\mu_j) && \text{(by the mean monotonicity property of $\probc$)}
\end{align}

We restate the next proposition and give its proof 
\leastutil*
\begin{proof}
Given strategy $\bbold$ which is possibly randomized (i.e., randomizing over the selected arms in a given round $t$), we can prove the following
\begin{align*}
        \uh_T(\bbold;\astar(\bbold)   & = \E[\sum_{t=1}^T \big[  \probc (\mu_{b_t}, \maxgenmu{t})    -   c_{b_t} \big]  ] \\ 
                               & \leq \E[\sum_{t=1}^T \big[  \probc (\mu_{b_t}, \Tilde{\mu}_{\armtai})    -   c_{b_t} \big]  ] && \text{(By the monotonicity of $\probc$)}\\ 
                               & = \uh_T(\bbold;\abold) 
\end{align*}
\end{proof}

We restate the next proposition and give its proof 
\myopicoptobv* 
\begin{proof}
We denote the human's oblivious strategy by $\bobv$. Let $\aprime$ denote some arbitrary GenAI strategy. Note that both $\bobv$ and $\aprime$ could possibly be randomized. 
We follow a proof by induction and accordingly start with the base case (utility over round $1$) which follows immediately by the monotonicity of $\probc$. 
\paragraph{Base Case:} $\uairange{1}{1}(\smgm;\sobv) \ge \uairange{1}{1}(\aprime;\sobv)$. 


\paragraph{Inductive Step:} 
Note that for any two strategies (even if randomized) $\bbold$ for the human and $\abold$ for the GenAI, by the linearity of the expectations the utility brakes down as 
\begin{align}
    \uairange{1}{t}(\abold;\bbold) = \uairange{1}{t-1}(\abold;\bbold) + \Prob(a_t > b_t;\abold;\bbold)  
\end{align}
Since the human's strategy is oblivious then it is independent of the GenAI's strategy and therefore the probability that the human will pull arm $i$ at round $t$ will be $\bvec_t(i,\sobv)$ for both strategies $\smgm$ and $\aprime$. Therefore, it follows by Assumption \ref{prop:mgm_maximizes_prob} that $\Prob(\mgmt > b_t; \smgm,\sobv) \ge \Prob(a_t > b_t; \aprime, \sobv)$. Further, since $\uairange{1}{t-1}(\smgm;\sobv) \ge \uairange{1}{t-1}(\aprime;\sobv)$ by the inductive hypothesis then $\uairange{1}{t}(\smgm;\sobv) \ge \uairange{1}{t}(\aprime;\sobv)$. By finally setting $t=T$ the proposition is complete. 
\end{proof}

We restate the next lemma and give its proof
\pureoptexists*
\begin{proof}
Our proof here is non-constructive. Consider an optimal randomized strategy. Since the GenAI's actions are dependent on the pulls the human makes since the pulls decide the values of $\genmu{i}(t)$ and since the human's decisions are independent of the realized rewards, the probability of strategy $\bbold$ pulling an arm $i$ at round $t$ is dependent on the previous pulls made by the human $\{b_1,b_2,\dots,b_{t-1}\}$.

We will convert strategy $\bbold$ to another strategy that does not randomize at the last round $t=T$. Consider some realization where the pulls are $\{b'_1,\dots,b'_{t-1}\}$, then at the last round the strategy will pull each arm $i \in \arms$ with probability $p_T(i|\{b'_1,\dots,b'_{t-1}\})$. Selecting no arm, i.e., playing $\emptyset$ happens with probability $p_T(\emptyset|\{b'_1,\dots,b'_{t-1}\})$. The incremental utility will be $\sum_{i \in \arms} p_T(i|\{b'_1,\dots,b'_{t-1}\}) (\probc(\mu_i,\maxgenmu{T})-c_i)$ by choosing instead arm $i = \argmax\limits_{j \in \arms \cup \emptyset} (\probc(\mu_j,\maxgenmu{T})-c_j)$ then we have a strategy that achieves the same utility but does not randomize in the last round. By successively invoking the same argument for rounds $T-1, T-2, \dots, 1$ we will have a deterministic strategy which achieves the same utility. 
\end{proof}

%% file: state_transition_graph.tex
\section{Dynamic Optimization as Longest Path in State Transition Graphs}\label{app:basic_stg} 
Here we show how to re-formulate human's dynamic optimization problem in Problem \eqref{fig:prob-formulation} as a longest path problem in a natural state transition graph, described below.
Before doing that it is useful to introduce the following notation for the generative mean of an arm. Specifically, we let $\genmu{i}(p=r)= \mu_i - g_i(r)$, i.e., the generative mean value of arm $i$ when the human has pulled the arm $r$ many times (or in general $r$ is the discounted number of pulls when $\gamma_i <1$). Note that we include ``$p=$'' in the argument of $\genmu{i}(p=r)$ to emphasize that it is a function of the number of pulls. This does not imply however that $\genmu{i}$ does not vary with different rounds. As stated earlier the action of not pulling an arm will be represented by pulling arm $\emptyset$. Following Lemma \ref{lemma:pure_opt_exists} our focus is restricted to deterministic strategies.

We now recall a standard result from graph theory. 
\begin{lemma}\label{th:dag_solvable_in_poly}
There exists an algorithm that can find the longest path in a directed acyclic graph (DAG) from a single source in $O(|V|+|E|)$ run-time where $V$ is the set of vertices and $E$ is the set of edges\footnote{Note that this holds even if the edge weights can be negative}. 
\end{lemma}
\begin{proof}
See for example \cite{sedgewick2011algorithms}. 
\end{proof}

This leads us to construct the state transition graph as follows:
\begin{enumerate}
    \item \textbf{Node Attributes}: Each node will have $k+1$ many numbers $(n_1,n_2,\dots,n_k; d)$ associated with it. For all $i \in [k], \ n_i$ denotes the number of (discounted) pulls for arm $i$. Additionally, the attribute $d$ indicates its distance from the root (first node). 
    \item \textbf{First Node:} Draw the \emph{first node} in the graph with $n_1=n_2=\dots=n_k=0$ and $d=0$. 
    \item \textbf{Recursive Node Drawing and Edge Connections:} Given a node $(n_1,n_2,\dots,n_k;d)$ such that $d< T$ for each $i\in \arms$ draw an edge to a new node $(n'_1,n'_2,\dots,n'_k;d')$ such that $n'_i=\gamma_i + \gamma_i^2 \cdot n_i$ and $\forall j \in \arms, j \neq i: n'_j=\gamma_j^2 \cdot n_j$. Further, for the new node set $d'=d+1$. Moreover, set the edge weight to $\probc(\mu_i,\max\limits_{j \in [k]} \genmu{j}(p=n_j)) -c_i$.  Finally, draw an edge to a new no action $\emptyset$ node, with $n'_i = \gamma_i^2 n_i \ , \forall i \in \arms$ and set the edge weight to $0$ and $d'=d+1$. If $d =T$, then the node has no outgoing edges. 
\end{enumerate}
See Figure \ref{fig:dag_exp_app_new} for an example of this state transition graph. The graph has a node for each state that can be reached and the edge weight is set equal to the utility gained by taking the action (pulling the corresponding arm). The graph is clearly a directed acyclic graph (DAG). It follows directly that maximizing utility is equivalent to finding the longest path starting from the first node $(n_1=0,n_2=0,\dots,n_k=0 ; d=0)$. Although by Lemma \ref{th:dag_solvable_in_poly} this can be done in time polynomial in the size of the graph, the size of the graph is $\Omega(k^T)$ which is exponential in the input parameters $T$ and $k$. Therefore, it would take $O(k^T)$ time to construct the graph and solve the single source longest path problem for it. This graph will however be important in our discussion and we will refer to it as the \emph{exponential graph}. 

\begin{fact}
Maximizing utility is equivalent to solving the longest path problem starting from the first node $(n_1=0,n_2=0,\dots,n_k=0;d=0)$ in the exponential graph. 
\end{fact}



\begin{figure}
    \centering
    \includegraphics[width=0.75\linewidth]{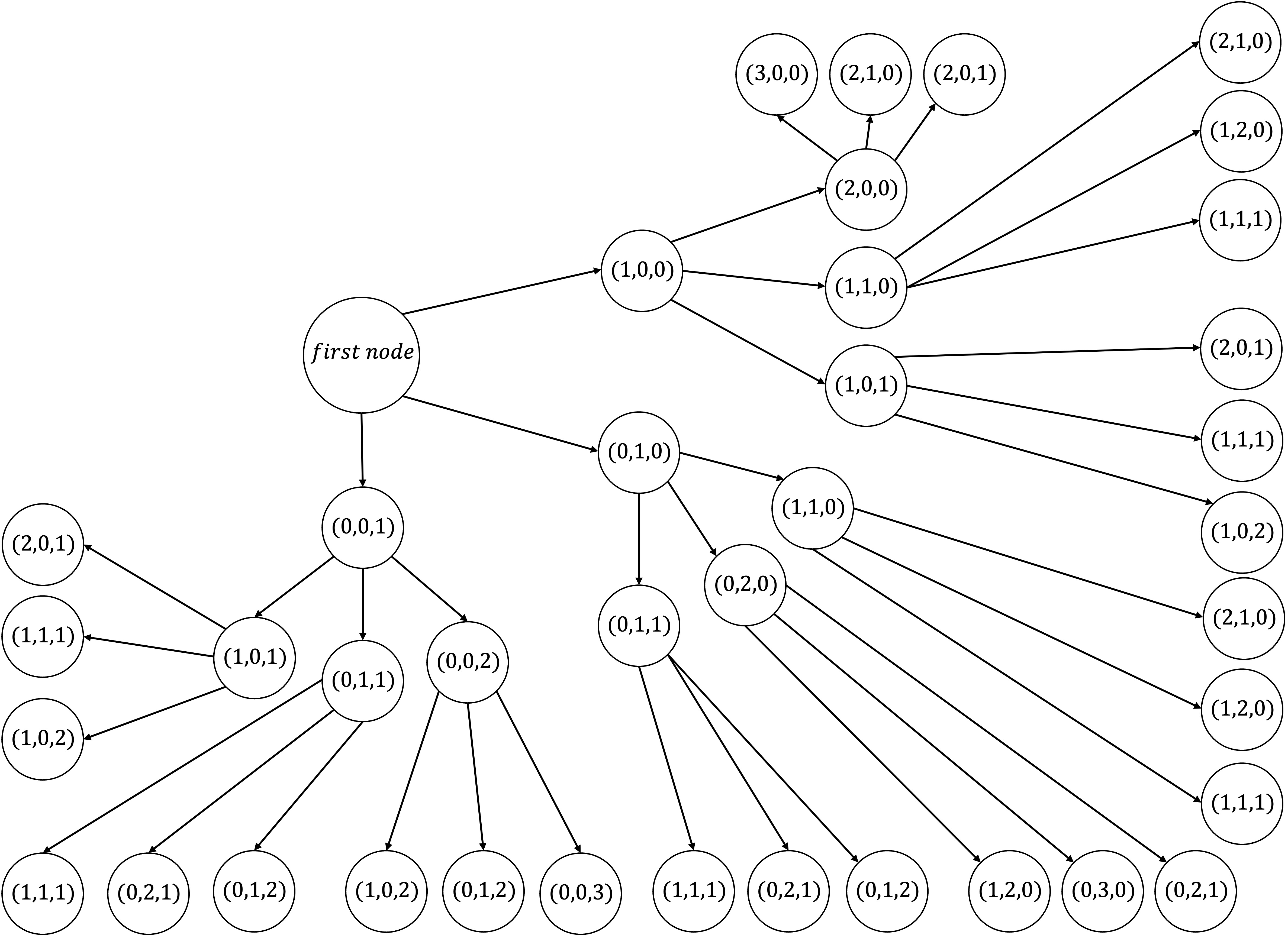}
    \caption{An example of the exponential graph for $k=3$ and $T=3$. For simplicity we assume that $\gamma_i=1 , \forall i \in [k]$ so that each node is recording the total number of pulls. The value of the attribute $d$ and  no action $\emptyset$ nodes are not drawn to keep the graph simpler.}
    \label{fig:dag_exp_app_new}
\end{figure}

%% file: discounted_proofs.tex
\section{Proofs for Section \ref{sec:discounted}} \label{app:discounted_proofs}
\subsection{Hardness Proofs for Subsection \ref{subsec:discount_hardness}}\label{app:hardness_proof}
\paragraph{Brief Review of the Blocking Bandits and the \maxreward{} Problem:}  In blocking bandits \cite{basu2019blocking}, each arm $i \in \arms$ has a mean $\mu_i$ and additionally has associated with it a \emph{delay value} $D_i \ge 1$. If arm $i$ is pulled at round $t$ then it cannot be pulled again until at least round $t+D_i$, i.e.,  for $D_i-1$ consecutive rounds we cannot pull arm $i$. The following is a simple instance to illustrate  the problem, consider the example below with delay values for arms $1,2$ and $3$ as shown in the table. Then a feasible solution to the given example would be $3\ 1\ 2\ 1\ 3\ 1\ 2\ 1\ $.


Given the above, suppose that all of the mean and delay values are known and that we want to pull arms through the horizon to maximize the accumulated expected rewards while not violating the delay values, i.e., solving the following \maxreward{} problem: 
\begin{align}
    \maxreward:
    & \max_{b_1,b_2,\dots,b_T} \sum_{t=1}^T \mu_{b_t} \\ 
    & \text{s.t. each arm being pulled such that its delay $D_i$ is not violated}
\end{align}
\citet{basu2019blocking} show that the following simple instance of \maxreward{} is computationally hard. Specifically, the result is  
\begin{lemma}[\cite{basu2019blocking}]\label{th:blockingMAB_is_np_hard}
Suppose that we have $k+1$ arms where the first $k$ arms have a mean of $1$ and a delay value $D_i$ with $\sum_{i=1}^k \frac{1}{D_i}=1$, the last arm has a mean $0$ and a delay of $1$, then deciding if \maxreward{} has a solution of value of at least $T$ cannot be done in pseudo-polynomial time in the number of arms $k$ unless the randomized exponential time hypothesis is false. 
\end{lemma}
We will use this instance to construct a reduction to our problem. Moreover, since the theorem above shows that no pseudo-polynomial time algorithm in the number of arms $k$ exists, we will further assume that for all arms $i \in [k]$ we have $D_i =\poly(k)$ where $\poly(.)$ is some polynomial function. We will then show a polynomial time reduction of such instances to our problem. This would then imply that our problem does not admit a polynomial time solution unless the randomized exponential time hypothesis is false.


Theorem \ref{th:gen_hardness} is proved based on Lemmas \ref{th:NP_hard_same_gamma} and \ref{th:NP_hard_same_shrinkage} which are proved below. We start with Lemma \ref{th:NP_hard_same_gamma} where all arms have the same discount factor $\gamma$ 

\input{discounted_same_gamma_hardness}

\input{hardness_same_shrinkage}

\input{approximation_proofs}

%% file: discounted_same_gamma_hardness.tex
\begin{restatable}{lemma}{NPhardsamegamma}\label{th:NP_hard_same_gamma}
Suppose $\gamma_i < 1$ for all $i \in \arms$, unless the randomized exponential time hypothesis is false there does not exist a polynomial time deterministic optimal strategy for the human even if all arms have the same cost $c$ and the same discount factor $\gamma$. 
\end{restatable}
\begin{proof}
Similar to the \maxreward{} instance of Lemma \ref{th:blockingMAB_is_np_hard} we will have $k+1$ arms. We refer to the first $k$ arms when we say $i \in [k]$ instead of $i \in [k+1]$ and the last arm will have index $k+1$. The means of the first $k$ arms is set to the same value of $\mu=0.9$ whereas the last arm will have a value of $\mu_{k+1}=0.1$. Further, all arms have the same cost of $c=0.6$.   

We start by showing that a discount factor of $\gamma = \frac{1}{2}$ is sufficient to make the discounted number of contents strictly higher whenever the delay value is violated for an arm with a delay $D_i \ge 2$ in our reduction. 
We start with the following claim
\begin{claim}\label{cl:nbar_lower_bound}
If a pull violates the delay rule for an arm $i$ at round $t$, then we have:
\begin{align}
    N_i(t) \ge \gamma + \gamma^{D_i}
\end{align}
\end{claim}
\begin{proof}
Suppose the pull violates the delay rule at round $t$ then it must be the case that arm $i$ was pulled at least $D_i$ rounds before $t$. This means that the smallest value for the discounted number of contents will be:
\begin{align}
    N_i(t) \ge  \gamma + \gamma^{D_i} 
\end{align}
See Figure \ref{fig:viol_delay} for an example. 
\begin{figure}[h!]
    \centering
    \includegraphics[width=1.0\linewidth]{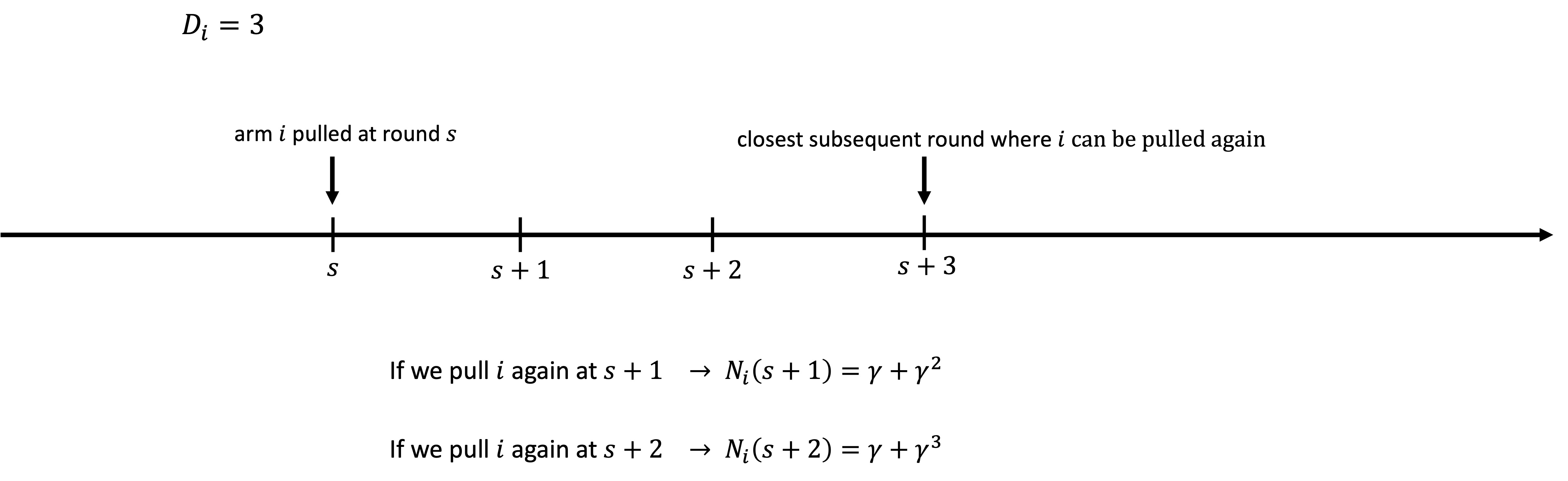}
    \caption{Example illustrating the value for the discounted number of contents for an arm $i$ with $D_i=3$ if its delay value is violated.}
    \label{fig:viol_delay}
\end{figure}
\end{proof}
Further, if the delay is never violated then we can establish an upper bound on the discounted number of contents as shown in the following claim
\begin{claim}\label{cl:nbar_upper_bound} 
If the delay rule for an arm $i$ is not violated then the discounted number of contents at any round $t$ is at most   
\begin{align}
    N_i(t) \leq \gamma \cdot \frac{1}{1-\gamma^{D_i}}
\end{align}
\end{claim}
\begin{proof}
If we are at round $t$ and the pulls for arm $i$ did not violate the delay then the highest value for the discounted number of contents would be as follows (see Figure \ref{fig:respect_delay} for an illustration): 
\begin{align*}
    N_i(t) & \leq \gamma + \gamma^{D_i+1} + \gamma^{2D_i+1} + \gamma^{3D_i+1} + ... \\ 
                 & = \sum_{\ell=0}^{\infty} \gamma^{D_i \ell + 1} \\
                 & = \gamma \cdot \frac{1}{1-\gamma^{D_i}}
\end{align*} 
 \begin{figure}
    \centering
    \includegraphics[width=1.0\linewidth]{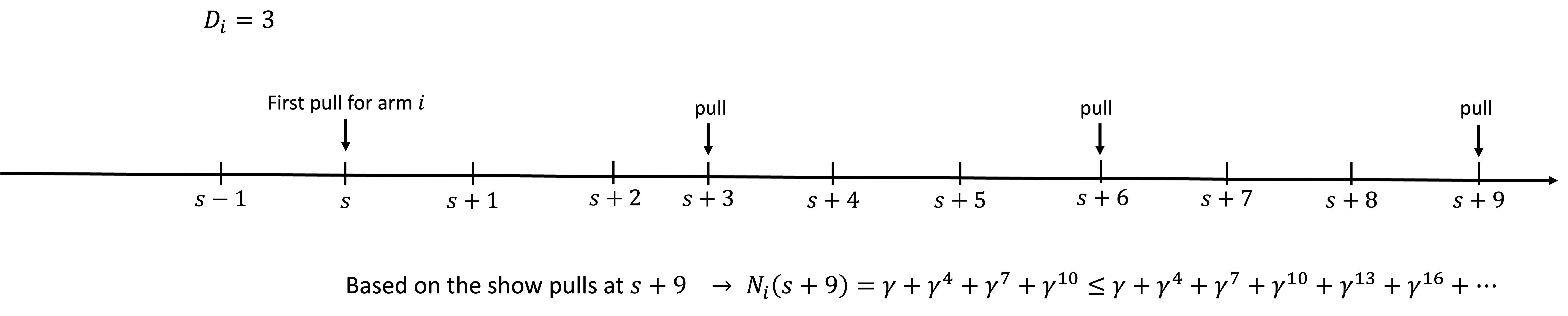}
    \caption{Example illustrating the upper bound on the discounted number of contents for an arm $i$ with $D_i=3$ if its delay value is not violated.}
    \label{fig:respect_delay}
\end{figure}
\end{proof}
We want to set $\gamma$ to a value such that the resulting discounted number of pulls for an arm when violating its delay is strictly higher than the discounted number of pulls when its delay is not violated. To do that based on the two previous claims it is sufficient to satisfy the following inequality 
\begin{align}
    \gamma \cdot \frac{1}{1-\gamma^{D_i}} < \gamma + \gamma^{D_i}   \label{eq:discount_ineq}
\end{align}
The following claim shows that this can be done by simply setting $\gamma=\frac{1}{2}$ as stated earlier. 
\begin{claim}\label{cl:half_special}
If $\gamma =\frac{1}{2}$ then for any arm $i$ with $D_i\ge 2$ we have
\begin{align}
    \gamma \cdot \frac{1}{1-\gamma^{D_i}} < \gamma + \gamma^{D_i}   \label{eq:discount_ineq}
\end{align}
\end{claim}
\begin{proof}
By Lemma \ref{lemma:bound_f_gamma_D} setting $\gamma=\frac{1}{2}$ would satisfy the inequality. 
\end{proof}

With $\gamma=\frac{1}{2}$ we can write the following fact.
\begin{fact}\label{fact:discounted_sep}
For any arm $i \in [k]$ with $D_i \ge 2$ if at round $t-1$ the delay rule was violated then $N_i(t-1) \ge \gamma+\gamma^{D_i} > \frac{\gamma}{1-\gamma^{D_i}}$. On the other hand, if the delay rule was not violated at any round up to and including $t-1$ then we have $N_i(t-1) \leq  \frac{\gamma}{1-\gamma^{D_i}} < \gamma+\gamma^{D_i}$. 
\end{fact}

Now we complete the details of the reduction, the generative mean for an arm $i$ is 
\begin{align}
    \genmu{i}(t)     & = \mu_i - g_i(N_i(t-1)) 
\end{align}
Note that since $\gamma=\frac{1}{2}$, the highest discounted number of contents will be $\frac{\gamma}{1-\gamma} =1$. For the last arm $g_{k+1}(n)=0.1-0.1n$. On the other hand, for the first $k$ arms we set 
\begin{align}
    g_i(n) = g(\lambda_i n) \ , \forall i \in [k]
\end{align} 
where $\lambda_i >0$ is to be set according to the delay value $D_i$ and $g(.)$ is a function shared by all of the first $k$ arms. The function $g(.)$ is set equal to the following: 
\begin{align}\label{eq:def_p_inv}
g(n) = 0.7 - 0.35 n 
\end{align}

Let $p^*=\frac{1}{2}$. Now for any $i \in \arms$ we will set $\lambda_i$ as follows: 
\begin{align}
    \lambda_i = \frac{p*}{\gamma + \gamma^{D_i}}
\end{align}
This implies that if the delay rule for an arm $i$ is violated then in the next round we have $\lambda_i N_i(t) \ge p^*$, see Figure \ref{fig:lambda_illust} for an illustration. Critically we note that since $D_i =\poly(k)$, then $\lambda_i$ can be computed in polynomial time and saved in polynomial space.

Based on the above, similar to Fact \ref{fact:discounted_sep} the value of the maximum generative mean will exceed a threshold if the delay rule is violated in the first $k$ arms. 
\begin{fact}\label{fact:discounted_sep_genmu}
For any arm $i \in [k]$ with $D_i \ge 2$ if at round $t-1$ the delay rule was violated then $\genmu{i}(t) \ge \genmuu = 0.375$. On the other hand, if the delay rule had not violated at any round up to and including $t-1$ then we have $\genmu{i}(t) \leq  \genmul < \genmuu$ where $\genmul= 0.9 - g(\max\limits_{i \in [k]} \frac{p^*}{\gamma + \gamma^{D_i}} \frac{\gamma}{1-\gamma^{D_i}})$. 
\end{fact}
\begin{proof}
If the delay is violated at round $t-1$ for an arm $i \in [k]$ then we have in the next round $t$
\begin{align*}
    \genmu{i}(t) & = \mu_i - g_i(N_i(t-1)) \\ 
                 & \ge \mu_i - g( \lambda_i \cdot (\gamma+\gamma^{D_i})) && \text{(by Fact \ref{fact:discounted_sep})}\\
                 & = 0.9 - g(\frac{1}{2}) && \text{($p^*=\frac{1}{2}$)}\\ 
                 & = 0.9  - (0.7 - 0.35 \times \frac{1}{2}) \\ 
                 & = 0.9 - 0.525  \\
                 & = 0.375 \\ 
                 & = \genmuu 
\end{align*} 
On the other hand if the delay rule is not violated then the maximum possible value given to the function $g(.)$ can be 
\begin{align}
    \max\limits_{i \in [k]} \lambda_i \frac{\gamma}{1-\gamma^{D_i}} = \max\limits_{i \in [k]} \frac{p^*}{\gamma + \gamma^{D_i}} \frac{\gamma}{1-\gamma^{D_i}} < p^*
\end{align}
Therefore, based on the definition of $\genmul$ in the fact, we will always have $\genmu{i}(t) \leq  \genmul < \genmuu$.
\end{proof}


Before we set the values of the link function we notes that the human will choose arms of mean either $0.9$ (first $k$) or $0.1$ (the last arm). Further, the maximum generative mean value is at least $0.9-g_i(0)=0.9-0.7=0.2$ and is at most $0.9-g_i(\frac{p^*}{\gamma+\gamma^{D_i} } \cdot \frac{\gamma}{1-\gamma} ) \leq 0.9 - g_i(1) = 0.9-(0.7-0.35)=0.55$. Therefore, it is sufficient to define the link function over $[0.1,0.9]\times[0.1,0.9]$.

We will now set the value for link function as follows: (A) if $\mu_1=0.1$ then $\probc(0.1,\mu_2) =\frac{1+(0.1-\mu_2)}{2}$. (B) if $\mu_1=0.9$ then we set the link function as
 \begin{align}
    \probc(0.9,\mu_2) = \begin{cases} 
         \frac{1+(0.9-\mu_2)}{2} & \text{if } \mu_2 \ge  \genmul \\
         \frac{1+(0.9-\genmul)}{2} & \text{if } \mu_2 <  \genmul 
\end{cases}
\end{align}

(C) for $\mu_1 \in (0.1,0.9)$ the link function is as follows 
 \begin{align}
    \probc(\mu_1,\mu_2) = \begin{cases} 
         \frac{1}{2} (\mu_1 - 0.1) + \probc(0.1,\mu_2) & \text{if } \mu_1 \leq  \mu_2 \\ \\ 
         \frac{\probc(0.9,\mu_2)-\frac{1}{2}}{0.9-\mu_2} (\mu_1-\mu_2)+ \frac{1}{2} & \text{if } \mu_1 >  \mu_2
\end{cases}
\end{align}
It is straightforward to see that $\probc$ is continuous, $\probc(\mu,\mu)=\frac{1}{2}$ and that it is increasing in the first argument.

 Note that by our choice for the shrinkage functions, the GenAI will always choose an arm from the first $k$ set since it is straightforward to show that $\maxgenmu{t} \ge \genmu{i}(t) \ge 0.2 > 0.1 \ge \genmu{k+1}(t) , \ \forall i \in [k]$ and $\forall t \in [T+1]$. Therefore, we can upper bound the utility of the last arm of mean $\mu_{k+1}$
\begin{align}
    \probc(\mu_{k+1},\maxgenmu{t}) - c \leq \frac{1+(0.1-0.2)}{2} - 0.6 = -0.15 < 0 
\end{align}
Therefore, the last arm always gives negative utility.

From the above we can now separate the utility we gain from the first $k$ arms based on whether we violated their delay rule or not 
\begin{fact}\label{fact:seperate_util}
If the delay rule for any arm arm $i \in [k]$ with $D_i \ge 2$ is violated at round $t-1$ then the utility at the next round will be at most $\thetal-c$ where $\thetal=\frac{1+(0.9-\genmuu)}{2}$. On the other hand, if the delay rule is not violated then we can obtain a maximum utility of exactly $\thetau-c>0$ where $\thetau=\frac{1+(0.9-\genmul)}{2} > \thetal$.  
\end{fact}
\begin{proof}
If the delay rule is violated then by Fact \ref{fact:discounted_sep_genmu} it follows that the utility will be at most $\frac{1+(0.9-\genmuu)}{2}-0.6=\frac{1+(0.9-0.375)}{2}-0.6=\thetal-c$ where $\thetal=\frac{1+(0.9-\genmuu)}{2}$ as defined. 

On the other hand, if the delay rule is not violated then pulling one of the first $k$ arms we can obtain a utility of $\frac{1+(0.9-\genmul)}{2}-0.6=\thetau-0.6$ where $\thetau=\frac{1+(0.9-\genmul)}{2}$ as defined. 

$\thetau > \thetal$ since $\genmuu > \genmul$. Further, since $\thetal-c = \frac{1+(0.9-\genmuu)}{2}-0.6 =\frac{1+(0.9-0.375)}{2}-0.6 >0$, it follows that $\thetau - c>0$. 
\end{proof}

\begin{figure}
    \centering
    \includegraphics[width=1.0\linewidth]{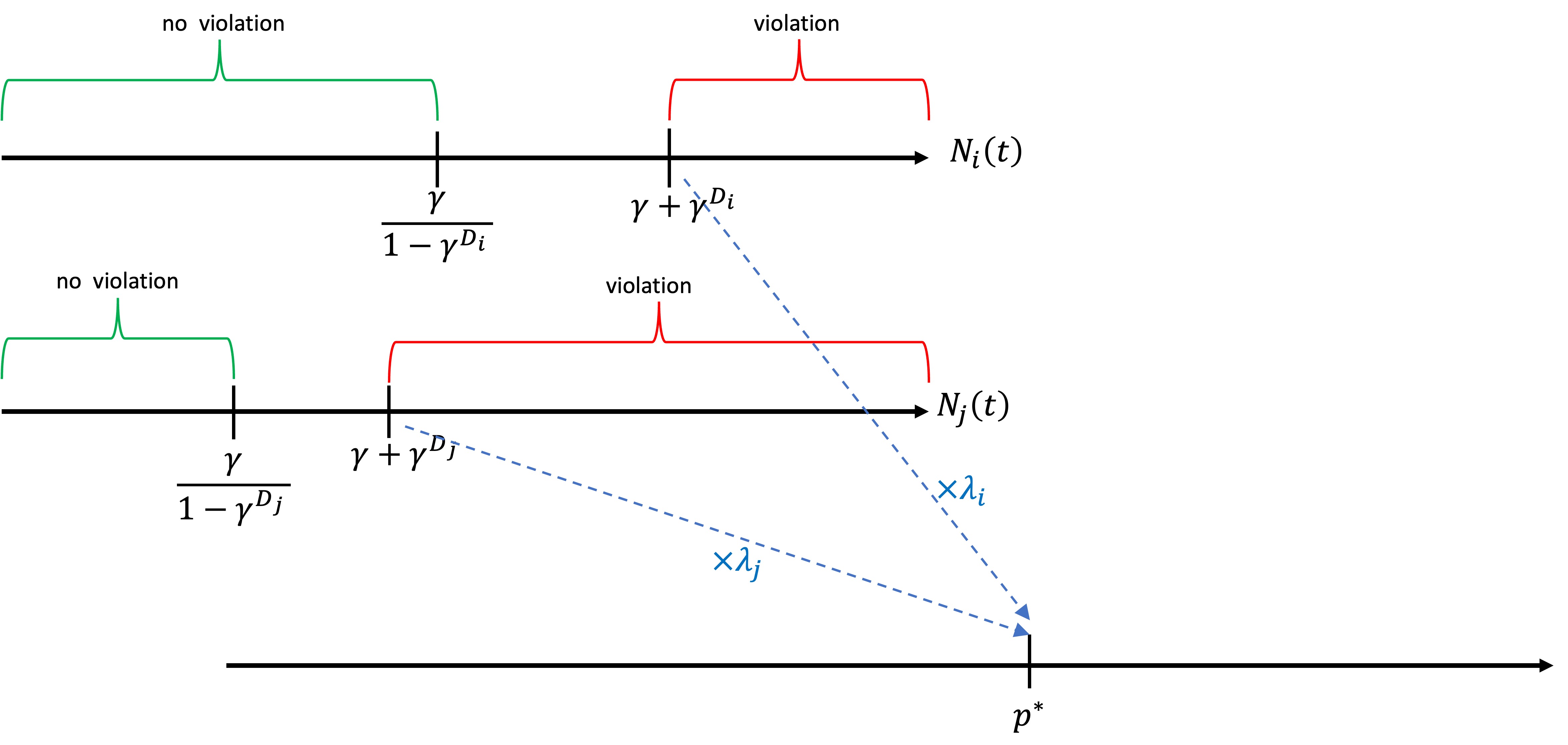}
    \caption{The figure shows how when the delay value is violated for any arm $i$ its coefficient $\lambda_i$ would lead to $\lambda_i N_i(t) \ge p^*$.}
    \label{fig:lambda_illust}
\end{figure}

Now that we have given the details of the reduction. The following lemma is immediate
\begin{lemma}
The human can obtain a utility value of $(\thetau-c) \cdot (T+1)$ if and only if he always pulls an arm from the set of first $k$ arms without violating their delay value except possibly at the last round $T+1$. 
\end{lemma}
\begin{proof}
Suppose the human violates the delay for some arm $i$ in the first $k$ set at a round $t \leq T$ then it follows that in the next round the highest utility that can be obtained is $\thetal-c<\thetau-c$ by Fact \ref{fact:seperate_util} and therefore the highest utility throughout the horizon will be $(\thetau-c) \cdot (T-1)$. 

Now suppose that only arms from the first $k$ set are pulled without violating the delay rule in the first round $T$ rounds, then it follows that the utility will be $(\thetau-c) \cdot (T+1)$. 
\end{proof}
By the previous lemma it follows that given a \maxreward{} instance of Lemma \ref{th:blockingMAB_is_np_hard} is a YES instance if and only if the human can obtain a utility of value $(\thetau-c) \cdot (T+1)$ in our reduction. 
\end{proof}

%% file: hardness_same_shrinkage.tex
Now we state and prove the following lemma which proves hardness even when all arms have the same shrinkage function $g(.)$.
\begin{restatable}{lemma}{NPhardsameshrinkage}\label{th:NP_hard_same_shrinkage}
Suppose $\gamma_i < 1$ for all $i \in \arms$, unless the randomized exponential time hypothesis is false there does not exist a polynomial time deterministic optimal strategy for the human even if all arms have the same cost $c$ and the same shrinkage function $g(.)$.
\end{restatable}
\begin{proof}
We will follow a reduction similar to the previous of Lemma \ref{th:NP_hard_same_gamma}. The difference is that all arms will have the same shrinkage function $g(.)$ but the discount factor for the arms will be will different and set according to the delay value $D_i$. Moreover, the first $k$ arms will have a mean of $\mu=1$ and the last arm will have a mean of $\mu_{k+1}=0.6$. All arms will have a cost of $c=0.77$. The shrinkage function for all arms will be set to
\begin{align}
    g(n) = 0.6 - 0.15 n  
\end{align}
Since we will have $\max\limits_{i \in [k+1]} \gamma_i \leq \frac{1}{2}$ so that $N_i(t) \leq \frac{\gamma_i}{1-\gamma_i} \leq 1$, the above values immediately imply that we have $\maxgenmu{t} \ge \genmu{i}(t) \ge 0.4 > 0.15 \ge \genmu{k+1}(t) \ , \forall i \in \arms, t \in [T+1]$. Therefore, the GenAI will never pull arm $k+1$ since it always picks the arm with the maximum generative mean in each round. 

The following two claims follow from Claims \ref{cl:nbar_lower_bound} and \ref{cl:nbar_upper_bound}, they are simply restated for different values of $\gamma_i$. 
\begin{claim}\label{cl:lower_bound_discount}
If a pull violates the delay rule for an arm $i$ at round $t$, then we have:
\begin{align}
    N_i(t) \ge \gamma_i + \gamma_i^{D_i}
\end{align}
\end{claim}

\begin{claim}\label{cl:upper_bound_discount} 
If the delay rule for an arm $i$ is not violated then the discounted number of contents at any round $t$ is at most   
\begin{align}
    N_i(t) \leq \gamma_i \cdot \frac{1}{1-\gamma_i^{D_i}}
\end{align}
\end{claim}
We further have the following claim which is similar to Claim \ref{cl:half_special}.
\begin{claim}
$\forall \gamma_i \in (0,\frac{1}{2}]$ with $D_i \ge 2$ we have 
\begin{align}
    \gamma_i \cdot \frac{1}{1-\gamma_i^{D_i}} < \gamma_i + \gamma_i^{D_i}   \label{eq:discount_ineq}
\end{align}
\end{claim}
\begin{proof}
The proof follows immediately by Lemma \ref{lemma:bound_f_gamma_D}. 
\end{proof}

For the last arm we set $\gamma_{k+1}=\frac{1}{2}$, but for the first $k$ arms with $D_i \ge 2$ we want to find values of $\gamma_i \in (0,\phi]$ where $\phi=0.35$ such that violating the delay values results in the discounted number of contents $N_i(t)$ being at least $\gamma_i + \gamma_i^{D_i}$ with 
\begin{align}
     \gamma_i + \gamma_i^{D_i} = \phi \label{eq:viol_satisfy}
\end{align}
It is straightforward to see that for $D_i \ge 2$ there always exist values of $\gamma_i \in (0,\phi]$ such that \eqref{eq:viol_satisfy} is satisfied. To see that, note that the function $\gamma_i + \gamma_i^{D_i}$ is continuous and that $\gamma_i + \gamma_i^{D_i}\Big|_{\gamma=0}=0 < \phi$ and $\gamma_i + \gamma_i^{D_i}  \Big|_{\gamma=\phi} > \phi$. 
Note the following simple fact which will be helpful later
\begin{fact}
If $D_i \ge 2$ and $\gamma_i + \gamma_i^{D_i} = \phi$ then we must have $\gamma_i \ge 0.2$. 
\end{fact}
\begin{proof}
This follows since $\gamma_i + \gamma_i^{D_i}$ is an increasing function and $0.2 + 0.2^{D_i} \leq 0.2 + 0.2^{2} \leq 0.24 < 0.35=\phi$. 
\end{proof}

Now, although the values of $\gamma_i$ that satisfy \eqref{eq:viol_satisfy} exist. It is not straightforward to find an exact closed form solution for them, so we will approximate them. Note that $\gamma_i + \gamma_i^{D_i}$ is increasing, therefore by doing binary search over the interval $[0.2,\phi]$ we can obtain an approximation. However, it is important to establish a bound on the desired approximation error. To do that we define the following function which equals the difference between the lower bound when violating the delay rule and the upper bound when not violating it for a given value of $\gamma \in [0.2,\phi]$ and $D \ge 2$ 
\begin{align}
    h(\gamma,D) = \gamma + \gamma^D - \frac{\gamma}{1-\gamma^D}  
\end{align}
We now reach the following lemma
\begin{lemma}\label{lemma:min_seperation}
For $\gamma \in [0.2,\phi]$ and $D \ge 2$. Let $\delta = \frac{5}{6} \times 0.2^{\dmax}$ then we have 
\begin{align}
     h(\gamma,D) \ge  \delta 
\end{align}
\end{lemma}
\begin{proof}
It is straightforward to show that $h(\gamma,D)$ is increasing in $\gamma$ and decreasing in $D$ for $\gamma \in [0.2,\phi]$ and $D \ge 2$. Therefore, it obtains the minimum value at $\gamma=0.2$ and $\dmax = \max\limits_{i \in \arms} D_i$ this lead to 
\begin{align}
    h(\gamma,D) & \ge h(0.2,\dmax) \\ 
                & = 0.2 + 0.2^{\dmax} - \frac{0.2}{1-0.2^{\dmax}}\\ 
                & = 0.2^{\dmax} + \frac{0.16 - 0.2^{\dmax+1}}{1-0.2^{\dmax}} \\ 
                & = 0.2^{\dmax} - \frac{1}{6} 0.2^{\dmax}  \qquad  \text{(since $\frac{0.16 - 0.2^{\dmax+1}}{1-0.2^{\dmax}} \leq \frac{1}{6} 0.2^{\dmax}$ for $\dmax \ge 2$)}\\ 
                & = \frac{5}{6} 0.2^{\dmax} 
\end{align}
\end{proof}

Therefore, we will approximate $\gamma_i$ for the first $k$ arms so that \eqref{eq:viol_satisfy} holds within an error of at most $\frac{\delta}{4}$. We will now show that $O(\poly(k))$ iterations are sufficient for any arm $i$ with $D_i \ge 2$. 

\begin{lemma}
Let $R_0 = \phi -0.2$ and let $f_{D_i}(\gamma_i)=\gamma_i + \gamma_i^{D_i}$. Then for a given value of $D_i \ge 2$ after $\log_2 \big(\frac{4 (D_i+1) R_0}{\delta} \big)$ iterations of binary search for the function $f_{D_i}(\gamma_i)$ over $\gamma_i \in [0.2,\phi]$ if we obtain $\ghat_i$ then we have 
\begin{align}
    |f_{D_i}(\ghat_i) - \phi| \leq \frac{\delta}{4}
\end{align}
\end{lemma}
\begin{proof}
After $q$ many iterations of binary search over the interval $R_0=\phi -0.2$ the size of the search region will be $\frac{R_0}{2^q}$. Any returned solution $\ghat_i$ in this region will be at most at a distance of $\frac{R_0}{2^q}$ from the true solution $\gstari$ where $f_{D_i}(\gstari)=\phi$. Based on Lemma \ref{lemma:lipshitz_f} we will have 
\begin{align}
    |f_{D_i}(\ghat_i) -f_{D_i}(\gstari)| & \leq (D_i+1) \frac{R_0}{2^q} \\ 
                                          & = (D_i+1) \frac{R_0 \delta}{4 (D_i+1) R_0} \\
                                          & = \frac{\delta}{4}
\end{align}
\end{proof}
Note that since $\dmax = \poly(k)$, then it follows that $\log_2 \big(\frac{4 (D_i+1) R_0}{\delta} \big) = \log_2 \big(4 (D_i+1) R_0 \cdot \frac{6}{5} \cdot 5^{\dmax} \big) = O(\log \dmax) + O(\dmax) = \poly(k)$. 

We will denote with $\ghat_i$ the resulting approximated discounting factors. We note that using our approximation we can establish the following separation between the discounted number of pulls based on whether the delay rule was violated or not. Specifically, we have the following fact
\begin{fact}\label{fact:sample_seperation_ver_2}
For any arm $i \in \arms$ with $D_i \ge 2$ if at round $t-1$ the delay rule was violated then $N_i(t-1)\ge \phi- \frac{\delta}{4}$. On the other hand, if the delay rule was not violated at any round up to and including $t-1$ then we have $N_i(t-1)\leq \phi -\frac{3 \delta}{4}$. 
\end{fact}
\begin{proof}
Since our approximations for \eqref{eq:viol_satisfy} are up to an additive error of $\frac{\delta}{4}$, then it follows that if we violate the delay rule at round $t-1$ that we have 
\begin{align}
    N_i(t-1) \ge \ghat + \ghat^{D_i} \ge \phi - \frac{\delta}{4} 
\end{align}

Now if we do not violate the delay rule, then we have 
\begin{align*}
    N_i(t-1) & \leq \frac{\ghat_i}{1-\ghat_i^{D_i}} \\ 
             & = \ghat_i + \ghat_i^{D_i} - h(\ghat_i, D_i) \\ 
             &  \leq \ghat_i + \ghat_i^{D_i} - \delta && \text{(by Lemma \ref{lemma:min_seperation})}\\
            &  \leq \phi + \frac{\delta}{4} - \delta && \text{(by the $\frac{\delta}{4}$ additive approximation)} \\
            & = \phi - \frac{3\delta}{4}
\end{align*}
\end{proof}

Based on the above we reach the following fact
\begin{fact}\label{fact:discounted_sep_genmu_second_ver}
For any arm $i \in [k]$ with $D_i \ge 2$ if at round $t-1$ the delay rule was violated then $\genmu{i}(t) \ge \genmuu = 0.4525 - \frac{0.15}{4} \delta$. On the other hand, if the delay rule had not violated at any round up to and including $t-1$ then we have $\genmu{i}(t) \leq  \genmul < \genmuu$ where $\genmul= 0.4525 - \frac{0.45}{4} \delta$. 
\end{fact}
\begin{proof}
If the delay rule for an arm $i \in \arms$ is violated at round $t-1$ then in the next round we have
\begin{align*}
    \genmu{i}(t) & = \mu_i - g(N_i(t-1)) \\ 
                 & \ge \mu_i - g(\phi-\frac{\delta}{4}) && \text{(based on Fact \ref{fact:sample_seperation_ver_2})} \\ 
                 & = 1 - [0.6 - 0.15 (\phi - \frac{\delta}{4})] \\ 
                 & = 0.4525 - \frac{0.15}{4}  \delta
\end{align*}
On the other hand, if the delay rule had not violated at any round up to and including $t-1$ then we have
\begin{align*}
   \genmu{i}(t) & = \mu_i - g(N_i(t-1)) \\  
                & \leq \mu_i - g(\phi -\frac{3 \delta}{4}) && \text{(based on Fact \ref{fact:sample_seperation_ver_2})} \\ 
                 & = 1 - [0.6 - 0.15 (\phi - \frac{3 \delta}{4})] \\ 
                 & = 0.4525 - \frac{0.45}{4}  \delta
\end{align*}
\end{proof}

Note that the human will choose arms of mean either $1.0$ (first $k$) or $0.6$ (the last arm). Further, the maximum generative mean value is at least $1.0-g_i(0)=1.0-0.6=0.4$ and is at most $1.0-g_i(1) \leq 1.0 - [0.6-0.15]= 0.55$ (since the highest number of discounted samples is $\frac{\gamma_i}{1-\gamma_i} \leq \frac{1/2}{1-\frac{1}{2}}=1$ since $\gamma_i \leq \frac{1}{2}$). Therefore, it is sufficient to define the link function over $[0.4,1.0]\times[0.4,1.0]$.

We will now set the link function as follows: (A) if $\mu_1=0.4$ then $\probc(0.4,\mu_2) =\frac{1+(0.4-\mu_2)}{2}$. (B) if $\mu_1=1.0$ then we the link function is 
 \begin{align}
    \probc(1.0,\mu_2) = \begin{cases} 
         \frac{1+(1-\mu_2)}{2} & \text{if } \mu_2 \ge  \genmul \\
         \frac{1+(1-\genmul)}{2} & \text{if } \mu_2 <  \genmul 
\end{cases}
\end{align}

(C) for $\mu_1 \in (0.4,1.0)$ the link function is as follows 
 \begin{align}
    \probc(\mu_1,\mu_2) = \begin{cases} 
         \frac{1}{2} (\mu_1 - 0.4) + \probc(0.4,\mu_2) & \text{if } \mu_1 \leq  \mu_2 \\ \\ 
         \frac{\probc(1,\mu_2)-\frac{1}{2}}{1-\mu_2} (\mu_1-\mu_2)+ \frac{1}{2} & \text{if } \mu_1 >  \mu_2
\end{cases}
\end{align}
It is straightforward to see that $\probc$ is continuous, $\probc(\mu,\mu)=\frac{1}{2}$ and that it is increasing in the first argument.

Based on the above we can upper bound the utility the last arm ($k+1$) can give
\begin{align}
    \probc(\mu_{k+1},0.4) - c = \probc(0.6,0.4) - c = \frac{\probc(1,0.4)-\frac{1}{2}}{1-0.4} (0.6-0.4)+ \frac{1}{2} - c \leq = 0.76458 - 0.77 < 0  
\end{align}
Note in the above that we used the fact that $\genmul=0.4525-\frac{0.45}{4} \delta \ge 0.4525 - \delta \ge 0.4525 - 0.04$ since $\delta = \frac{5}{6} 0.2^{\dmax} \leq 0.2^{\dmax} \leq 0.2^2=0.04$. 

Therefore, the last arm always gives negative utility.

We now can separate the utility we gain from the first $k$ arms based on whether we violated their delay rule or not 
\begin{fact}\label{fact:seperate_util_ver_2}
If the delay rule for any arm arm $i \in [k]$ with $D_i \ge 2$ is violated at round $t-1$ then the utility at the next round will be at most $\thetal-c$ whet $\thetal= \frac{1.5475 + \frac{0.15}{4}\delta}{2}$. On the other hand, if the delay rule is not violated then we can obtain a maximum utility of exactly $\thetau-c>0$ where $\thetau=\frac{1.5475 + \frac{0.45}{4}\delta}{2} > \thetal$.  
\end{fact}
\begin{proof}
The calculations follows from Fact \ref{fact:discounted_sep_genmu_second_ver}. 

If the delay rule is violated then the utility will be at most $\frac{1+(1-\genmuu)}{2} - c = \frac{1+(1-(0.4525 - \frac{0.15}{4}\delta))}{2} -c = \thetal -c$ where $\thetal =  \frac{1.5475 + \frac{0.15}{4}\delta}{2}$ as defined. 

On the other hand if the delay rule is not violated then the utility will be exactly $\frac{1+(1-\genmul)}{2} - c = \frac{1+(1-(0.4525-\frac{0.45}{4}\delta))}{2} - c = \thetau-c$ where $\thetau= \frac{1.5475 + \frac{0.45}{4}\delta}{2}$ as defined. 

It is straightforward to see that $\thetau-c>0$ and that $\thetau > \thetal$.  
\end{proof}

From the above, similar to the proof of Lemma \ref{th:NP_hard_same_gamma} the human can obtain a utility of $(\thetau-c) \cdot (T+1)$ if and only if he pulls from the set of first $k$ arms without violating their delay rule. Therefore, an instance of \maxreward{} is a YES instance if and only if we can obtain a utility of $(\thetau-c) \cdot (T+1)$. 
\end{proof}

%% file: approximation_proofs.tex
\subsection{Proofs for Subsection \ref{subsec:approx_discounting}}\label{app:approximation_proofs}

We restate the lemma and give its proof 
\lemmamyopic*
\begin{proof}
Suppose we follow the same strategy as $\bstar$ from rounds $t$ and ending at $t+(\tauwindow-1)$ but in the first $\tauwindow$ rounds, i.e. starting from round $1$ to round $\tauwindow$ to optimize utility for the first $\tauwindow$ rounds, we will call this strategy $\bprime$. Then it follows that for any arm $i \in \arms$ the new discounted number of pulls $\bn'_i(r)$ satisfies $\bn'_i(r) \leq \bn_i(r+t-1) \ , \forall r \in \{1,2,\dots,\tauwindow\}$. Denote the maximum generative mean by $\maxgenmu{r}'$ and $\maxgenmu{r}$ for strategies $\bprime$ and $\bstar$, respectively. Further, we denote by $b'_r$ and and $b_r$ the arms pulled at round $r$ by $\bprime$ and $\bstar$, respectively.

Then we have the following bound:
\begin{align*}
    \sum_{r=1}^{\tauwindow} \probc(\mu_{b'_r},\maxgenmu{r}') - c_{b'_r} & \ge \sum_{r=1}^{\tauwindow} \probc(\mu_{b'_{r}},\maxgenmu{r+t-1}) - c_{b'_{r}}  \\
    & = \sum_{r=t}^{t+\tauwindow-1} \probc(\mu_{b_{r}},\maxgenmu{r}) - c_{b_{r}} \\ 
    & = u_{t: t+ \tauwindow -1}(\bstar)
\end{align*}
where the above follows since for all $r \in \{1,2,\dots,\tauwindow\}$ we must have $\maxgenmu{r}' \leq \maxgenmu{r+t-1}$ since $\bn'_i(r) \leq \bn_i(r+t-1)$.

Therefore, we have $\myopicopttau \ge \sum_{r=1}^{\tauwindow} \probc(\mu_{b'_r},\maxgenmu{r}') - c_{b'_r} \ge  u_{t: t+ \tauwindow -1}(\bstar)$. 
\end{proof}

Before we give the proof of Theorem \ref{th:optimzethenpause} it is helpful to establish some results. Define $\neps$ to be 
\begin{align}
    \neps = \frac{\gmax}{1-\gmax} \cdot \gmax^{\tauwindow} \label{eq:neps_def}
\end{align} 
Since we set $\tauwindow = \ceil{\log_{(\frac{1}{\gmax})} \Big( \frac{\gmax}{(1-\gmax)^2} \cdot \frac{L_g L_{\probc}}{\uo} \cdot \frac{1}{\epsilon} \Big) }$, the following lemma can be established.
\begin{lemma}\label{lemma:nepsbound}
\begin{align}
    \frac{L_g L_{\probc}}{1-\gmax} \neps \leq \epsilon \myopicopttau
\end{align}
\end{lemma}
\begin{proof}
\begin{align*}
    \frac{L_g L_{\probc}}{1-\gmax} \neps & = L_g L_{\probc} \cdot \frac{\gmax}{(1-\gmax)^2} \cdot \gmax^{\tauwindow} \\
    & \leq L_g L_{\probc} \cdot \frac{\gmax}{(1-\gmax)^2} \cdot \frac{(1-\gmax)^2}{\gmax}  \cdot \frac{\uo}{L_g L_{\probc}} \cdot \epsilon \quad \text{(by plugging the value of $\tauwindow$ in $\gmax^{\tauwindow}$)} \\
    & = \uo \epsilon \\
    & \leq \epsilon  \myopicopttau \quad \text{(since $\myopicopttau \ge \myopicoptone \ge \uo$)}
\end{align*}

\end{proof}

Now we restate the theorem for the approximation ratio and give its proof. 
\optimzethenpause* 
\begin{proof}
Note that the algorithm cycles between running $\bmyopic(\tauwindow)$ and pausing for $\tauwindow$ many rounds. 
Let $i_0$ be the arm that has the maximum number of discounted pulls at the end of round $\tauwindow$ when running $\bmyopic(\tauwindow)$ and denote that number of pulls by $\bn_0$. We will show that at the beginning of any cycle the number of pulls for any arm will be small. Specifically, we have
\begin{claim}
At the beginning of any cycle at round $t$ we have $\forall i \in \arms$ we have $\bn_i(t-1) \leq \neps$. 
\end{claim}
\begin{proof}
Arms only receive pulls when $\bmyopic(\tauwindow)$ is active. It follows that after the end of $\bmyopic(\tauwindow)$ in the first cycle we have
\begin{align}
    \bn_0 \leq \sum_{r=1}^{\tauwindow} \gamma^r_{\max} = \frac{\gmax}{1-\gmax} \cdot (1-\gmax^{\tauwindow}) \leq \frac{\gmax}{1-\gmax} \cdot (1-\gmax^{2\tauwindow})
\end{align}
At the beginning of the $m^{th}$ cycle we would have an upper bound for the discounted number of pulls of (see Figure \ref{fig:discounted_pulls_bound} for an illustration)
\begin{align*}
    \bn_0 \cdot \big(\gamma^{\tauwindow}_{\max} + \gamma^{3 \tauwindow}_{\max} + \gamma^{5 \tauwindow}_{\max} + \dots \big) & = \bn_0 \sum_{\ell=0}^{\infty} \gmax^{\tauwindow (2\ell+1)} \\
    & = \bn_0 \frac{\gmax^{\tauwindow}}{1-\gmax^{2\tauwindow}} \\
    & \leq \frac{\gmax}{1-\gmax} \cdot \gmax^{\tauwindow} \\
    & \leq \neps && \text{(based on the definition of $\neps$ in \eqref{eq:neps_def})}
    \end{align*}
\end{proof}

\begin{figure}
    \centering
    \includegraphics[width=1.0\linewidth]{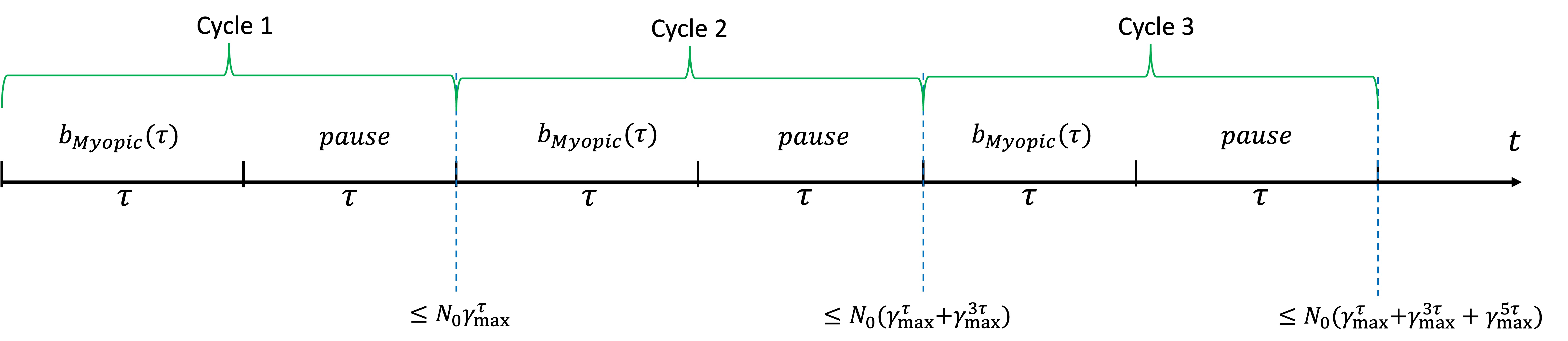}
    \caption{Illustration of the upper bound on the discounted number of contents at the end/beginning of each cycle.}
    \label{fig:discounted_pulls_bound}
\end{figure}

Now we establish that $\bmyopic(\tauwindow)$ is approximately optimal if $\forall i \in \arms$ we have $\bn_i(t) \leq \neps$.
\begin{claim}
Suppose we run the $\bmyopic(\tauwindow)$ starting from a state where the number of pulls for all arms is at most $\neps$ then the utility we obtain is at least $(1-\epsilon) \cdot \myopicopttau$.     
\end{claim}
\begin{proof}
We will add a prime $'$ to denote a quantity when the number of pulls is starting from at most $n_\epsilon$, if it starts from $0$ then we omit the prime $'$. If we follow strategy $\bmyopic(\tauwindow)$ then at any round $r$ in the start of the cycle we would have $\bn'_i(r) \leq \bn_i(r) + \neps \gamma^{r-1}$. It follows by the Lipshitz property of $g_i(.)$ and the fact that it is a decreasing function that $g_i(\bn'_i(r)) \ge g_i(\bn_i(r)) - L_g \neps \gamma^{r-1}$. Therefore, the maximum generative mean increases by at most $L_g \neps \gamma^{r-1}$, i.e. $\maxgenmu{r}' \leq \maxgenmu{r} + L_g \neps \gamma^{r-1}$.  Further, it follows by the Lipshitz property of $\probc$ that in the first half of the cycle the utility is : 
\begin{align*}
    \sum_{r=1}^{\tauwindow} \probc(\mu_{b_r},\maxgenmu{r}') - c_{b_r} & \ge \sum_{r=1}^{\tauwindow} \probc(\mu_{b_r},\maxgenmu{r} + L_g \neps \gamma^{r-1}) - c_{b_r} \\
    & \ge \Big[\sum_{r=1}^{\tauwindow} \probc(\mu_{b_r},\maxgenmu{r} ) - c_{b_r}\Big] - \Big[ \sum_{r=1}^{\tauwindow}  L_g L_{\probc} \neps \gamma^{r-1} \Big] \\
    & \ge \Big[\sum_{r=1}^{\tauwindow} \probc(\mu_{b_r},\maxgenmu{r} ) - c_{b_r}\Big] -  L_g L_{\probc} \neps \frac{1}{1-\gmax} \\ 
    & \ge  \myopicopttau - \epsilon \cdot \myopicopttau && \text{(Based on Lemma \ref{lemma:nepsbound})} \\
    & = (1-\epsilon) \myopicopttau
\end{align*}
\end{proof}
The strategy runs $\bmyopic(\tauwindow)$ for a window of $\tauwindow$ where it obtains a $(1-\epsilon)$ approximation for that window then it pauses for another window of size $\tauwindow$. Therefore, its approximation ratio over the entire horizon $T$ is at least $\frac{1-\epsilon}{2} \frac{\floor{\frac{T}{2 \cdot \tauwindow}}}{\floor{\frac{T}{2 \cdot \tauwindow}}+1}$. 

The run-time to find $\bmyopic(\tauwindow)$ using the exponential DAG of Appendix \ref{app:basic_stg} is $O(k^{\tauwindow})=O(k^{O(\log \frac{1}{\epsilon})})$ and it takes $O(\tau)=O(1)$ memory to save the arm pulls of $\bmyopic(\tauwindow)$. Further, since the algorithm has to make $T$ many decisions for the arm pulls over the whole horizon, we obtain a run-time of $O(T + k^{O(\log \frac{1}{\epsilon})} )$.
\end{proof}

\input{pauseproof}

%% file: pauseproof.tex
We restate the next theorem and give its proof.
\pauseth*
\begin{proof}
Let us assume that the human pulls arms for $\alpha T$ rounds where $\alpha \ge 0.1$. Consider the following instance where we have two arms, i.e., $k=2$ and $\mu_1=\mu_2=\mu$, $g_1(.)=g_2(.)=g(.)$, $c_1=c_2=c$, and $\gamma_1=\gamma_2=\gamma=\frac{1}{2}-\epsilon$ where $\epsilon \in (0,\frac{1}{2})$. 

Since $\gamma = \frac{1}{2} - \epsilon$, then the following is always true for all $\kappa \ge 2$: 
\begin{align}
    \gamma^{\kappa+1} \frac{1}{1-\gamma} < \gamma^{\kappa} \label{eq:half_minus_eps_sep}
\end{align}

\begin{claim}
The discounted number of pulls if an arm was played in the last $\kappa$ rounds is at least $\gamma^{\kappa}$ and if the arm has not been played in the last $\kappa$ rounds then the number of pulls is at most $\gamma^{\kappa+1} \frac{1}{1-\gamma}$.
\end{claim}
\begin{proof}
The first part is immediate. 

The second part holds since if the arm is not played in the last $\kappa$ rounds then the maximum number of pulls at the round proceeding the current round by $\kappa+1$ is at most $\frac{1}{1-\gamma}$ and therefore if we do not play the arm for $\kappa+1$ rounds then the discounted number of pulls is at most $\gamma^{\kappa+1} \frac{1}{1-\gamma}$. 
\end{proof}

The set of rounds where arms are played can be decomposed into $\tless$ and $\tmore$ which are the set of rounds where an arm is pulled after waiting for less than $\kappa$ and $\kappa$ rounds or more, respectively. Clearly, we have:
\begin{align}
    \alpha \cdot T = |\tless| + |\tmore| 
\end{align}
Note that $|\tmore| \leq \frac{1}{\kappa+1} T$ since once if an arm is pulled then it must be that we have at least $\kappa+1$ rounds where no arm was played.  
Therefore, we can lower bound $|\tless|$ as follows
\begin{align}
    |\tless| & = \alpha \cdot T - |\tmore| \\
             & \ge \alpha \cdot T - \frac{1}{\kappa+1} T \\
             & = (\alpha - \frac{1}{\kappa+1}) T 
\end{align}
Now we set $\kappa = \argmin\limits_{r \in \mathbb{Z}_{>0}} \frac{1}{r+1} < \alpha$. 

Let $N_l = \gamma^{\kappa+1} \frac{1}{1-\gamma}$ and let $N_u = \gamma^{\kappa}$. Note that $N_u > N_l$ by \eqref{eq:half_minus_eps_sep}. We set the utility values to the following:
\begin{align}
    \probc(\mu,\maxgenmu{t}) -c & = \nu >0  \ \ \text{if} \ \ N_1(t-1), N_2(t-1) \leq N_l  \\ 
    \probc(\mu,\maxgenmu{t}) -c & = \nu'  \ \ \text{if} \ \ N_1(t-1) \ \text{or} \ N_2(t-1) \ge N_u  \\ 
\end{align}
We set the values of $\nu$ and $\nu'$ such that 
\begin{align}
    \nu' = \frac{-\beta}{(\alpha (\kappa+1)-1)} \nu 
\end{align}
where $\beta>1$. From the above the utility is at most
\begin{align}
    \nu' \cdot  |\tless| + \nu \cdot |\tmore| & \leq  \frac{-\beta}{(\alpha (\kappa+1)-1)} \nu (\alpha - \frac{1}{\kappa+1}) T + \nu \frac{1}{\kappa+1} T \\ 
    & = \nu T \frac{1}{\kappa+1} (1-\beta) \\
    & < 0 
\end{align}
Since $\beta >1$ then the above utility is negative. 

Now we prove the second part. To do that we only need to establish a lower bound on the utility of the optimal algorithm $\OPT$. Consider an algorithm that pulls arm 1 and then pauses for $\kappa$ rounds. Then it must achieve a utility of at least $\nu \floor{\frac{T}{\kappa+1}}$. Therefore, $\OPT \ge \nu \floor{\frac{T}{\kappa+1}} >0$. It follows since the utility obtained by the algorithm is negative that the multiplicative approximation ratio is less than zero. 
 
\end{proof}

%% file: long_horizon_proofs.tex
\section{Proof of Theorem \ref{th:long_horizon_max_utility} from Section \ref{sec:long_horizon}}\label{app:long_horizon_proofs}
We follow notation similar to that introduced in Appendix \ref{app:basic_stg}. Therefore, $\genmu{i}(p=r)= \mu_i - g_i(r)$, i.e., the generative mean value of arm $i$ when the human has pulled the arm $r$ many times. While the problem can be clearly solved by finding the longest path in the exponential DAG of Appendix \ref{app:basic_stg}, it would run in exponential time. Therefore, we will reduce the size of the graph.

To reduce the size of the graph we will now identity simplifying structures that an optimal strategy can satisfy. Specifically, we have two lemmas (Lemma \ref{lemma:never_pull_neg} and Lemma \ref{lemma:sync}) that show that a deterministic optimal strategy with simple special properties exists. In the proof of both of lemmas, we will assume that we start from some deterministic strategy and show that we gain at least the same utility by deviating to another deterministic strategy that satisfies the desired properties.  

The first lemma states that an optimal deterministic strategy should only pull arms that give non-negative utility and stop playing if no such arm exists. This verifies that the assumption made in Section \ref{sec:long_horizon} that the human never plays an arm that gives negative utility and would opt-out in case all arms provide negative utility is without loss of generality. 
\begin{restatable}{lemma}{nevepullneg}\label{lemma:never_pull_neg} 
There exists a deterministic optimal strategy which always pulls arms from $i \in \agmgm(t)$ when $|\agmgm(t)| \ge 1$ and stops playing if $|\agmgm(t)| = 0$. 
\end{restatable}
\begin{proof}
The second part of the statement is proved immediately since when $|\agmgm(t)| = 0$ only negative utility can be accumulated so higher utility can be obtained by deviating to a strategy that stops playing.

The first part can be proved using the following claim. 
\begin{claim}
Consider a deterministic strategy which at round $t$ pulls an arm $i \notin \agmgm(t)$ when $|\agmgm(t)| \ge 1$ then it can replaced by a strategy which pulls the same arms before round $t$, pulls an arm from $\agmgm(t)$ in round $t$, and achieves at least the same utility. 
\end{claim}
\begin{proof}
Let the followed strategy be denoted by $\bbold$ and the deviation strategy be denoted by $\bprime$. Note that if $\bbold$ does not pull an arm from $\agmgm(t)$ at round $t$ when $|\agmgm(t)|\ge 1$ then either (1) no arm from $\agmgm(t)$ is pulled again until the end of the horizon or (2) there exists some arm in $\agmgm(t)$ that was not pulled at round $t$ and that will be pulled after round $t$. 

For case (1), if we deviate to a strategy $\bprime$ which pulls the same arms as $\bbold$ up to round $t-1$, then pulls some arm $j \in \agmgm(t)$ at round $t$ and then stops playing we would obtain a higher utility: 
\begin{align}
    \uh_T(\bprime) = \uh_{1:t-1}(\bbold) +  (\probc(\mu_j,\maxgenmu{t}) -c_j)  > \uh_{1:t-1}(\bbold) + \uh_{t:T}(\bbold) = \uh_T(\bbold)
\end{align}
Note in the above that $\uh_{t:T}(\bbold)<0$.

For case (2), let arm $j \in \agmgm(t)$ be the first arm in $\agmgm(t)$ to be pulled after round $t$ and let that round when it is pulled be  $t_j\leq T$. We will deviate to a strategy which pulls the same arms as $\bbold$ up to round $t-1$, then pulls $j$ at round $t$ and then pulls the same arms pulled by $\bbold$ from round $t_j+1$ to $T$ and then stops playing. It follows by the above that for a round $r \in \{t+1,\dots,T-(t_j-t)\}$ strategy $\bprime$ pulls the arm as pulled by $\bbold$ in round $r+(t_j-t)$, see Figure \ref{fig:tape_1} for an illustration. Further, for any round $r$ denote by $\maxgenmu{r}$ and $\maxgenmu{r}'$ the maximum generative mean at that round $r$ under strategies $\bbold$ and $\bprime$, respectively. Note that by construction of $\bprime$ we have $\maxgenmu{r}' \leq \maxgenmu{r}, \forall r \in [T]$. Now we can lower bound the utility under $\bprime$ as follows
\begin{align*}
    \uh_T(\bprime) & = \uh_{1:t-1}(\bbold) + (\probc(\mu_j,\maxgenmu{t}') -c_j) + \sum_{r=t+1}^{T-(t_j-t)} [\probc(\mu_{b'_r},\maxgenmu{r}') - c_{b'_r}] \\ 
              & \ge \uh_{1:t-1}(\bbold) + (\probc(\mu_j,\maxgenmu{t}') -c_j) + \sum_{r=t_j+1}^{T} [\probc(\mu_{b_r},\maxgenmu{r}) - c_{b_r}] \\ 
              & \ge \uh_{1:t-1}(\bbold) + (\probc(\mu_j,\maxgenmu{t_j}) -c_j) + \sum_{r=t_j+1}^{T} [\probc(\mu_{b_r},\maxgenmu{r}) - c_{b_r}] \\ 
              & = \uh_{1:t-1}(\bbold) + u_{t_j:T}(\bbold) \\ 
              & >  \uh_{1:t-1}(\bbold)) + \uh_{t:t_j-1}(\bbold) + \uh_{t_j:T}(\bbold) \qquad \quad \text{(since $u_{t:t_j-1}(\bbold)<0$ by definition of $\agmgm(t)$)}\\ 
              & = \uh_T(\bbold) 
\end{align*}
 
\begin{figure}[h!]
    \centering
    \includegraphics[width=1.0\linewidth]{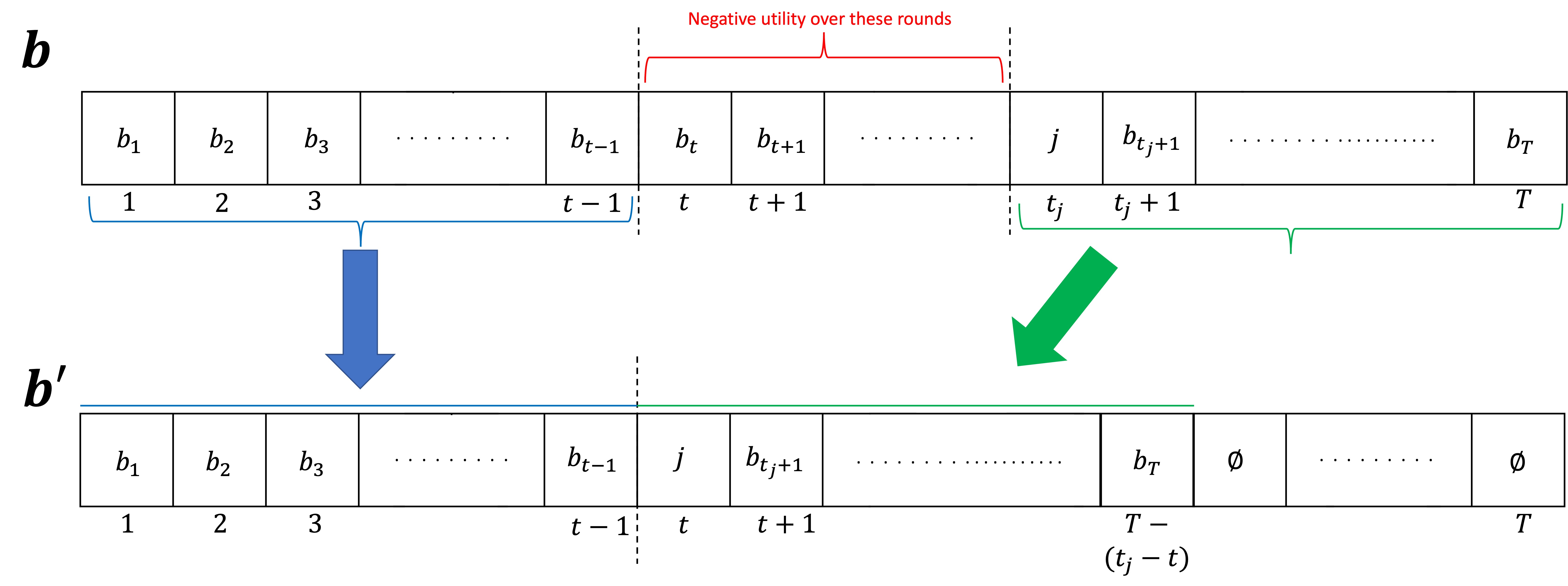}
    \caption{Illustration of the arm pulls under strategy $\bbold$ and the deviation strategy $\bprime$ for case (2) of Lemma \ref{lemma:never_pull_neg}.}
    \label{fig:tape_1}
\end{figure}

\end{proof}
By (recursively) invoking the above claim on any strategy which does not pull an arm $i \notin \agmgm(t)$ when $|\agmgm(t)| \ge 1$ we can obtain a new strategy which achieves at least the same utility that always pulls arms from $\agmgm(t)$.  
\end{proof}

At a given round $t$ define $\hold(t)=\{i \in \agmgm(t) | \genmu{i}(p=N_i(t-1)+1) \leq \maxgenmu{t}\}$, i.e. the set of arms which if pulled would not increase the value of the maximum generative mean. I.e., pulling an arm from $\hold(t)$ would hold $\maxgenmu{t}$ leading to $\Tilde{\mu}_{a^*_{t+1}}=\maxgenmu{t}$. 

We now have the following lemma which essentially states that if the maximum generative mean increases because of a pull, then an optimal deterministic strategy would then ``synchronize'' the generative means of the other arms to make them equal to the current maximum generative mean or slightly below it. 
\begin{restatable}{lemma}{lemmasync}\label{lemma:sync} 
In the long horizon setting there exists an optimal deterministic strategy which at a given round $t$ would always pull an arm $i \in \hold(t)$ when $|\hold(t)| \ge 1$. 
\end{restatable}
\begin{proof}
We will prove the lemma by contradiction. Specifically, we will consider an optimal deterministic strategy $\bbold$ that violates the condition and show that by deviating to a strategy $\bprime$ we would obtain at least the same utility, similar to what we did in Lemma \ref{lemma:never_pull_neg}. 
 
Since we are in the long horizon setting ($T \ge \sum_{i \in \agmgm{(1)}} \nexit{i}$) then by Lemma \ref{lemma:never_pull_neg} there exists a deterministic optimal strategy which always pulls arms from $\agmgm(t)$ and exists (stops playing) at some round $\tau \leq T$. Now consider the following claim: 
\begin{claim}
In the long horizon setting, consider a deterministic strategy which at round $t$ pulls an arm $i \notin \hold(t)$ when $|\hold(t)| \ge 1$,  then it can replaced by another strategy which pulls the same arms before round $t$, pulls an arm from $\hold(t)$ in round $t$, and achieves at least the same utility. 
\end{claim}
\begin{proof}
Suppose that there exists $j \in \hold(t)$ that was not pulled under strategy $\bbold$ at round $t$ then either (1) $j$ is pulled in some following round $t_j>t$ or (2) it is not pulled again. Consider case (1) where $j$ is pulled again. We will deviate to strategy $\bprime$ which pulls the same arms as $\bbold$ up to round $t-1$, then pulls $j$ at round $t$ and then in the rounds from $t+1$ to $t_j$ pulls the same arms as pulled by $\bbold$ from $t$ to $t_j-1$ in the same order, then finally pulls the same arms as $\bbold$ from rounds $t_j +1$ to $\tau$ as shown in Figure \ref{fig:tape_2}. The utility using $\bprime$ can be lower bounded as follows  
\begin{align}
    u_T(\bprime) & = u_{t-1}(\bbold) + (\probc(\mu_j,\maxgenmu{t}) -c_j) + \sum_{r=t+1}^{\tau} [\probc(\mu_{b'_r},\maxgenmu{r}') - c_{b'_r}] \nonumber \\ 
    & = u_{t-1}(\bbold) + (\probc(\mu_j,\maxgenmu{t}) -c_j) + \sum_{r=t+1}^{t_j} [\probc(\mu_{b'_r},\maxgenmu{r}') - c_{b'_r}] + \sum_{r=t_{j+1}}^{\tau} [\probc(\mu_{b'_r},\maxgenmu{r}') - c_{b'_r}] \nonumber \\ 
    & = u_{t-1}(\bbold) + (\probc(\mu_j,\maxgenmu{t}) -c_j) + \sum_{r=t+1}^{t_j} [\probc(\mu_{b'_r},\maxgenmu{r}') - c_{b'_r}] + \sum_{r=t_{j+1}}^{\tau} [\probc(\mu_{b_r},\maxgenmu{r}) - c_{b_r}] \label{eq:3rdline} \\ 
    & = u_{t-1}(\bbold) + (\probc(\mu_j,\maxgenmu{t}) -c_j) + \sum_{r=t}^{t_j-1} [\probc(\mu_{b_r},\maxgenmu{r}) - c_{b_r}] + \sum_{r=t_{j+1}}^{\tau} [\probc(\mu_{b_r},\maxgenmu{r}) - c_{b_r}] \label{eq:4line} \\ 
    & \ge u_{t-1}(\bbold) + \sum_{r=t}^{t_j} [\probc(\mu_{b_r},\maxgenmu{r}) - c_{b_r}] + \sum_{r=t_{j+1}}^{\tau} [\probc(\mu_{b_r},\maxgenmu{r}) - c_{b_r}] \label{eq:5line} \\ 
    & = u_T(\bbold)
\end{align}
In Line \eqref{eq:3rdline} we used the fact that $b'_r=b_r \text{ and } \maxgenmu{r}'=\maxgenmu{r} \ , \forall r\in \{t_j+1,\dots,T\}$. In Line \eqref{eq:4line} we used the fact that $b'_r=b_{r-1} \text{ and } \maxgenmu{r}' = \maxgenmu{r-1} \ , \forall r\in \{t+1,\dots,t_j\}$. In Line \eqref{eq:5line} we used that fact that $\maxgenmu{t} \leq \maxgenmu{t_j}$ and therefore $(\probc(\mu_j,\maxgenmu{t}) -c_j) \ge (\probc(\mu_j,\maxgenmu{t_j}) -c_j)$. 



Now we consider case (2) where arm $j$ is never pulled again. Since we are in the long horizon setting, i.e. $T \ge \sum_{i \in \agmgm{(1)}} \nexit{i}$ it follows that $\tau <T$. To see that, suppose $\tau=T$ then since we assume that $\bbold$ is an optimal strategy then $\tau \leq \sum_{i \in \agmgm{(1)}} \nexit{i}$ and since $T \ge \sum_{i \in \agmgm{(1)}} \nexit{i}$ then $\tau = \sum_{i \in \agmgm{(1)}} \nexit{i}$. The only way for an optimal strategy to have $\tau = \sum_{i \in \agmgm{(1)}} \nexit{i}$ is to pull each arm $i \in \agmgm(1), \ \nexit{i}$ many times. Sine $j \in \hold(t)$ is not pulled at or after round $t$ it follows that $j$ was not pulled $\nexit{j}$ many times and therefore we must have $\tau < \sum_{i \in \agmgm{(1)}} \nexit{i} \leq T$. 

Now that we established that $\tau <T$, the deviation strategy $\bprime$ is simple. Specifically, $\bprime$ pulls the same arms as $\bbold$ up to round $t-1$, then it pulls $j$ at round $t$ and then from rounds $t+1$ to $\tau+1$ it pulls the same arms $\bbold$ pulls from rounds $t$ to $\tau$ in the same order. Since $\probc(\mu_j,\maxgenmu{t}) - c_j \ge 0$ and pulling $j$ does not increase the maximum generative mean it follows that $u_T(\bprime) \ge u_T(\bbold)$

\begin{figure}
    \centering
    \includegraphics[width=1.0\linewidth]{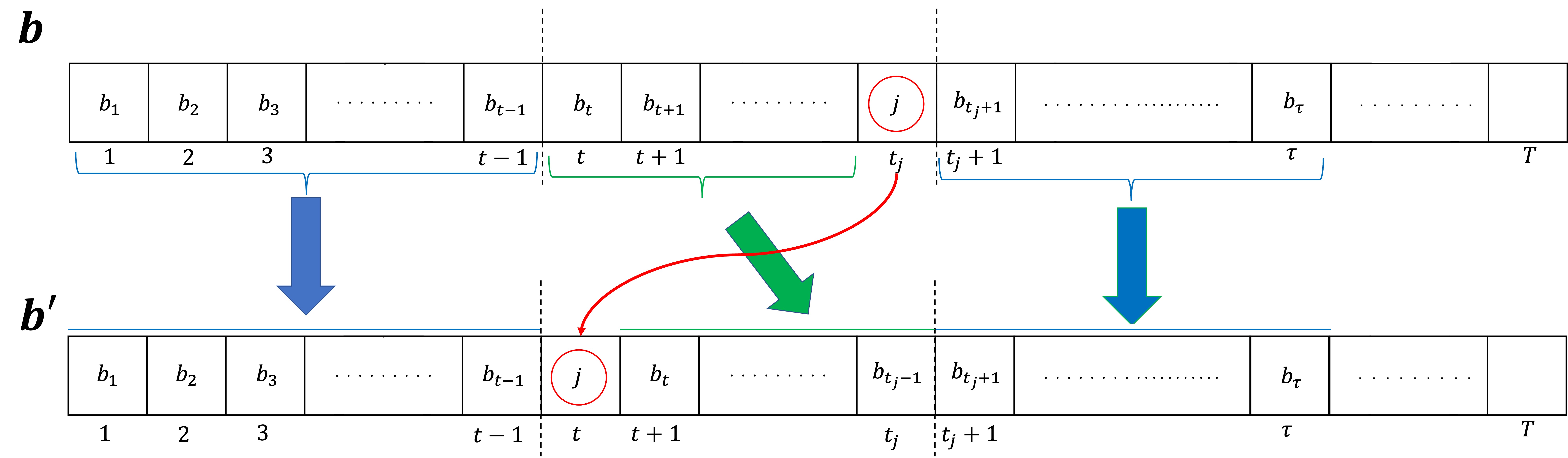}
    \caption{Illustration of the arm pulls under strategy $\bbold$ and the deviation strategy $\bprime$ for case (1) of Lemma \ref{lemma:sync}.}
    \label{fig:tape_2}
\end{figure}
\end{proof}
By the invoking the above claim recursively we can obtain a deterministic optimal strategy which at any round $t$ would pull an arm from $\hold(t)$ if $|\hold(t)|\ge 1$.
\end{proof}

We call a deterministic optimal strategy that follows Lemma \ref{lemma:never_pull_neg} and Lemma \ref{lemma:sync} a \emph{synchronizing strategy}. Such a strategy would always pull from $\agmgm(t)$ in each round $t$ and stop playing if $\agmgm(t)$ is empty, further at any round $t$ it would pull an arm from $\hold(t)$ if it is not empty. In fact, the strategy always alternates between two kinds of pulls: \emph{gain pulls} where the strategy pulls from $\agmgm(t) \setminus \hold(t)$ causing the value of the maximum generative mean to increase and \emph{synchronizing pulls} where the strategy pulls from $\hold(t)$ and the value of the maximum generative mean does not increase. Note that the order of pulls when synchronizing does not change the value of utility gained since the maximum generative mean does not increase. Further, we call a sequence of rounds where a synchronizing strategy does a gain pull followed possibly by synchronizing pulls an \emph{epoch}.  




We will now construct the \emph{reduced} state transition graph which is of polynomial size. Before doing that for each arm $i \in [k]$ we introduce the following function $p^{-1}_i$ defined as follows: \
\begin{align}\label{eq:def_p_inv}
p^{-1}_i(\Tilde{\mu},n_i) =
\begin{cases} 
          0 & \text{if } \Tilde{\mu} \ge \muexit{i} \\
          0 & \text{if } \genmu{i}(p=n_i) > \Tilde{\mu} \\              
          m = \argmax\limits_{r \in \mathbb{Z}_{\ge 0}}  \genmu{i}(p=n_i + r) \leq \Tilde{\mu}  & \text{o.w.}
\end{cases}
\end{align}
Essentially given some value for a generative mean $\Tilde{\mu}$ and the current number of pulls arm $i$ has received $p^{-1}_i(\Tilde{\mu},n_i)$ gives the maximum number of additional pulls $m$ that arm $i$ can receive such that the new value of $i$'s  generative mean does not exceed $\Tilde{\mu}$, i.e., $\genmu{i}(p=n_i+m) \leq \Tilde{\mu}$. On the other hand, if $\genmu{i}(p=n_i) \ge \Tilde{\mu}$ ($i$ already has a higher generative mean) or $\Tilde{\mu} \ge \muexit{i}$ (the value of the generative mean exceeds $i$'s the exit mean) then $p^{-1}_i(\Tilde{\mu},n_i)=0$. 

\paragraph{Reduced Graph Construction:} Algorithm block \ref{alg:form_reduced_dag} gives a formal description for our graph construction algorithm. The graph is constructed so that its nodes and connections encode the gain and synchronizing pulls (epochs) of a synchronizing strategy. Similar to the exponential graph each node has attributes denoting the number of pulls for each arm. The \emph{first node} however, has $(n^1_1=p^{-1}_1(\maxgenmu{1},0),n^1_2=p^{-1}_2(\maxgenmu{1},0),\dots,n^1_k=p^{-1}_k(\maxgenmu{1},0))$ instead of zero pulls. This is done since a synchronizing strategy would always start by pulling an arm $i \in \arms$ a number of times equal $p^{-1}_i(\maxgenmu{1},0)$ to synchronize as will be shown in Claim \ref{cl:firstsync}. Further, we will associate a node $(i,r)$ for arm $i \in [k]$ receiving $r$ many pulls, we call such nodes \emph{gain nodes}. The values for the number of pulls for a gain node $(i,r)$ are set to $\ngain{i}{r}_i=r$ and $\forall j \in [k], j\neq i:\ngain{j}{r}_i=p^{-1}_j(\genmu{i}(p=r),0)$. We only have an outgoing edge from the first node to a node $(i,r)$ if $\muexit{i}> \maxgenmu{1}$ and $\ngain{i}{r}_i=n^1_i+1$. This implies that pulling $i$ is a gain pull, further the edge weight is set not only to equal the utility from pulling $i$ but also the possible following sequence of synchronizing pulls as done in Eq \eqref{eq:util_first}. It will be shown in Claim \ref{cl:bc_epoch} that this encodes the first epoch of a synchronizing strategy. Similarly, we only have an outgoing edge from a node $(i,r)$ to another $(j,r')$ if $\muexit{j} > \genmu{i}(p=r)$ and the number of pulls for $j$ are higher than they are at $(i,r)$ by $1$ (i.e., $\ngain{j}{r'}_j=\ngain{i}{r}_j+1$). This would encode another epoch. Therefore, the edge weight value is again set carefully to encode not only the utility from the gain pull (pulling arm $j$) but also the utility gained by synchronizing as done in Eq \eqref{eq:util_gen}. It will be proved in Claim \ref{cl:inductive_epoch} that this encodes another epoch in a synchronizing strategy. Figure \ref{fig:poly_dag} shows an example of the reduced graph. 

\begin{figure}
    \centering
    \includegraphics[width=.7\linewidth]{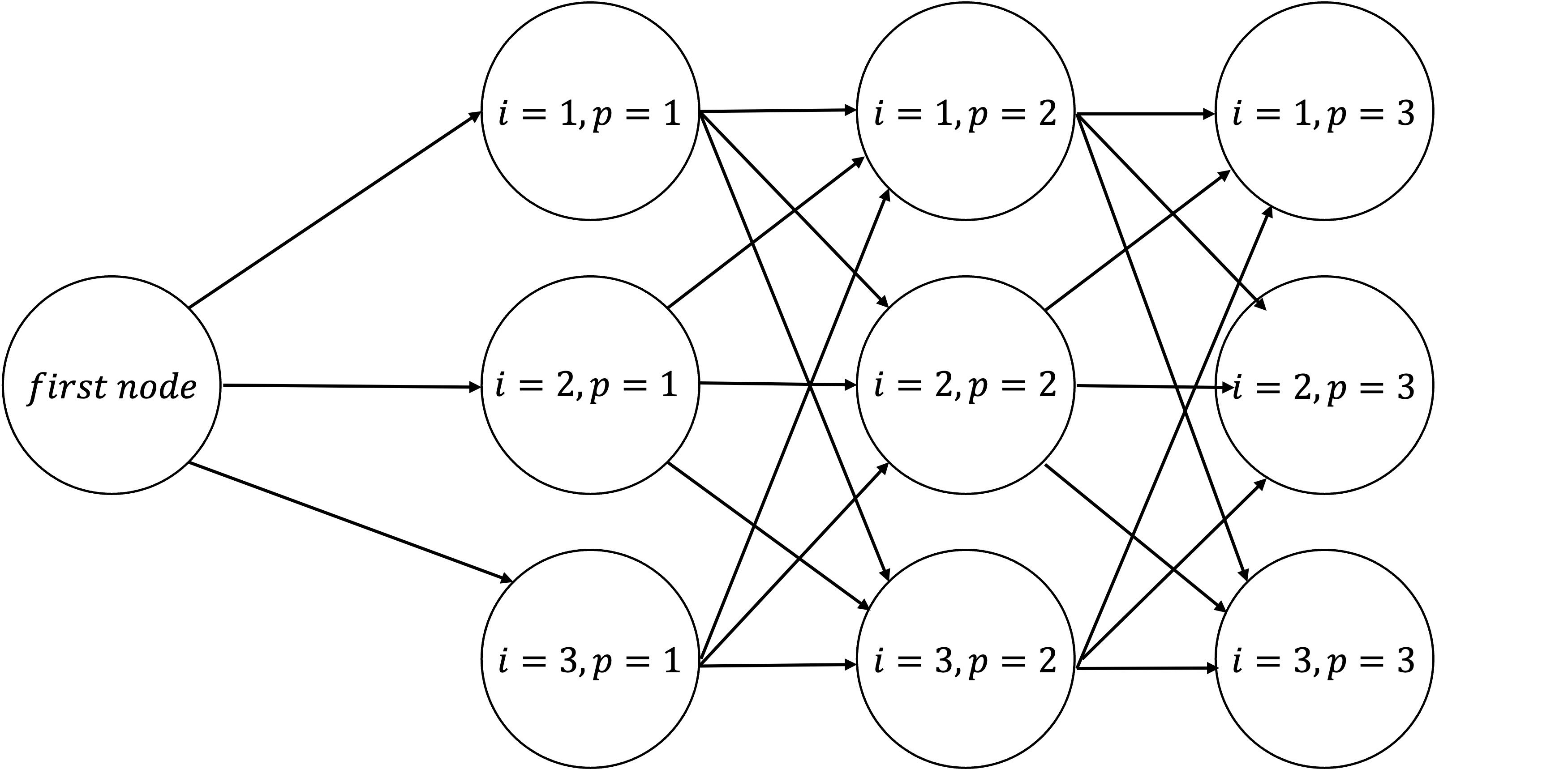}
    \caption{An example of the reduced DAG for $k=3$ and $T=3$. Note that is only an illustrative example. Some edges may not exist, further the connection to other vertices may also be different. But in terms of size the graph would not be larger. Compare this graph to the exponential graph of Figure \ref{fig:dag_exp_app_new}.}
    \label{fig:poly_dag}
\end{figure}


\begin{algorithm}[h!] 
    \caption{\textsc{FormRducedDAG}}\label{alg:form_reduced_dag}
    \begin{algorithmic}[1]
    \STATE \textbf{Input}: $\{\mu_1, \mu_2, \dots, \mu_k\}, \{g_1(.),\dots,g_k(.)\}, \{c_1,\dots, c_k\},\probc$.    
    \STATE \textbf{Output}: Form Reduced Directed Acyclic Graph  
    \item[] 
    \STATE Draw the \emph{first node} $(n^1_1=p^{-1}_1(\maxgenmu{1},0),n^1_2=p^{-1}_2(\maxgenmu{1},0),\dots,n^1_k=p^{-1}_k(\maxgenmu{1},0))$.
    \STATE Draw a gain node $(i,r)$ for each $i \in [k]$ and $r \in [T]$. 
    \STATE For each gain node $(i,r)$ set $\ngain{i}{r}_i=r$ and $\forall j \in [k], j\neq i:\ngain{j}{r}_i=p^{-1}_j(\genmu{i}(p=r),0)$. 
    \STATE Draw an edge from the first node to every node $(i,r)$ only if $\muexit{i}> \maxgenmu{1}$ and $\ngain{i}{r}_i=n^1_i+1$. Further, set the edge weight to 
    \begin{align}
        \probc(\mu_i,\maxgenmu{1})-c_i + \sum_{\ell \neq i} [\probc(\mu_\ell,\maxgenmu{1})-c_\ell] \cdot p^{-1}_\ell(\genmu{i}(p=r),n^1_\ell) \label{eq:util_first} 
    \end{align}
    \STATE Draw an edge from every node $(i,r)$ to $(j,r')$ only if $\muexit{j}> \genmu{i}(p=r)$ and $\ngain{j}{r'}_j=\ngain{i}{r}_j+1$. Further, set the edge weight to 
    \begin{align}
        \probc(\mu_j,\genmu{i}(p=r))-c_j + \sum_{\ell \neq j} [\probc(\mu_\ell,\genmu{j}(p=r'))-c_\ell]\cdot p^{-1}_\ell(\genmu{j}(p=r'),\ngain{i} {r}_\ell) \label{eq:util_gen} 
    \end{align}
    \STATE \textbf{return} Graph 
    \end{algorithmic}
\end{algorithm}

For completeness we give the following Fact about our reduced graph.
\begin{fact}\label{fact:reduced_graph_properties}
The reduced graph has $O(Tk)$ nodes,  $O(Tk^2)$ edges, is acyclic, and can be constructed using Algorithm \ref{alg:form_reduced_dag} in $O(T k^3)$ time. 
\end{fact}
\begin{proof}
We have the first node and $T k$ many gain nodes. Therefore, $O(Tk)$ many nodes. 

Second, since each node has at most $k$ many outgoing edges, it follows that we have $O(Tk^2)$ edges. 

To see that the graph is acyclic note that there is no incoming node into the first node. Further, if there is an edge from $(i,r)$ to $(j,r')$ then it follows that $\ngain{j}{r'}_j = \ngain{i}{r}_j + 1$ and $\ngain{j}{r'}_\ell \ge \ngain{i}{r}_\ell \ \forall \ell \neq j$. Therefore, a path cannot be formed from $(j,r')$ to $(i,r)$. 

Finally, since the edge weight is set according to \eqref{eq:util_first} and \eqref{eq:util_gen} it can be computed in at most $O(k)$ time per edge. Given that we have $O(T k^2)$ edges, the construction time is $O(T k^3)$. 
\end{proof}



We now reach the following theorem.
\begin{restatable}{theorem}{expreduced}\label{th:exp_reduced_same_opt}
Solving the longest path problem starting from the first node in the reduced graph gives an optimal synchronizing strategy. 
\end{restatable}
\begin{proof}


We start with the following claim
\begin{claim}\label{cl:firstsync}
The first node $(n^1_1=p^{-1}_1(\maxgenmu{1},0),n^1_2=p^{-1}_2(\maxgenmu{1},0),\dots,n^1_k=p^{-1}_k(\maxgenmu{1},0))$ in the reduced graph encodes the same number of pulls in the first synchronizing sequence of pulls in any synchronizing strategy. 
\end{claim}
\begin{proof}
It is clear from Lemma \ref{lemma:sync} that each arm in $\hold(1)$ should be pulled a number of times until its maximum generative mean is either equal to $\maxgenmu{1}$ or one pull less than it. If $i \in \hold(1)$ then $p^{-1}_i(\maxgenmu{1},0)$ equals that number of pulls by definition of $p^{-1}_i$, see Eq \eqref{eq:def_p_inv}. 

If $i \notin \hold(1)$ then $i$ will not pulled. Since $\maxgenmu{1}=\max\limits_{j \in \arms} \genmu{j}(p=0)$, it follows that there does not exists an arm $i$ such that $\genmu{i}(p=0) > \maxgenmu{1}$. Therefore, if $i \notin \hold(1)$ then either $\muexit{i} \leq \maxgenmu{1}$ leading to $p^{-1}_i(\maxgenmu{1},0)=0$ or $\genmu{i}(p=1) > \maxgenmu{1}$ also leading to $p^{-1}_i(\maxgenmu{1},0)=0$. The values of $0$ again follows from the definition of $p^{-1}_i$ in Eq \eqref{eq:def_p_inv}.
\end{proof}

Now we can prove the following lemma
\begin{lemma}\label{lemma:sync_to_path}
Any path starting from the first node $(n^1_1=p^{-1}_1(\maxgenmu{1},0),\dots,n^1_k=p^{-1}_k(\maxgenmu{1},0))$ corresponds to a synchronizing strategy. Further, any synchronizing strategy can be represented as a path starting from the first node $(n^1_1=p^{-1}_1(\maxgenmu{1},0),\dots,n^1_k=p^{-1}_k(\maxgenmu{1},0))$.  
\end{lemma}
\begin{proof}
We start by proving the first part of the statement. Denote a path starting from the first node by $\big(v_1,(i_1,r_1),(i_2,r_2),\dots,(i_q,r_q)\big)$. Clearly, $v_1$ denotes the first node. By Claim \ref{cl:firstsync} the first node encodes the first synchronizing set of pulls. We follow a prove by induction and therefore prove the following claim for the base case. 
\begin{claim}\label{cl:bc_epoch}
The edge $\big(v_1,(i_1,r_1)\big)$ encode the pulls of the first epoch in a synchronizing strategy. Further, the weight of the edge $w\big(v_1,(i_1,r_1)\big)$ equals the utility gained and $\forall \ell \in [k]$ we have $\ngain{i_1}{r_1}_\ell$ equals the number of pulls at the end of the first epoch unless $\genmu{i_1}(p=r_1) > \muexit{\ell}$. 
\end{claim}
\begin{proof}
Let $t_1$ be the first round after the end of the synchronizing set of pulls. At the start of $t_1$ each arm has a number of pulls equal to $p^{-1}_i(\maxgenmu{1},0)=n^1_i$. If there is an edge $(v_1,(i_1,r_1))$ then it must be that $\muexit{i_1}> \maxgenmu{1} = \maxgenmu{t_1}$ by construction of the graph. Further, $\genmu{i_1}(p=r_1) > \maxgenmu{1}$ also by construction of the graph. Therefore, $i_1 \in \agmgm{(t_1)} \setminus \hold{(t_1)}$ and pulling it is a gain pull (increasing the maximum generative mean). The utility obtained from pulling arm $i_1$ at round $t_1$ followed by synchronizing would equal
\begin{align}
    \underbrace{\probc(\mu_{i_1},\maxgenmu{1})-c_{i_1}}_{\text{utility for gain pull $i$ which increasing max generative mean}} + \underbrace{\sum_{\ell \neq i_1} [\probc(\mu_\ell,\genmu{{i_1}}(p=r_1))-c_\ell] \cdot p^{-1}_\ell(\genmu{{i_1}}(p=r_1),n^1_\ell)}_{\text{utility for synchronizing at the new value of the max generative mean}}
\end{align}
Note that this is the same value the weight of edge $w\big(v_1,(i_1,r_1)\big)$ is set to by Eq \eqref{eq:util_first}. Further, the pulls of the strategy in the epoch would be $i_1$ followed by the synchronizing pulls according to $p^{-1}_\ell(\genmu{i_1}(p=r_1),n^1_\ell)$ for $\ell \neq i_1$. 

For the last part, if $\muexit{\ell} \leq \genmu{i_1}(p=r_1)$ then $\ell$ would receive a number of pulls equal to $p^{-1}_\ell(\genmu{i_1}(p=r_1),n^1_\ell) + n^1_\ell = p^{-1}_\ell(\genmu{i_1}(p=r_1),0) = \ngain{i_1}{r_1}_\ell$.  
\end{proof}

We now do the inductive step 
\begin{claim}\label{cl:inductive_epoch}
If node $(i_{\kappa},r_{\kappa})$ encodes the number of pulls for each arm $\ell \in \arms$ such that $\muexit{\ell} < \genmu{i_\kappa}(p=r_\kappa)$ at the end of an epoch then $\big((i_{\kappa},r_{\kappa}),(i_{\kappa+1},r_{\kappa+1})\big)$ encodes the pulls of a subsequent epoch. Further, the weight of the edge $w\big((i_{\kappa},r_{\kappa}),(i_{\kappa+1},r_{\kappa+1})\big)$ equals the utility gained and $\forall \ell \in [k]$ we have that $\ngain{i_{\kappa+1}}{r_{\kappa+1}}_\ell$ equals the number of pulls at the end of subsequent epoch unless $\genmu{i_{\kappa+1}}(p=r_{\kappa+1}) > \muexit{\ell}$.
\end{claim}

\begin{proof}
For easier notation we set $\big((i_{\kappa},r_{\kappa}),(i_{\kappa+1},r_{\kappa+1})\big)=\big((i,r),(j,r')\big)$. Since edge $\big((i,r),(j,r')\big)$ exists it follows that $\genmu{i}(p=r) < \muexit{j}$. Further, $\ngain{j}{r'}_{j} = \ngain{i}{r}_{j} +1$ it follows that pulling arm $j$ would increase the value of the generative mean. The utility gained from pulling arm $j$ followed by synchronizing would be 
\begin{align}
    \underbrace{\probc(\mu_j,\genmu{i}(p=r))-c_j}_{\text{utility for gain pull $j$ which increasing max generative mean}} \ + \ \underbrace{\sum_{\ell \neq j} [\probc(\mu_\ell,\genmu{j}(p=r'))-c_\ell]\cdot p^{-1}_\ell(\genmu{j}(p=r'),\ngain{i}{r}_\ell)}_{\text{utility gained by synchronizing at the new value of the max generative mean}}
\end{align}

which is equal to the edge weight $w\big((i,r),(j,r')\big)$. Further, pulls are equal to $j$ followed by the synchronizing pulls according to  $p^{-1}_\ell(\genmu{j}(p=r'),\ngain{i}{r}_\ell), \ \forall \ell \neq j$. 

Finally, if $\muexit{\ell} \leq \genmu{j}(p=r')$ then $\ell$ would receive a number of pulls equal to $p^{-1}_\ell(\genmu{j}(p=r'),\ngain{i}{r}_\ell)+  \ngain{i}{r}_\ell = p^{-1}_\ell(\genmu{j}(p=r'),0) = \ngain{j}{r'}_\ell$.  
\end{proof}
Based on Claims \ref{cl:bc_epoch} and \ref{cl:bc_epoch} it follows that every path starting from the first node in the reduced graph corresponds to a synchronizing strategy. 

By similar arguments it is straightforward to show that every synchronizing strategy corresponds to a path starting from the first node in the reduced graph. 
\end{proof}
Based on Lemma \ref{lemma:sync_to_path} since every path starting from the first node encodes a synchronizing strategy by finding the longest path we would always obtain a synchronizing strategy. Further, since for every synchronizing strategy there exists a path starting from the first node we would obtain an optimal synchronizing strategy by solving the longest path problem with the source being the first node.    
\end{proof}

Based on Fact \ref{fact:reduced_graph_properties} and Theorems \ref{th:dag_solvable_in_poly} and \ref{th:exp_reduced_same_opt} the proof of Theorem \ref{th:long_horizon_max_utility} follows.

%% file: app_experiments.tex
\section{Additional Experimental Results and Details}\label{app:experiments} 
We run a few experiments for both the time-sensitive case and time-insensitive case, showing accumulated utility plots as well as additional data for the time-sensitive case that showcases the strength of enforced pausing. Additionally, we run experiments where the opposing GenAI is allowed to play in a suboptimal manner, showing empirically that our algorithms for both the time-sensitive and time-insensitive case perform well against a generalized opponent (this follows from Proposition \ref{theorem:utility_lower_bound}). For completeness, the full data for the experiments run in Section \ref{sec:experiments} are shown here. For these experiments, the link function $\sigma$ and shrinkage functions $g_i$ take on the same form as in Section \ref{sec:experiments}:
\[
\sigma(\mu_i, \mu_j) := \frac{e^{\mu_i}}{e^{\mu_i} + e^{\mu_j}}
\]
\[
g_i(N) := \frac{q_i}{\sqrt{N + (\frac{q_i}{h_i})^2}}
\]

\subsection{Time-Sensitive Domains}
For this setting, the horizon is again set to $T=1,000$ and the approximation constant for \textbf{Myopic-then-Pause} is set to $\epsilon=0.1$. We conduct a total of 4 experiments in this setting --- for completeness, Experiment $3$ repeats the parameters used for the time-sensitive experiment in Section \ref{sec:experiments}. As can be seen in Figures \ref{fig:time-sensitive-app-1}-\ref{fig:time-sensitive-app-4}, the periods of pausing enforced by the \textbf{Myopic-and-Pause} algorithm allow the generative means of each arm to routinely return close to their original levels, whereas the baseline strategies (in particular \textbf{BT-Pull}, \textbf{Cycle}, and \textbf{Greedy}) prioritize consistent pulling at the expense of long-term accumulated utility. This highlights the benefit of enforced pausing. Note that the timescales for the graphs of the generative means are set to $50$ rounds so that the dynamics of the generative means can be seen with greater clarity. The details of the instances are discussed below. Table \ref{table:ts_experiments} shows the final total utility $u_T$ gained by each algorithm. 

\begin{table}[h!]
    \centering
    \begin{tabular}{|c|c|c|c|c|c|}
    \hline
    & Myopic-and-Pause & Pure-Myopic & Greedy & Cycle & BT-Pull \\
    \hline
    Experiment 1 & 9.551 & 3.125 & 2.356 & 2.373 & 3.611\\
    \hline
    Experiment 2 & 5.051 & 0.823 & 1.088 & 1.074 & 1.077\\
    \hline
    Experiment 3 & 9.481 & 7.039 & 1.058 & 1.280 & 1.113\\
    \hline
    Experiment 4 & 6.858 & -2.890 & 4.013 & 4.011 & 2.940\\
    \hline
    \end{tabular}
    \caption{Final total utilities for each experiment in the time-sensitive scenario.}
    \label{table:ts_experiments}
\end{table}

\textbf{Experiment 1:} We have $3$ arms here and the parameters are set as follows: $\muvec = [.7, .85, .95]$, $\gammavec = [.5, .48, .44]$, $\cvec = [.53, .56, .6]$, $\qvec = [.1, .15, .5]$, $\hvec = [.4, .45, .6]$. Accumulated utility and generative mean plots are given in Figure \ref{fig:time-sensitive-app-1}.


\begin{figure}[h!]
    \centering
    \begin{subfigure}[b]{0.48\textwidth}
        \centering
        \includegraphics[width=\linewidth]{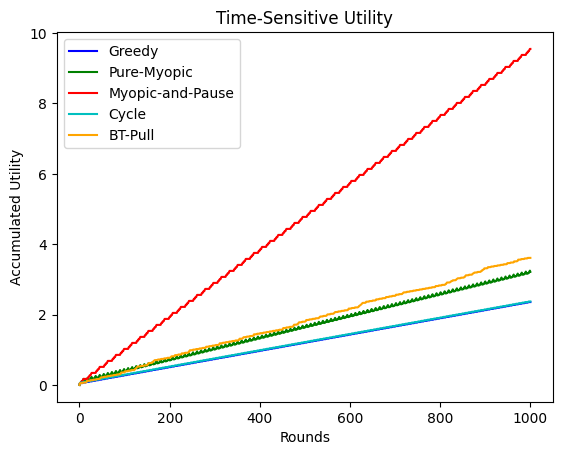}
        \caption{Accumulated utility over time}
    \end{subfigure}
    \hfill
    \begin{subfigure}[b]{0.48\textwidth}
        \centering
        \includegraphics[width=\linewidth]{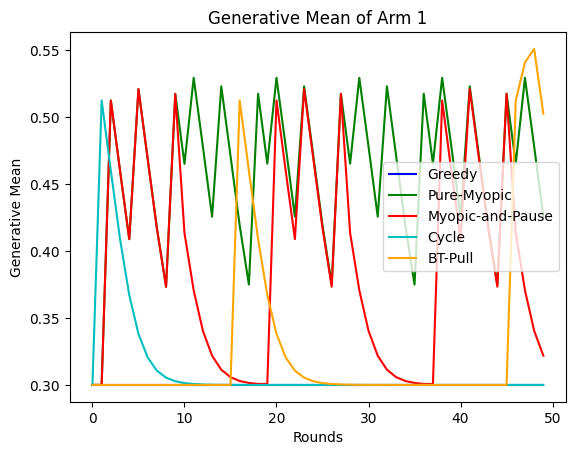}
        \caption{Generative mean of arm 1}
    \end{subfigure}
    \vspace{0.5cm}
    
    \begin{subfigure}[b]{0.48\textwidth}
        \centering
        \includegraphics[width=\linewidth]{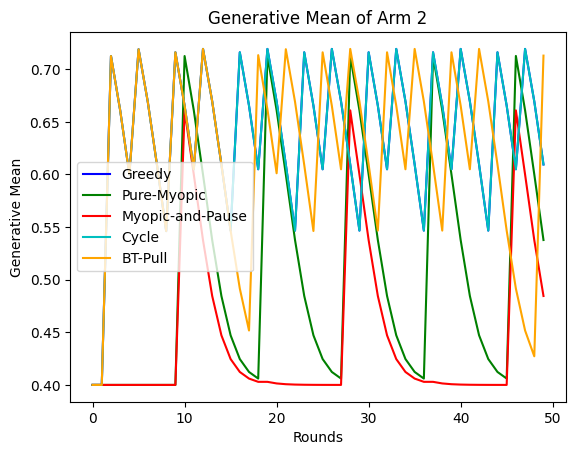}
        \caption{Generative mean of arm 2}
    \end{subfigure}
    \hfill
    \begin{subfigure}[b]{0.48\textwidth}
        \centering
        \includegraphics[width=\linewidth]{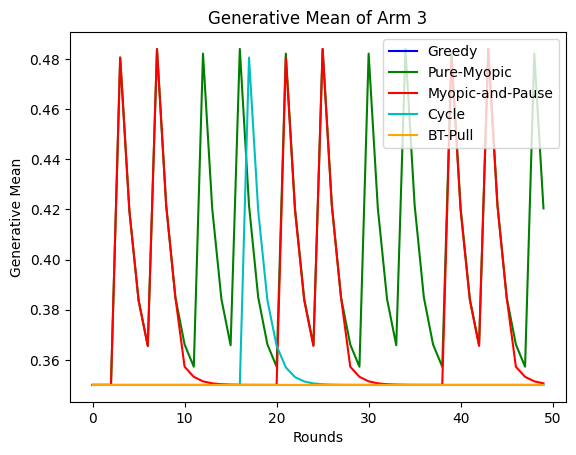}
        \caption{Generative mean of arm 3}
    \end{subfigure}
    \caption{Figure shows the accumulated utilities and generative means for \textbf{Experiment 1}. Plot (a) shows the accumulated utility over time for our algorithm and the relevant baselines. Plots (b), (c), and (d) show the dynamics of the generative means throughout the first 50 rounds.}
    \label{fig:time-sensitive-app-1}
\end{figure}



\textbf{Experiment 2:} We also have $3$ arms here and the parameters are set as follows: $\muvec = [.75, .8, .85]$, $\gammavec = [.5, .55, .6]$, $\cvec = [.55, .55, .55]$, $\qvec = [.15, .125, .1]$, $\hvec = [.4, .4, .4]$. Accumulated utility and generative mean plots are given in Figure \ref{fig:time-sensitive-app-2}.

\begin{figure}[h!]
    \centering
    \begin{subfigure}[b]{0.48\textwidth}
        \centering
        \includegraphics[width=\linewidth]{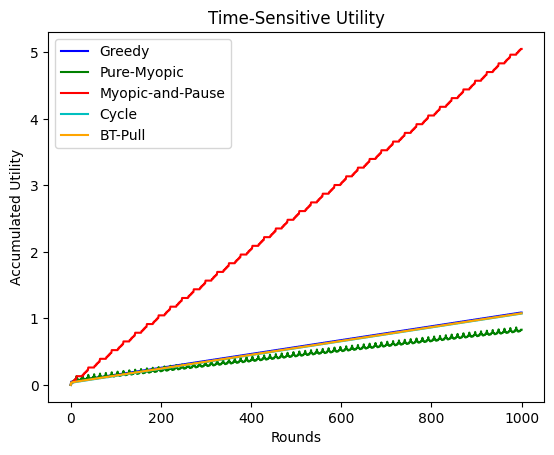}
        \caption{Accumulated utility over time}
    \end{subfigure}
    \hfill
    \begin{subfigure}[b]{0.48\textwidth}
        \centering
        \includegraphics[width=\linewidth]{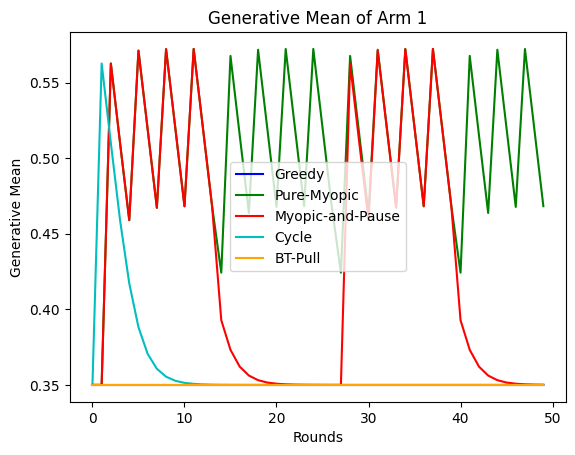}
        \caption{Generative mean of arm 1}
    \end{subfigure}
    \vspace{0.5cm}
    
    \begin{subfigure}[b]{0.48\textwidth}
        \centering
        \includegraphics[width=\linewidth]{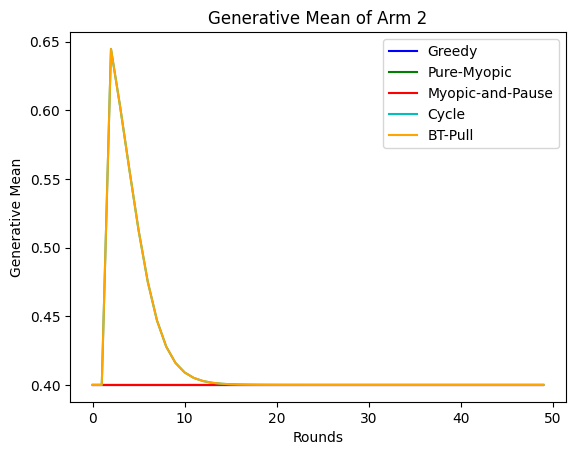}
        \caption{Generative mean of arm 2}
    \end{subfigure}
    \hfill
    \begin{subfigure}[b]{0.48\textwidth}
        \centering
        \includegraphics[width=\linewidth]{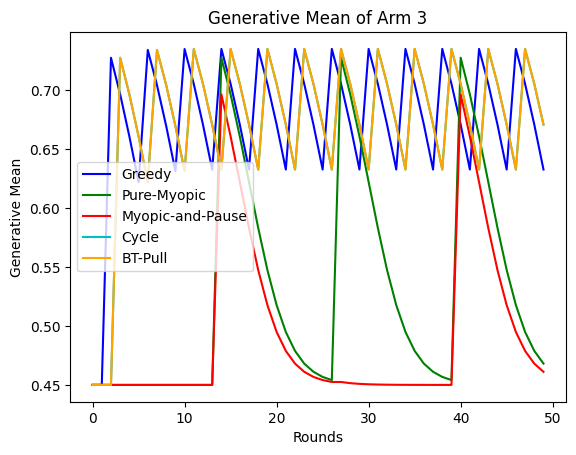}
        \caption{Generative mean of arm 3}
    \end{subfigure}
    \caption{Figure shows the accumulated utilities and generative means for \textbf{Experiment 2}. Plot (a) shows the accumulated utility over time for our algorithm and the relevant baselines. Plots (b), (c), and (d) show the dynamics of the generative means throughout the first 50 rounds.}
    \label{fig:time-sensitive-app-2}
\end{figure}


\textbf{Experiment 3:} We repeat the experiment done in Section \ref{sec:experiments} to showcase the generative means of each arm and record the accumulated utilities of each algorithm. Parameters are set as follows: $\muvec= [.73, .85, .9, .95],\gammavec = [.5, .48, .47, .45]$, $\cvec = [.53, .56, .59, .58]$, $\qvec = [.1, .2, .5, .2]$, and $\hvec= [.3, .45, .6, .6]$. Accumulated utility and generative mean plots are given in Figure \ref{fig:time-sensitive-app-3}.

\begin{figure}[h!]
    \centering
    \begin{subfigure}[b]{0.48\textwidth}
        \centering
        \includegraphics[width=\linewidth]{4topics_appendix_util_1.png}
        \caption{Accumulated utility over time}
    \end{subfigure}
    \hfill
    \begin{subfigure}[b]{0.48\textwidth}
        \centering
        \includegraphics[width=\linewidth]{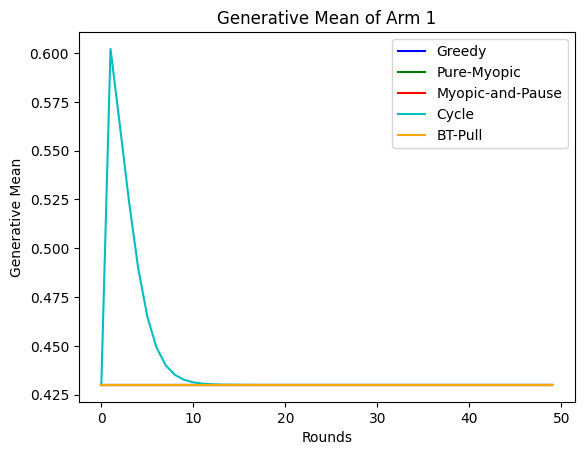}
        \caption{Generative mean of arm 1}
    \end{subfigure}
    \vspace{0.5cm}
    
    \begin{subfigure}[b]{0.48\textwidth}
        \centering
        \includegraphics[width=\linewidth]{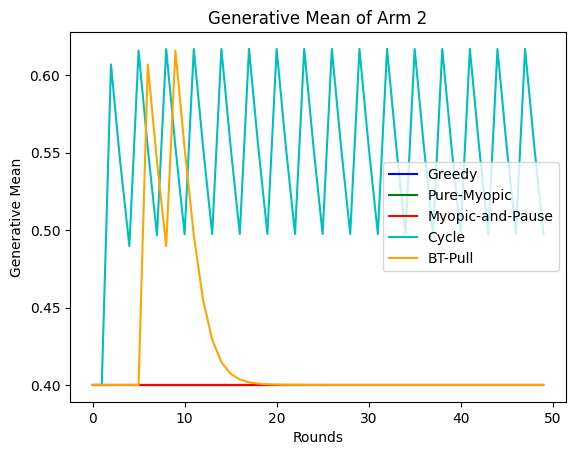}
        \caption{Generative mean of arm 2}
    \end{subfigure}
    \hfill
    \begin{subfigure}[b]{0.48\textwidth}
        \centering
        \includegraphics[width=\linewidth]{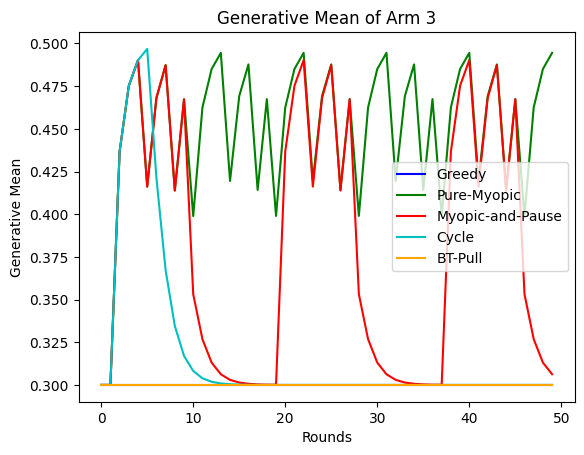}
        \caption{Generative mean of arm 3}
    \end{subfigure}
    \vspace{0.5cm}
    
    \begin{subfigure}[b]{0.48\textwidth}
        \centering
        \includegraphics[width=\linewidth]{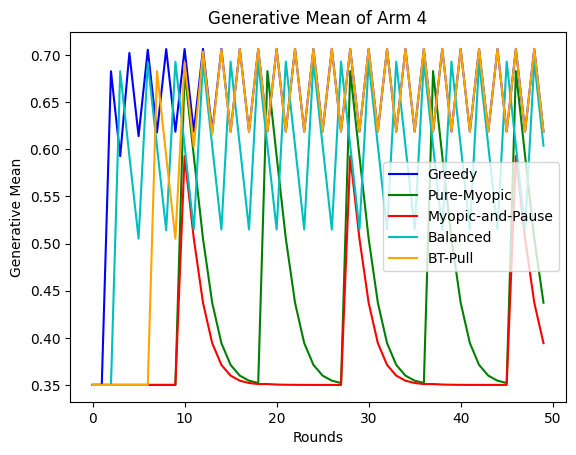}
        \caption{Generative mean of arm 4}
    \end{subfigure}
    
    \caption{Figure shows the accumulated utilities and generative means for \textbf{Experiment 3}. Plot (a) shows the accumulated utility over time for our algorithm and the relevant baselines. Plots (b), (c), (d), and (e) show the dynamics of the generative means throughout the first 50 rounds.}
    \label{fig:time-sensitive-app-3}
\end{figure}


\textbf{Experiment 4:} We have 4 arms and the parameters are set as follows: $\muvec = [.64, .95, .8, .88]$, $\gammavec = [.5, .47, .5, .4]$, $\cvec = [.53, .57, .54, .56]$, $\qvec = [.16, .15, .125, .1]$, $\hvec = [.4, .5, .4, .4]$. Accumulated utility and generative mean plots are given in Figure \ref{fig:time-sensitive-app-4}.

\begin{figure}[h!]
    \centering
    \begin{subfigure}[b]{0.48\textwidth}
        \centering
        \includegraphics[width=\linewidth]{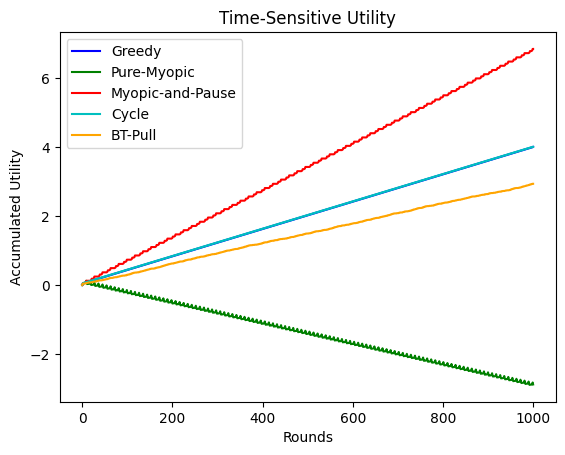}
        \caption{Accumulated utility over time}
    \end{subfigure}
    \hfill
    \begin{subfigure}[b]{0.48\textwidth}
        \centering
        \includegraphics[width=\linewidth]{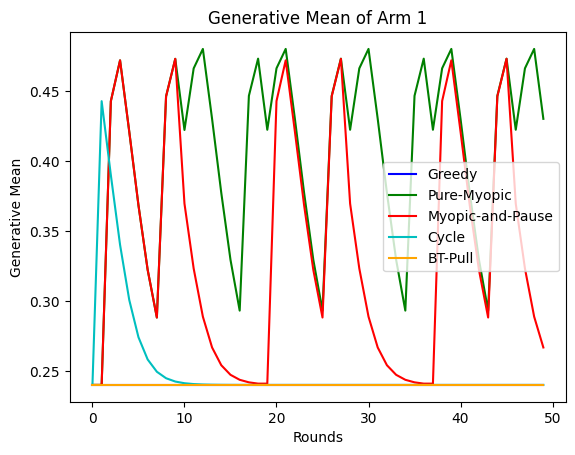}
        \caption{Generative mean of arm 1}
    \end{subfigure}
    \vspace{0.5cm}
    
    \begin{subfigure}[b]{0.48\textwidth}
        \centering
        \includegraphics[width=\linewidth]{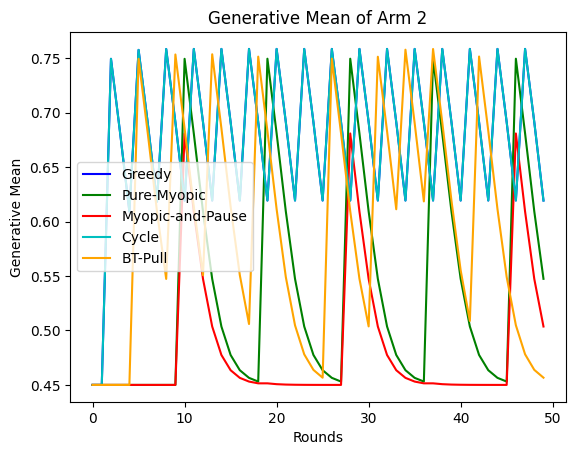}
        \caption{Generative mean of arm 2}
    \end{subfigure}
    \hfill
    \begin{subfigure}[b]{0.48\textwidth}
        \centering
        \includegraphics[width=\linewidth]{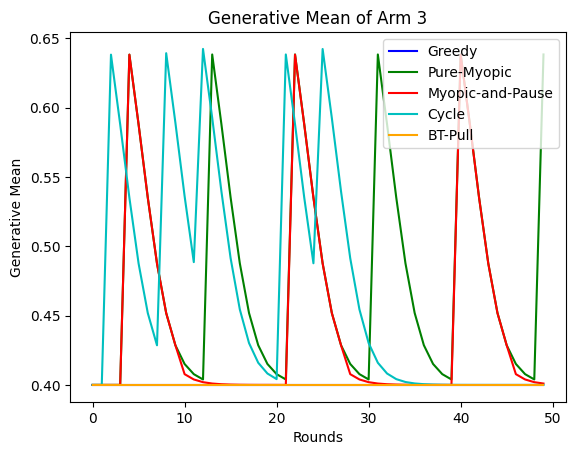}
        \caption{Generative mean of arm 3}
    \end{subfigure}
    \vspace{0.5cm}
    
    \begin{subfigure}[b]{0.48\textwidth}
        \centering
        \includegraphics[width=\linewidth]{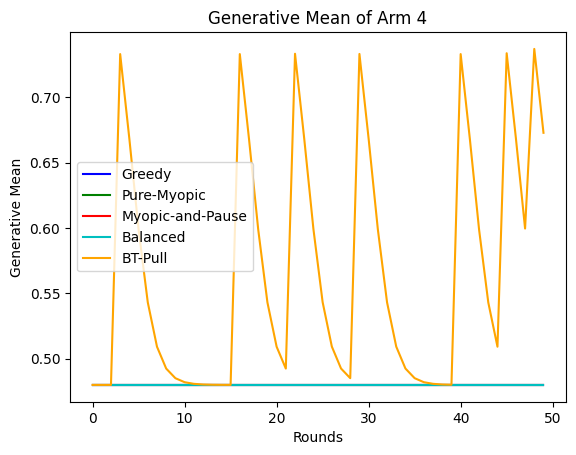}
        \caption{Generative mean of arm 4}
    \end{subfigure}
    
    \caption{Figure shows the accumulated utilities and generative means for \textbf{Experiment 4}. Plot (a) shows the accumulated utility over time for our algorithm and the relevant baselines. Plots (b), (c), (d), and (e) show the dynamics of the generative means throughout the first 50 rounds.}
    \label{fig:time-sensitive-app-4}
\end{figure}


\clearpage
\newpage 
\subsection{Time-Insensitive Domains}
For this setting, the length of the horizon is set to $T \ge \sum_{i \in \agmgm(1)} \nexit{i}$ (the long horizon value) as mentioned in Definition \ref{def:long horizon}. The \textbf{Greedy-Delay} strategy consistently performs well in these experiments, accumulating utility within a fraction of a percent of the optimal for longer games. However, it is worth noting the value of the \textbf{OPT} algorithm since (1) it allows us to show empirically that \textbf{Greedy-Delay} performs near-optimally, and (2) it achieves the true optimum in settings where \textbf{Greedy-Delay} may perform badly, such as the one given in Section \ref{app:long_horizon_greedy_subopt} of the Appendix. We conduct a total of 3 experiments in this setting. The details of the instances are discussed below. Table \ref{table:time-insensitive-appendix} shows the final total utility $u_T$ gained by each algorithm. 

\begin{table}[h!]
    \centering
    \begin{tabular}{|c|c|c|c|c|c|}
    \hline
    & OPT & Greedy-Delay & Greedy & Cycle & BT-Pull \\
    \hline
    Experiment 1 & 10.9460 & 10.9459 & 4.2894 & 7.1552 & 4.7460\\
    \hline
    Experiment 2 & 0.328 & 0.319 & 0.184 & 0.220 & 0.274\\
    \hline
    Experiment 3 & 8.2303 & 8.2301 & 1.8745 & 2.6974 & 4.0030\\
    \hline
    \end{tabular}
    \caption{Accumulated utilities for each experiment in the time-insensitive scenario.}
    \label{table:time-insensitive-appendix}
\end{table}

\textbf{Experiment 1:} For completeness, we record the full results from the time-insensitive experiment in Section \ref{sec:experiments}. Parameters are set as follows: $\muvec=[.86, .95, .87, .95]$, $\cvec=[.56, .52, .55, .53]$, $\qvec = [1.4, 1.3, 1.5, 1.8]$, $\hvec = [.75, .8, .6, .78]$. Figure \ref{fig:time-insensitive-1-app} shows the accumulated utility versus time for each algorithm.

\begin{figure}[h!]
    \centering
    \includegraphics[width=0.5\linewidth]{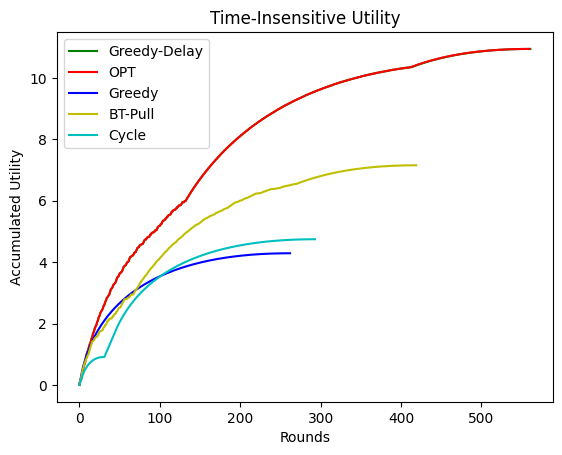}
    \caption{Figure shows the accumulated utility over time for each algorithm in the time-insensitive scenario for \textbf{Experiment 1}.}
    \label{fig:time-insensitive-1-app}
\end{figure}

\textbf{Experiment 2:} Parameters are as set as follows: $\muvec = [.86, .95, .87, .95]$, $\cvec = [.57, .54, .56, .55]$, $\qvec = [.5, .3, .45, .4]$, $\hvec = [.688, .76, .696, .76]$. Figure \ref{fig:time-insensitive-2-app} shows the accumulated utility versus time for each algorithm.

\begin{figure}[h!]
    \centering
    \includegraphics[width=0.5\linewidth]{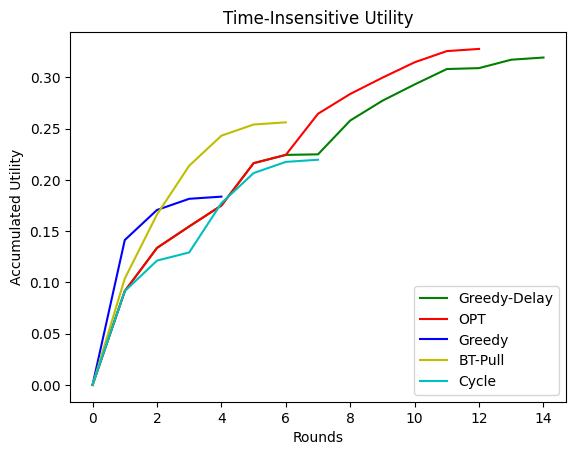}
    \caption{Figure shows the accumulated utility over time for each algorithm in the time-insensitive scenario for \textbf{Experiment 2}.}
    \label{fig:time-insensitive-2-app}
\end{figure}

\textbf{Experiment 3:} Parameters are as set as follows: $\muvec = [.95, .83, .7, .75]$, $\cvec = [.6, .55, .58, .56]$, $\qvec = [3.5, 1.7, 2.2, 3.4]$, $\hvec = [.8, .72, .54, .6]$. Figure \ref{fig:time-insensitive-3-app} shows the accumulated utility versus time for each algorithm.

\begin{figure}[h!]
    \centering
    \includegraphics[width=0.5\linewidth]{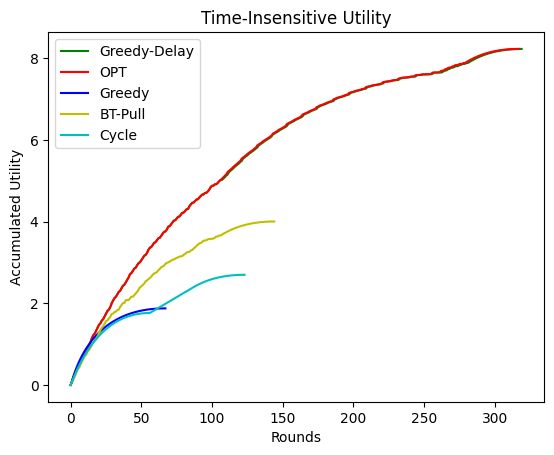}
    \caption{Figure shows the accumulated utility over time for each algorithm in the time-insensitive scenario for \textbf{Experiment 3}.}
    \label{fig:time-insensitive-3-app}
\end{figure}

\clearpage
\newpage 
\subsection{Human Utility under GenAI Deviation}
Lastly, it is interesting to consider the performance of the \textbf{Myopic-and-Pause} and \textbf{OPT} algorithms in a more general scenario where the opposing GenAI may not play the strategy of pulling the arm with the highest generative mean $\astar$. It follows from Proposition \ref{theorem:utility_lower_bound} that for a fixed human strategy, any deviation on the side of the GenAI from the optimal strategy will result in a greater  human utility. 

To show this empirically, we simulate scenarios with a deviating GenAI parameterized by a value $p \in [0, 1]$, where in each round the GenAI chooses the arm with maximum generative mean with probability $p$ and chooses a random one with probability $1-p$. However, the human strategies \textbf{Myopic-and-Pause} and \textbf{OPT} are still calculated under the assumption that the GenAI plays the arm with the maximum generative mean. The scenario with $p=1$ is equivalent to the previous case where the GenAI always plays the arm with the maximum generative mean, and the $p=0$ scenario corresponds to the GenAI playing completely at random. Table \ref{table:deviant_genAI} shows the utility accumulated by the \textbf{Myopic-and-Pause} and \textbf{OPT} algorithms under varying levels of GenAI deviation, and Figure \ref{fig:deviant_genAI} shows the accumulated utilities as a function of the number of rounds. As we would expect the human obtains the lowest utility at $p=1$ and higher utility values for lower $p$. 

The parameters used here are the same as the two examples given in Section \ref{sec:experiments}:

\textbf{Time-Sensitive:} $\muvec= [.73, .85, .9, .95],\gammavec = [.5, .48, .47, .45]$, $\cvec = [.53, .56, .59, .58]$, $\qvec = [.1, .2, .5, .2]$, $\hvec= [.3, .45, .6, .6]$.

\textbf{Time-Insensitive:} $\muvec=[.86, .95, .87, .95]$, $\cvec=[.56, .52, .55, .53]$, $\qvec = [1.4, 1.3, 1.5, 1.8]$, $\hvec = [.75, .8, .6, .78]$.

\begin{table}[h!]
    \centering
    \begin{tabular}{|c|c|c|}
    \hline
    & Myopic-and-Pause (Time-Sensitive) & OPT (Time-Insensitive)\\
    \hline
    $p=1.0$ & 9.481 & 10.946\\
    \hline
    $p=0.9$ & 9.923 & 11.973\\
    \hline
    $p=0.8$ & 10.324 & 12.688\\
    \hline
    $p=0.7$ & 10.701 & 13.536\\
    \hline
    $p=0.6$ & 11.264 & 14.868\\
    \hline
    $p=0.5$ & 11.511 & 15.562\\
    \hline
    $p=0.4$ & 12.025 & 16.502\\
    \hline
    $p=0.3$ & 11.264 & 17.259\\
    \hline
    $p=0.2$ & 12.762 & 17.959\\
    \hline
    $p=0.1$ & 13.053 & 18.532\\
    \hline
    $p=0.0$ & 13.891 & 19.861\\
    \hline
    \end{tabular}
    \caption{Accumulated utilities for the \textbf{Myopic-and-Pause} and \textbf{OPT} algorithms under varying levels of GenAI deviation.}
    \label{table:deviant_genAI}
\end{table}

\begin{figure}[h!]
    \centering
    \begin{subfigure}[b]{0.48\textwidth}
        \centering
        \includegraphics[width=\linewidth]{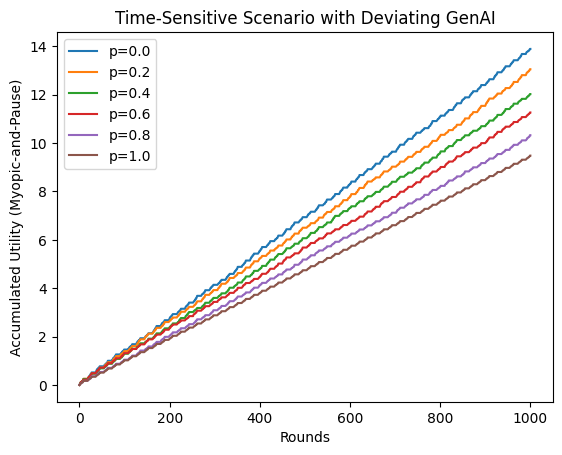}
        \caption{Accumulated utility of \textbf{Myopic-and-Pause}.}
    \end{subfigure}
    \hfill
    \begin{subfigure}[b]{0.48\textwidth}
        \centering
        \includegraphics[width=\linewidth]{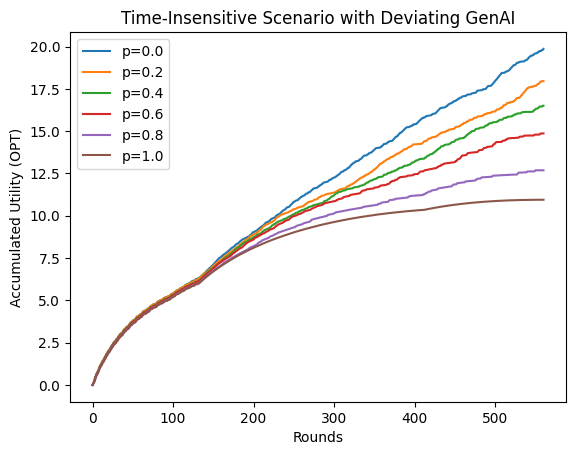}
        \caption{Accumulated utility of \textbf{OPT}.}
    \end{subfigure}
    \caption{Figure shows the accumulated utility over time for the \textbf{Myopic-and-Pause} and \textbf{OPT} algorithms under varying levels of GenAI deviation, parameterized by the value $p$.}
    \label{fig:deviant_genAI}
\end{figure}

%% file: main.bbl
\begin{thebibliography}{42}
\providecommand{\natexlab}[1]{#1}
\providecommand{\url}[1]{\texttt{#1}}
\expandafter\ifx\csname urlstyle\endcsname\relax
  \providecommand{\doi}[1]{doi: #1}\else
  \providecommand{\doi}{doi: \begingroup \urlstyle{rm}\Url}\fi

\bibitem[OpenAI(2023)]{chatgpt}
OpenAI.
\newblock Chatgpt (mar 14 version).
\newblock \emph{Large Language Model}, 2023.

\bibitem[Saharia et~al.(2024)Saharia, Chan, Saxena, Lit, Whang, Denton,
  Ghasemipour, Ayan, Mahdavi, Gontijo-Lopes, Salimans, Ho, Fleet, and
  Norouzi]{imagen}
Chitwan Saharia, William Chan, Saurabh Saxena, Lala Lit, Jay Whang, Emily
  Denton, Seyed Kamyar~Seyed Ghasemipour, Burcu~Karagol Ayan, S.~Sara Mahdavi,
  Raphael Gontijo-Lopes, Tim Salimans, Jonathan Ho, David~J Fleet, and Mohammad
  Norouzi.
\newblock Photorealistic text-to-image diffusion models with deep language
  understanding.
\newblock In \emph{Proceedings of the 36th International Conference on Neural
  Information Processing Systems}, NIPS '22, 2024.

\bibitem[Ho et~al.(2022)Ho, Chan, Saharia, Whang, Gao, Gritsenko, Kingma,
  Poole, Norouzi, Fleet, and Salimans]{ho2022imagen}
Jonathan Ho, William Chan, Chitwan Saharia, Jay Whang, Ruiqi Gao, Alexey
  Gritsenko, Diederik~P. Kingma, Ben Poole, Mohammad Norouzi, David~J. Fleet,
  and Tim Salimans.
\newblock Imagen video: High definition video generation with diffusion models,
  2022.

\bibitem[Bar-Tal et~al.(2024)Bar-Tal, Chefer, Tov, Herrmann, Paiss, Zada,
  Ephrat, Hur, Liu, Raj, Li, Rubinstein, Michaeli, Wang, Sun, Dekel, and
  Mosseri]{bartal2024lumiere}
Omer Bar-Tal, Hila Chefer, Omer Tov, Charles Herrmann, Roni Paiss, Shiran Zada,
  Ariel Ephrat, Junhwa Hur, Guanghui Liu, Amit Raj, Yuanzhen Li, Michael
  Rubinstein, Tomer Michaeli, Oliver Wang, Deqing Sun, Tali Dekel, and Inbar
  Mosseri.
\newblock Lumiere: A space-time diffusion model for video generation, 2024.

\bibitem[Brooks et~al.(2024)Brooks, Peebles, Holmes, DePue, Guo, Jing, Schnurr,
  Taylor, Luhman, Luhman, Ng, Wang, and Ramesh]{sora}
Tim Brooks, Bill Peebles, Connor Holmes, Will DePue, Yufei Guo, Li~Jing, David
  Schnurr, Joe Taylor, Troy Luhman, Eric Luhman, Clarence Ng, Ricky Wang, and
  Aditya Ramesh.
\newblock Video generation models as world simulators.
\newblock 2024.
\newblock URL
  \url{https://openai.com/research/video-generation-models-as-world-simulators}.

\bibitem[Clegg()]{meta}
Nick Clegg.
\newblock Labeling ai-generated images on facebook, instagram and threads.
\newblock
  \url{https://about.fb.com/news/2024/02/labeling-ai-generated-images-on-facebook-instagram-and-threads/}.
\newblock President, Global Affairs, Meta.

\bibitem[tik()]{tiktok}
New labels for disclosing ai-generated content.
\newblock
  \url{https://newsroom.tiktok.com/en-us/new-labels-for-disclosing-ai-generated-content}.

\bibitem[Ben-Porat and Tennenholtz(2017)]{ben2017shapley}
Omer Ben-Porat and Moshe Tennenholtz.
\newblock Shapley facility location games.
\newblock In \emph{International Conference on Web and Internet Economics},
  pages 58--73. Springer, 2017.

\bibitem[Ben-Porat and Tennenholtz(2018)]{ben2018game}
Omer Ben-Porat and Moshe Tennenholtz.
\newblock A game-theoretic approach to recommendation systems with strategic
  content providers.
\newblock \emph{Advances in Neural Information Processing Systems}, 31, 2018.

\bibitem[Yao et~al.(2023{\natexlab{a}})Yao, Li, Sankararaman, Liao, Zhu, Wang,
  Wang, and Xu]{yao2023rethinking}
Fan Yao, Chuanhao Li, Karthik~Abinav Sankararaman, Yiming Liao, Yan Zhu, Qifan
  Wang, Hongning Wang, and Haifeng Xu.
\newblock Rethinking incentives in recommender systems: Are monotone rewards
  always beneficial?
\newblock \emph{Advances in Neural Information Processing Systems},
  2023{\natexlab{a}}.

\bibitem[Yao et~al.(2023{\natexlab{b}})Yao, Li, Nekipelov, Wang, and
  Xu]{yao2023bad}
Fan Yao, Chuanhao Li, Denis Nekipelov, Hongning Wang, and Haifeng Xu.
\newblock How bad is top-$ k $ recommendation under competing content creators?
\newblock In \emph{International Conference on Machine Learning}, pages
  39674--39701. PMLR, 2023{\natexlab{b}}.

\bibitem[Zhu et~al.(2023)Zhu, Karimireddy, Jiao, and Jordan]{zhu2023online}
Banghua Zhu, Sai~Praneeth Karimireddy, Jiantao Jiao, and Michael~I Jordan.
\newblock Online learning in a creator economy.
\newblock \emph{arXiv preprint arXiv:2305.11381}, 2023.

\bibitem[Hu et~al.(2023)Hu, Jagadeesan, Jordan, and
  Steinhard]{hu2023incentivizing}
Xinyan Hu, Meena Jagadeesan, Michael~I Jordan, and Jacob Steinhard.
\newblock Incentivizing high-quality content in online recommender systems.
\newblock \emph{arXiv preprint arXiv:2306.07479}, 2023.

\bibitem[Jagadeesan et~al.(2023)Jagadeesan, Garg, and
  Steinhardt]{jagadeesan2023supply}
Meena Jagadeesan, Nikhil Garg, and Jacob Steinhardt.
\newblock Supply-side equilibria in recommender systems.
\newblock \emph{Advances in Neural Information Processing Systems}, 2023.

\bibitem[Hron et~al.(2022)Hron, Krauth, Jordan, Kilbertus, and
  Dean]{hron2022modeling}
Jiri Hron, Karl Krauth, Michael Jordan, Niki Kilbertus, and Sarah Dean.
\newblock Modeling content creator incentives on algorithm-curated platforms.
\newblock In \emph{The Eleventh International Conference on Learning
  Representations}, 2022.

\bibitem[Yao et~al.(2024)Yao, Li, Nekipelov, Wang, and Xu]{yao2024human}
Fan Yao, Chuanhao Li, Denis Nekipelov, Hongning Wang, and Haifeng Xu.
\newblock Human vs. generative ai in content creation competition: Symbiosis or
  conflict?
\newblock \emph{arXiv preprint arXiv:2402.15467}, 2024.

\bibitem[Immorlica et~al.(2024)Immorlica, Jagadeesan, and
  Lucier]{immorlica2024clickbait}
Nicole Immorlica, Meena Jagadeesan, and Brendan Lucier.
\newblock Clickbait vs. quality: How engagement-based optimization shapes the
  content landscape in online platforms.
\newblock In \emph{Proceedings of the ACM Web Conference 2024}, 2024.

\bibitem[Mladenov et~al.(2020)Mladenov, Creager, Ben-Porat, Swersky, Zemel, and
  Boutilier]{mladenov2020optimizing}
Martin Mladenov, Elliot Creager, Omer Ben-Porat, Kevin Swersky, Richard Zemel,
  and Craig Boutilier.
\newblock Optimizing long-term social welfare in recommender systems: A
  constrained matching approach.
\newblock In \emph{International Conference on Machine Learning}, pages
  6987--6998. PMLR, 2020.

\bibitem[Taitler and Ben-Porat(2024)]{taitler2024braess}
Boaz Taitler and Omer Ben-Porat.
\newblock Braess's paradox of generative ai.
\newblock \emph{arXiv preprint arXiv:2409.05506}, 2024.

\bibitem[Raghavan(2024)]{raghavan2024competition}
Manish Raghavan.
\newblock Competition and diversity in generative ai.
\newblock \emph{arXiv preprint arXiv:2412.08610}, 2024.

\bibitem[Kaplan et~al.(2020)Kaplan, McCandlish, Henighan, Brown, Chess, Child,
  Gray, Radford, Wu, and Amodei]{kaplan2020scaling}
Jared Kaplan, Sam McCandlish, Tom Henighan, Tom~B Brown, Benjamin Chess, Rewon
  Child, Scott Gray, Alec Radford, Jeffrey Wu, and Dario Amodei.
\newblock Scaling laws for neural language models.
\newblock \emph{arXiv preprint arXiv:2001.08361}, 2020.

\bibitem[Gao et~al.(2023)Gao, Schulman, and Hilton]{gao2023scaling}
Leo Gao, John Schulman, and Jacob Hilton.
\newblock Scaling laws for reward model overoptimization.
\newblock In \emph{International Conference on Machine Learning}, pages
  10835--10866. PMLR, 2023.

\bibitem[Lewis et~al.(2020)Lewis, Perez, Piktus, Petroni, Karpukhin, Goyal,
  K{\"u}ttler, Lewis, Yih, Rockt{\"a}schel, et~al.]{lewis2020retrieval}
Patrick Lewis, Ethan Perez, Aleksandra Piktus, Fabio Petroni, Vladimir
  Karpukhin, Naman Goyal, Heinrich K{\"u}ttler, Mike Lewis, Wen-tau Yih, Tim
  Rockt{\"a}schel, et~al.
\newblock Retrieval-augmented generation for knowledge-intensive nlp tasks.
\newblock \emph{Advances in Neural Information Processing Systems},
  33:\penalty0 9459--9474, 2020.

\bibitem[Siriwardhana et~al.(2023)Siriwardhana, Weerasekera, Wen, Kaluarachchi,
  Rana, and Nanayakkara]{siriwardhana2023improving}
Shamane Siriwardhana, Rivindu Weerasekera, Elliott Wen, Tharindu Kaluarachchi,
  Rajib Rana, and Suranga Nanayakkara.
\newblock Improving the domain adaptation of retrieval augmented generation
  (rag) models for open domain question answering.
\newblock \emph{Transactions of the Association for Computational Linguistics},
  11:\penalty0 1--17, 2023.

\bibitem[Bradley and Terry(1952)]{bradley1952rank}
Ralph~Allan Bradley and Milton~E Terry.
\newblock Rank analysis of incomplete block designs: I. the method of paired
  comparisons.
\newblock \emph{Biometrika}, 39\penalty0 (3/4):\penalty0 324--345, 1952.

\bibitem[Davidson and Farquhar(1976)]{davidson1976bibliography}
Roger~R Davidson and Peter~H Farquhar.
\newblock A bibliography on the method of paired comparisons.
\newblock \emph{Biometrics}, pages 241--252, 1976.

\bibitem[Yue et~al.(2012)Yue, Broder, Kleinberg, and Joachims]{yue2012k}
Yisong Yue, Josef Broder, Robert Kleinberg, and Thorsten Joachims.
\newblock The k-armed dueling bandits problem.
\newblock \emph{Journal of Computer and System Sciences}, 78\penalty0
  (5):\penalty0 1538--1556, 2012.

\bibitem[Koller and Megiddo(1992)]{koller1992complexity}
Daphne Koller and Nimrod Megiddo.
\newblock The complexity of two-person zero-sum games in extensive form.
\newblock \emph{Games and economic behavior}, 4\penalty0 (4):\penalty0
  528--552, 1992.

\bibitem[Letchford and Conitzer(2010)]{letchford2010computing}
Joshua Letchford and Vincent Conitzer.
\newblock Computing optimal strategies to commit to in extensive-form games.
\newblock In \emph{Proceedings of the 11th ACM conference on Electronic
  commerce}, pages 83--92, 2010.

\bibitem[Hansen et~al.(2007)Hansen, Miltersen, and
  S{\o}rensen]{hansen2007finding}
Kristoffer~Arnsfelt Hansen, Peter~Bro Miltersen, and Troels~Bjerre S{\o}rensen.
\newblock Finding equilibria in games of no chance.
\newblock In \emph{Computing and Combinatorics: 13th Annual International
  Conference, COCOON 2007, Banff, Canada, July 16-19, 2007. Proceedings 13},
  pages 274--284. Springer, 2007.

\bibitem[Calabro et~al.(2008)Calabro, Impagliazzo, Kabanets, and
  Paturi]{calabro2008complexity}
Chris Calabro, Russell Impagliazzo, Valentine Kabanets, and Ramamohan Paturi.
\newblock The complexity of unique k-sat: An isolation lemma for k-cnfs.
\newblock \emph{Journal of Computer and System Sciences}, 74\penalty0
  (3):\penalty0 386--393, 2008.

\bibitem[Roughgarden(2020)]{roughgarden2020algorithms}
Tim Roughgarden.
\newblock \emph{Algorithms Illuminated: Algorithms for NP-hard Problems}.
\newblock Soundlikeyourself Publishing, LLC, 2020.

\bibitem[Carmosino et~al.(2016)Carmosino, Gao, Impagliazzo, Mihajlin, Paturi,
  and Schneider]{carmosino2016nondeterministic}
Marco~L Carmosino, Jiawei Gao, Russell Impagliazzo, Ivan Mihajlin, Ramamohan
  Paturi, and Stefan Schneider.
\newblock Nondeterministic extensions of the strong exponential time hypothesis
  and consequences for non-reducibility.
\newblock In \emph{Proceedings of the 2016 ACM Conference on Innovations in
  Theoretical Computer Science}, pages 261--270, 2016.

\bibitem[Braverman et~al.(2014)Braverman, Ko, and
  Weinstein]{braverman2014approximating}
Mark Braverman, Young~Kun Ko, and Omri Weinstein.
\newblock Approximating the best nash equilibrium in no (log n)-time breaks the
  exponential time hypothesis.
\newblock In \emph{Proceedings of the twenty-sixth annual ACM-SIAM symposium on
  Discrete algorithms}, pages 970--982. SIAM, 2014.

\bibitem[Dell et~al.(2014)Dell, Husfeldt, Marx, Taslaman, and
  Wahlen]{dell2014exponential}
Holger Dell, Thore Husfeldt, D{\'a}niel Marx, Nina Taslaman, and Martin Wahlen.
\newblock Exponential time complexity of the permanent and the tutte
  polynomial.
\newblock \emph{ACM Transactions on Algorithms (TALG)}, 10\penalty0
  (4):\penalty0 1--32, 2014.

\bibitem[Williams(2018)]{williams2018some}
Virginia~Vassilevska Williams.
\newblock On some fine-grained questions in algorithms and complexity.
\newblock In \emph{Proceedings of the international congress of mathematicians:
  Rio de janeiro 2018}, pages 3447--3487. World Scientific, 2018.

\bibitem[Basu et~al.(2019)Basu, Sen, Sanghavi, and
  Shakkottai]{basu2019blocking}
Soumya Basu, Rajat Sen, Sujay Sanghavi, and Sanjay Shakkottai.
\newblock Blocking bandits.
\newblock \emph{Advances in Neural Information Processing Systems}, 32, 2019.

\bibitem[Immorlica et~al.(2018)Immorlica, Mao, Slivkins, and
  Wu]{immorlica2018incentivizing}
Nicole Immorlica, Jieming Mao, Aleksandrs Slivkins, and Zhiwei~Steven Wu.
\newblock Incentivizing exploration with selective data disclosure.
\newblock \emph{arXiv preprint arXiv:1811.06026}, 2018.

\bibitem[Dai et~al.(2024)Dai, Flanigan, Jagadeesan, Haghtalab, and
  Podimata]{dai2024can}
Jessica Dai, Bailey Flanigan, Meena Jagadeesan, Nika Haghtalab, and Chara
  Podimata.
\newblock Can probabilistic feedback drive user impacts in online platforms?
\newblock In \emph{International Conference on Artificial Intelligence and
  Statistics}, pages 2512--2520. PMLR, 2024.

\bibitem[Auer et~al.(2002)Auer, Cesa-Bianchi, and Fischer]{auer2002finite}
Peter Auer, Nicolo Cesa-Bianchi, and Paul Fischer.
\newblock Finite-time analysis of the multiarmed bandit problem.
\newblock \emph{Machine learning}, 47\penalty0 (2):\penalty0 235--256, 2002.

\bibitem[Abbasi-Yadkori et~al.(2011)Abbasi-Yadkori, P{\'a}l, and
  Szepesv{\'a}ri]{abbasi2011improved}
Yasin Abbasi-Yadkori, D{\'a}vid P{\'a}l, and Csaba Szepesv{\'a}ri.
\newblock Improved algorithms for linear stochastic bandits.
\newblock \emph{Advances in neural information processing systems}, 24, 2011.

\bibitem[Sedgewick and Wayne(2011)]{sedgewick2011algorithms}
Robert Sedgewick and Kevin Wayne.
\newblock \emph{Algorithms}.
\newblock Addison-wesley professional, 2011.

\end{thebibliography}
